%% file: main.tex
\documentclass[conference]{IEEEtran}

\IEEEoverridecommandlockouts

\usepackage{array}
\usepackage{balance}
\usepackage{cite}
\usepackage{amsmath,amssymb,amsfonts,amsthm}
\usepackage{graphicx}
\usepackage{caption}
\usepackage{subcaption}
\usepackage{textcomp}
\usepackage{xcolor}
\usepackage{url}
\graphicspath{{./figures/}}

\usepackage{enumitem}
\usepackage{multirow}

\usepackage[linesnumbered,ruled,vlined]{algorithm2e}
\SetAlFnt{\small}
\SetKwInOut{Input}{Input}
\SetKwInOut{Output}{Output}
\SetKwProg{Fn}{Function}{}{}

\newtheorem{theorem}{Theorem}
\newtheorem{corollary}{Corollary}
\newtheorem{lemma}{Lemma}

\theoremstyle{definition}
\newtheorem{definition}{Definition}
\newtheorem{example}{Example}

\begin{document}

\title{A Fully Dynamic Algorithm for k-Regret Minimizing Sets\\
\thanks{\IEEEauthorrefmark{1}This research was mostly done when Yanhao Wang worked
at National University of Singapore (NUS).}
}

\author{
\IEEEauthorblockN{Yanhao Wang\IEEEauthorrefmark{1},
                  Yuchen Li\textsuperscript{\IEEEauthorrefmark{2},\IEEEauthorrefmark{5}},
                  Raymond Chi-Wing Wong\IEEEauthorrefmark{3},
                  Kian-Lee Tan\IEEEauthorrefmark{4}}
\IEEEauthorblockA{
  \textit{\IEEEauthorrefmark{1}University of Helsinki\quad
  \IEEEauthorrefmark{2}Singapore Management University\quad
  \IEEEauthorrefmark{3}The Hong Kong University of Science and Technology}\\
  \textit{\IEEEauthorrefmark{4}National University of Singapore\quad
  \IEEEauthorrefmark{5}Alibaba-Zhejiang University Joint Research Institute of Frontier Technologies}\\
  \IEEEauthorrefmark{1}yanhao.wang@helsinki.fi\quad
  \IEEEauthorrefmark{2}yuchenli@smu.edu.sg\quad
  \IEEEauthorrefmark{3}raywong@cse.ust.hk\quad
  \IEEEauthorrefmark{4}tankl@comp.nus.edu.sg
}
}

\maketitle

\begin{abstract}
  Selecting a small set of representatives from a large database is important in many
  applications such as multi-criteria decision making, web search, and recommendation.
  The $k$-regret minimizing set ($k$-RMS) problem was recently proposed for
  representative tuple discovery.
  Specifically, for a large database $P$ of tuples with multiple numerical attributes,
  the $k$-RMS problem returns a size-$r$ subset $Q$ of $P$ such that,
  for any possible ranking function, the score of the top-ranked tuple in $Q$
  is not much worse than the score of the $k$\textsuperscript{th}-ranked tuple in $P$.
  Although the $k$-RMS problem has been extensively studied in the literature,
  existing methods are designed for the static setting and
  cannot maintain the result efficiently when the database is updated.
  To address this issue, we propose the first fully-dynamic algorithm for the $k$-RMS problem
  that can efficiently provide the up-to-date result w.r.t.~any tuple insertion and deletion
  in the database with a provable guarantee.
  Experimental results on several real-world and synthetic datasets demonstrate that
  our algorithm runs up to four orders of magnitude faster than existing $k$-RMS algorithms
  while providing results of nearly equal quality.
\end{abstract}


\input{section/01-introduction}
\input{section/03-definition}
\input{section/04-framework}
\input{section/05-experiments}
\input{section/02-related-work}
\input{section/06-conclusion}

\section*{Acknowledgment}
We thank anonymous reviewers for their helpful comments to improve this research.
Yanhao Wang has been supported by the MLDB project of Academy of Finland (decision number: 322046).
The research of Raymond Chi-Wing Wong is supported by HKRGC GRF 16219816.

\bibliographystyle{IEEEtran}
\bibliography{references}

\input{section/appendix}

\end{document}

%% file: section/01-introduction.tex
\section{Introduction}
\label{sec:intro}

In many real-world applications such as multi-criteria decision
making~\cite{DBLP:journals/pvldb/NanongkaiSLLX10},
web search~\cite{DBLP:conf/edbt/StoyanovichYJ18},
recommendation~\cite{DBLP:journals/tkde/WangLT19,DBLP:conf/recsys/LiuMLY11},
and data description~\cite{DBLP:conf/kdd/WangLT19},
it is crucial to find a succinct representative subset from a large database
to meet the requirements of various users.
For example, when a user queries for a hotel on a website (e.g.,~booking.com and~expedia.com),
she/he will receive thousands of available options as results.
The website would like to display the best choices in the first few pages
from which almost all users could find what they are most interested in.
A common method is to rank all results using a utility function
that denotes a user's preference on different attributes
(e.g.,~\emph{price}, \emph{rating}, and \emph{distance to destination} for hotels)
and only present the top-$k$ tuples with the highest scores
according to this function to the user.
However, due to the wide diversity of user preferences,
it is infeasible to represent the preferences of all users by any single utility function.
Therefore, to select a set of highly representative tuples,
it is necessary to take all (possible) user preferences into account.

A well-established approach to finding such representatives from databases
is the \emph{skyline} operator~\cite{DBLP:conf/icde/BorzsonyiKS01}
based on the concept of \emph{domination}: a tuple $p$ dominates a tuple $q$ iff
$p$ is as good as $q$ on all attributes and strictly better than $q$ on at least one attribute.
For a given database, a skyline query returns its \emph{Pareto-optimal} subset which
consists of all tuples that are not dominated by any tuple.
It is guaranteed that any user can find her/his best choice from the skyline
because the top-ranked result according to any monotone function must not be dominated.
Unfortunately, although skyline queries are effective for representing low-dimensional databases,
their result sizes cannot be controlled
and increase rapidly as the dimensionality (i.e.,~number of attributes in a tuple) grows,
particularly so for databases with anti-correlated attributes.

Recently, the $k$-regret minimizing set ($k$-RMS)
problem~\cite{DBLP:journals/pvldb/NanongkaiSLLX10,DBLP:conf/icde/PengW14,DBLP:journals/pvldb/ChesterTVW14,DBLP:conf/wea/AgarwalKSS17,DBLP:conf/icdt/CaoLWWWWZ17,DBLP:conf/sigmod/AsudehN0D17,DBLP:conf/alenex/KumarS18,DBLP:conf/sigmod/XieW0LL18}
was proposed to alleviate the deficiency of skyline queries.
Specifically, given a database $P$ of tuples with $d$ numeric attributes,
the $k$-RMS problem aims to find a subset $Q \subseteq P$ such that, for any possible utility function,
the top-$1$ tuple in $Q$ can approximate the top-$k$ tuples in $P$ within a small error.
Here, the \emph{maximum $k$-regret ratio}~\cite{DBLP:journals/pvldb/ChesterTVW14} ($\mathtt{mrr}_k$)
is used to measure how well $Q$ can represent $P$.
For a utility function $f$, the $k$-regret ratio ($\mathtt{rr}_k$) of $Q$ over $P$
is defined to be $0$ if the top-$1$ tuple in $Q$ is among the top-$k$ tuples in $P$ w.r.t.~$f$,
or otherwise, to be one minus the ratio between the score of the top-$1$ tuple in $Q$
and the score of the $k$\textsuperscript{th}-ranked tuple in $P$ w.r.t.~$f$.
Then, the \emph{maximum $k$-regret ratio} ($\mathtt{mrr}_k$) is
defined by the maximum of $\mathtt{rr}_k$ over a class of (possibly infinite) utility functions.
Given a positive integer $r$,
a $k$-RMS on a database $P$ returns a subset $Q \subseteq P$ of size $r$ to minimize $\mathtt{mrr}_k$.
As an illustrative example, the website could run a $k$-RMS on all available hotels to
pick a set of $r$ candidates from which all users can find at least one close to her/his top-$k$ choices.

The $k$-RMS problem has been extensively studied recently.
Theoretically, it is
NP-hard~\cite{DBLP:journals/pvldb/ChesterTVW14,DBLP:conf/wea/AgarwalKSS17,DBLP:conf/icdt/CaoLWWWWZ17}
on any database with $d \geq 3$.
In general, we categorize existing $k$-RMS algorithms into three types.
The first type is dynamic programming
algorithms~\cite{DBLP:conf/sigmod/AsudehN0D17,DBLP:journals/pvldb/ChesterTVW14,DBLP:conf/icdt/CaoLWWWWZ17}
for $k$-RMS on two-dimensional data.
Although they can provide optimal solutions when $d = 2$,
they are not suitable for higher dimensions due to the NP-hardness of $k$-RMS.
The second type is the greedy
heuristic~\cite{DBLP:journals/pvldb/NanongkaiSLLX10,DBLP:conf/icde/PengW14,DBLP:journals/pvldb/ChesterTVW14},
which always adds a tuple that maximally reduces $\mathtt{mrr}_k$ at each iteration.
Although these algorithms can provide high-quality results empirically,
they have no theoretical guarantee and suffer from low efficiency on high-dimensional data.
The third type is to transform $k$-RMS into another problem
such as $\varepsilon$-kernel~\cite{DBLP:conf/sigmod/XieW0LL18,DBLP:conf/alenex/KumarS18,DBLP:conf/wea/AgarwalKSS17,DBLP:conf/icdt/CaoLWWWWZ17},
discretized matrix min-max~\cite{DBLP:conf/sigmod/AsudehN0D17},
hitting set~\cite{DBLP:conf/alenex/KumarS18,DBLP:conf/wea/AgarwalKSS17},
and $k$-\textsc{medoid} clustering~\cite{Shetiya2019},
and then to utilize existing solutions of the transformed problem for $k$-RMS computation.
Although these algorithms are more efficient than greedy heuristics while having theoretical bounds,
they are designed for the static setting and
cannot process database updates efficiently.
Typically, most of them precompute the skyline as an input to compute the result of
$k$-RMS. Once a tuple insertion or deletion triggers any change in the skyline,
they are unable to maintain the result without re-running from scratch.
Hence, existing $k$-RMS algorithms become very inefficient in highly dynamic environments
where tuples in the databases are frequently inserted and deleted.
However, dynamic databases are very common in real-world scenarios, especially for online services.
For example, in a hotel booking system,
the \emph{prices} and \emph{availabilities} of rooms are frequently changed over time.
As another example, in an IoT network,
a large number of sensors may often connect or disconnect with the server.
Moreover, sensors also update their statistics regularly.
Therefore, it is essential to address the problem of maintaining an up-to-date
result for $k$-RMS when the database is frequently updated.

In this paper, we propose the first fully-dynamic $k$-RMS algorithm that
can efficiently maintain the result of $k$-RMS w.r.t.~any tuple
insertion and deletion in the database with both theoretical guarantee and good empirical performance.
Our main contributions are summarized as follows.
\begin{itemize}
  \item We formally define the notion of \emph{maximum $k$-regret ratio}
  and the \emph{$k$-regret minimizing set} ($k$-RMS) problem in a fully-dynamic setting.
  (Section~\ref{sec:definition})
  \item We propose the first fully-dynamic algorithm called FD-RMS
  to maintain the $k$-RMS result over tuple insertions and deletions in a database.
  Our basic idea is to transform \emph{fully-dynamic $k$-RMS} into
  a \emph{dynamic set cover} problem. Specifically, FD-RMS computes the (approximate)
  top-$k$ tuples for a set of randomly sampled utility functions
  and builds a set system based on the top-$k$ results.
  Then, the $k$-RMS result can be retrieved from an approximate solution
  for \emph{set cover} on the set system.
  Furthermore, we devise a novel algorithm for dynamic set cover by introducing
  the notion of \emph{stable solution},
  which is used to efficiently update the $k$-RMS result
  whenever an insertion or deletion triggers some changes
  in top-$k$ results as well as the set system.
  We also provide detailed theoretical analyses of FD-RMS.
  (Section~\ref{sec:framework})
  \item We conduct extensive experiments on several real-world and synthetic datasets
  to evaluate the performance of FD-RMS.
  The results show that FD-RMS achieves up to four orders of magnitude speedup over existing
  $k$-RMS algorithms while providing results of nearly equal quality in a fully dynamic setting.
  (Section~\ref{sec:exp})
\end{itemize}


%% file: section/03-definition.tex
\section{Preliminaries}
\label{sec:definition}

In this section, we formally define the problem we study in this paper. We first introduce
the notion of maximum $k$-regret ratio. Then, we formulate the
$k$-regret minimizing set ($k$-RMS) problem in a fully dynamic setting.
Finally, we present the challenges of solving fully-dynamic $k$-RMS.

\begin{figure}
  \centering
  \includegraphics[width=.32\textwidth]{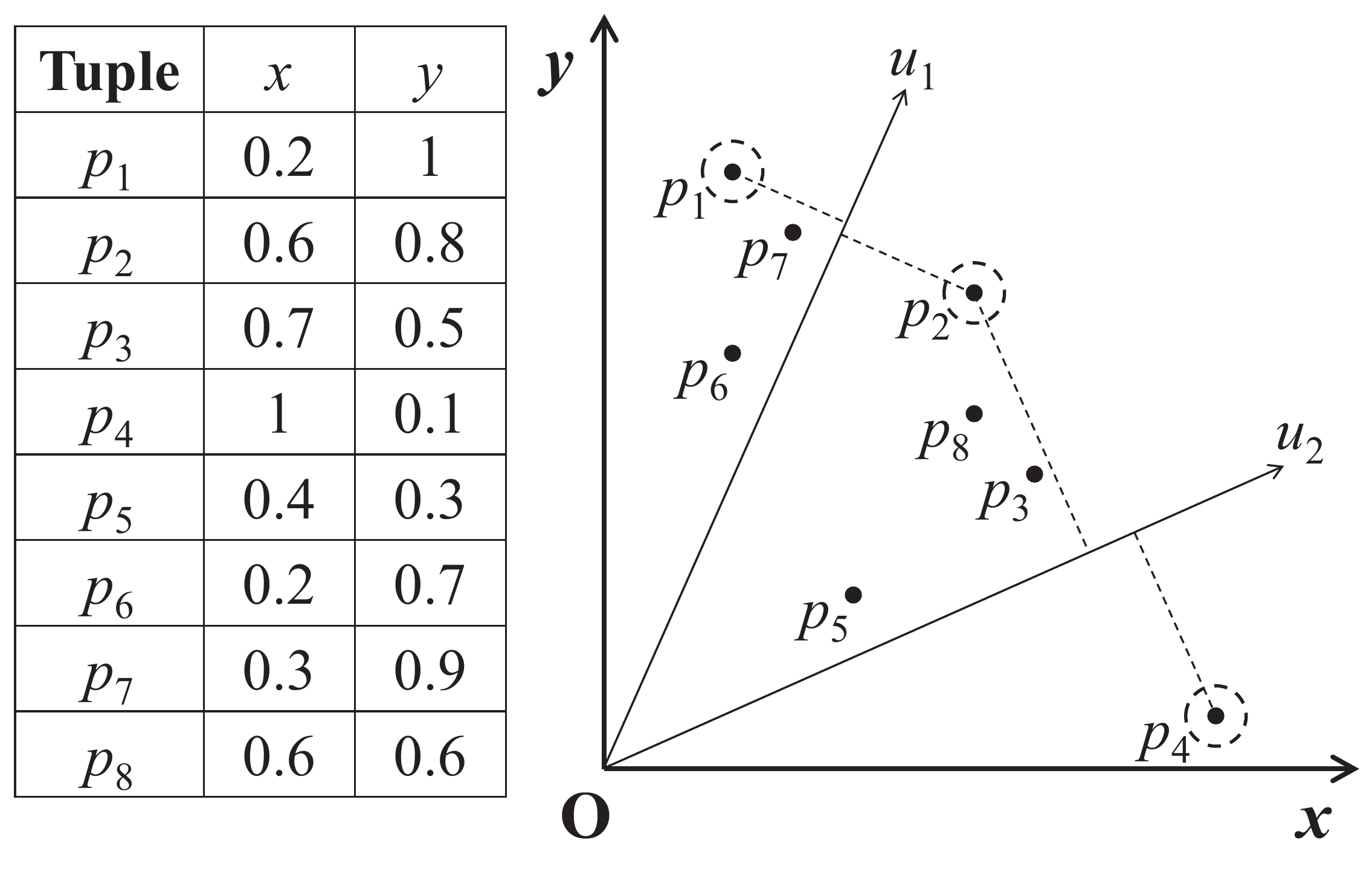}
  \caption{A two-dimensional database of 8 tuples.}
  \label{fig:database}
\end{figure}

\subsection{Maximum K-Regret Ratio}

Let us consider a database $P$ where each tuple $p \in P$
has $d$ nonnegative numerical attributes $p[1],\ldots,p[d]$
and is represented as a point in the nonnegative orthant $\mathbb{R}_+^d$.
A user's preference is denoted by a utility function
$f:\mathbb{R}_+^d \rightarrow \mathbb{R}^+$
that assigns a positive score $f(p)$ to each tuple $p$.
Following~\cite{DBLP:journals/pvldb/NanongkaiSLLX10,DBLP:conf/icde/PengW14,DBLP:journals/pvldb/ChesterTVW14,DBLP:conf/sigmod/AsudehN0D17,DBLP:conf/sigmod/XieW0LL18},
we restrict the class of utility functions to \emph{linear functions}.
A function $f$ is linear if and only if there exists
a $d$-dimensional vector $u=(u[1],\ldots,u[d]) \in \mathbb{R}_+^d$
such that $f(p)=\langle u,p \rangle =\sum_{i=1}^{d} u[i] \cdot p[i]$
for any $p \in \mathbb{R}_+^d$.
W.l.o.g., we assume the range of values on each dimension is scaled to $[0,1]$ and
any utility vector is normalized to be a unit\footnote{The normalization does not affect our results
because the maximum $k$-regret ratio is \emph{scale-invariant}~\cite{DBLP:journals/pvldb/NanongkaiSLLX10}.},
i.e., $\| u \|=1$.
Intuitively, the class of linear functions corresponds to the nonnegative orthant of
$d$-dimensional unit sphere
$ \mathbb{U}=\{ u \in \mathbb{R}_+^d \,:\, \lVert u \rVert=1 \} $.

We use $\varphi_j(u,P)$ to denote the tuple $p \in P$ with the $j$\textsuperscript{th}-largest
score w.r.t.~vector $u$ and $\omega_j(u,P)$ to denote its score.
Note that multiple tuples may have the same score w.r.t.~$u$
and any consistent rule can be adopted to break ties.
For brevity, we drop the subscript $j$ from the above notations when $j=1$,
i.e., $\varphi(u,P)=\arg\max_{p \in P} \langle u,p \rangle$
and $\omega(u,P)=\max_{p \in P} \langle u,p \rangle$.
The top-$k$ tuples in $P$ w.r.t.~$u$ is represented as
$\Phi_{k}(u,P)=\{\varphi_{j}(u,P) : 1 \leq j \leq k\}$.
Given a real number $\varepsilon \in (0,1)$,
the $\varepsilon$-approximate top-$k$ tuples in $P$ w.r.t.~$u$ is denoted
as $\Phi_{k,\varepsilon}(u,P)=\{p \in P \,:\, \langle u,p \rangle \geq (1-\varepsilon)\cdot\omega_{k}(u,P)\}$,
i.e., the set of tuples whose scores are at least $(1-\varepsilon)\cdot\omega_{k}(u,P)$.

For a subset $Q \subseteq P$ and an integer $k \geq 1$,
we define the \emph{$k$-regret ratio} of $Q$ over $P$ for a utility vector $u$ by
$\mathtt{rr}_k(u,Q)=\max\big(0,1-\frac{\omega(u,Q)}{\omega_{k}(u,P)}\big)$,
i.e., the relative loss of replacing the $k$\textsuperscript{th}-ranked tuple in $P$ by
the top-ranked tuple in $Q$.
Since it is required to consider the preferences of all possible users, our goal is to find
a subset whose $k$-regret ratio is small for an arbitrary utility vector.
Therefore, we define the \emph{maximum $k$-regret ratio} of $Q$ over $P$
by $\mathtt{mrr}_k(Q)=\max_{u \in \mathbb{U}}\mathtt{rr}_k(u,Q)$.
Intuitively, $\mathtt{mrr}_k(Q)$ measures how well the top-ranked tuple of $Q$
approximates the $k$\textsuperscript{th}-ranked tuple of $P$ in the worst case.
For a real number $\varepsilon \in (0,1)$,
$Q$ is said to be a $(k,\varepsilon)$-regret set of $P$ iff $\mathtt{mrr}_k(Q) \leq \varepsilon$,
or equivalently, $\varphi(u,Q) \in \Phi_{k,\varepsilon}(u,P)$ for any $u \in \mathbb{U}$.
By definition, it holds that $\mathtt{mrr}_k(Q) \in [0,1]$.

\begin{example}
  Fig.~\ref{fig:database} illustrates a database $P$ in $\mathbb{R}^2_+$
  with $8$ tuples $\{p_1,\ldots,p_8\}$.
  For utility vectors $u_1=(0.42,0.91)$ and $u_2=(0.91,0.42)$,
  their top-$2$ results are $\Phi_2(u_1,P)=\{p_1,p_2\}$ and $\Phi_2(u_2,P)=\{p_2,p_4\}$, respectively.
  Given a subset $Q_1=\{p_3,p_4\}$ of $P$, $\mathtt{rr}_2(u_1,Q_1)=1-\frac{0.749}{0.98} \approx 0.236$
  as $\omega(u_1,Q_1)=\langle u_1,p_3 \rangle=0.749$ and $\omega_2(u_1,P)=\langle u_1,p_2 \rangle=0.98$.
  Furthermore, $\mathtt{mrr}_2(Q_1) \approx 0.444$
  because $\mathtt{rr}_2(u,Q_1)$ is the maximum when $u=(0.0,1.0)$
  with $\mathtt{rr}_2(u,Q_1) = 1-\frac{5}{9} \approx 0.444$.
  Finally, $Q_2=\{p_1,p_2,p_4\}$ is a $(2,0)$-regret set of $P$ since $\mathtt{mrr}_2(Q_2)=0$.
\end{example}

\subsection{K-Regret Minimizing Set}

Based on the notion of \emph{maximum $k$-regret ratio},
we can formally define the $k$-\emph{regret minimizing set}
($k$-RMS) problem in the following.
\begin{definition}[$k$-Regret Minimizing Set]\label{def:k-rms}
  Given a database $P \subset \mathbb{R}_+^d$ and a size constraint $r \in \mathbb{Z}^+$
  ($r \geq d$), the $k$-regret minimizing set ($k$-RMS) problem
  returns a subset $Q^{*} \subseteq P$ of at most $r$ tuples
  with the smallest \emph{maximum $k$-regret ratio},
  i.e., $Q^{*}=\arg\min_{Q \subseteq P \,:\, |Q| \leq r} \mathtt{mrr}_k(Q)$.
\end{definition}
For any given $k$ and $r$, we denote the $k$-RMS problem by $\mathtt{RMS}(k,r)$
and the \emph{maximum $k$-regret ratio} of the optimal result $Q^{*}$ for $\mathtt{RMS}(k,r)$
by $\varepsilon^{*}_{k,r}$.
In particular, the $r$-regret query studied
in~\cite{DBLP:conf/icde/PengW14,DBLP:conf/sigmod/AsudehN0D17,DBLP:conf/sigmod/XieW0LL18,
DBLP:journals/pvldb/NanongkaiSLLX10}
is a special case of our $k$-RMS problem when $k=1$, i.e., $1$-RMS.

\begin{example}
  Let us continue with the example in Fig.~\ref{fig:database}.
  For a query $\mathtt{RMS}(2,2)$ on $P$, we have $Q^*=\{p_1,p_4\}$
  with $\varepsilon^*_{2,2}=\mathtt{mrr}_2(Q^*) \approx 0.05$
  because $\{p_1,p_4\}$ has the smallest maximum $2$-regret ratio
  among all size-$2$ subsets of $P$.
\end{example}

In this paper, we focus on the fully-dynamic $k$-RMS problem.
We consider an initial database $P_0$ and a (possibly countably infinite) sequence of operations
$\Delta = \langle \Delta_1,\Delta_2,\ldots \rangle$.
At each timestamp $t$ ($t \in \mathbb{Z}^+$), the database is updated from $P_{t-1}$
to $P_t$ by performing an operation $\Delta_t$ of one of the following two types:
\begin{itemize}
  \item \textbf{Tuple insertion} $\Delta_t=\langle p,+ \rangle$:
  add a new tuple $p$ to $P_{t-1}$, i.e., $P_t \gets P_{t-1}\cup\{p\}$;
  \item \textbf{Tuple deletion} $\Delta_t=\langle p,- \rangle$:
  delete an existing tuple $p$ from $P_{t-1}$, i.e., $P_t \gets P_{t-1}\setminus\{p\}$.
\end{itemize}
Note that the update of a tuple can be processed by a deletion followed by an insertion,
and thus is not discussed separately in this paper.
Given an initial database $P_0$, a sequence of operations $\Delta$,
and a query $\mathtt{RMS}(k,r)$,
we aim to keep track of the result $Q^{*}_t$ for $\mathtt{RMS}(k,r)$ on $P_t$ at any time $t$.

Fully-dynamic $k$-RMS faces two challenges.
First, the $k$-RMS problem is
NP-hard~\cite{DBLP:conf/icdt/CaoLWWWWZ17,DBLP:journals/pvldb/ChesterTVW14,DBLP:conf/wea/AgarwalKSS17}
for any $d \geq 3$. Thus, the optimal solution of $k$-RMS is intractable
for any database with three or more attributes unless P=NP in both static and dynamic settings.
Hence, we will focus on maintaining an approximate result of $k$-RMS in this paper.
Second, existing $k$-RMS algorithms can only work in the static setting.
They must recompute the result from scratch once an operation
triggers any update in the skyline
(Note that since the result of $k$-RMS is a subset of the
skyline~\cite{DBLP:journals/pvldb/NanongkaiSLLX10,DBLP:journals/pvldb/ChesterTVW14},
it remains unchanged for any operation on non-skyline tuples).
However, frequent recomputation leads to significant overhead
and causes low efficiency on highly dynamic databases.
Therefore, we will propose a novel method for fully-dynamic $k$-RMS
that can maintain a high-quality result for $\mathtt{RMS}(k,r)$
on a database w.r.t.~any tuple insertion and deletion efficiently.

%% file: section/04-framework.tex
\section{The FD-RMS Algorithm}
\label{sec:framework}

\begin{figure}
  \centering
  \includegraphics[width=.475\textwidth]{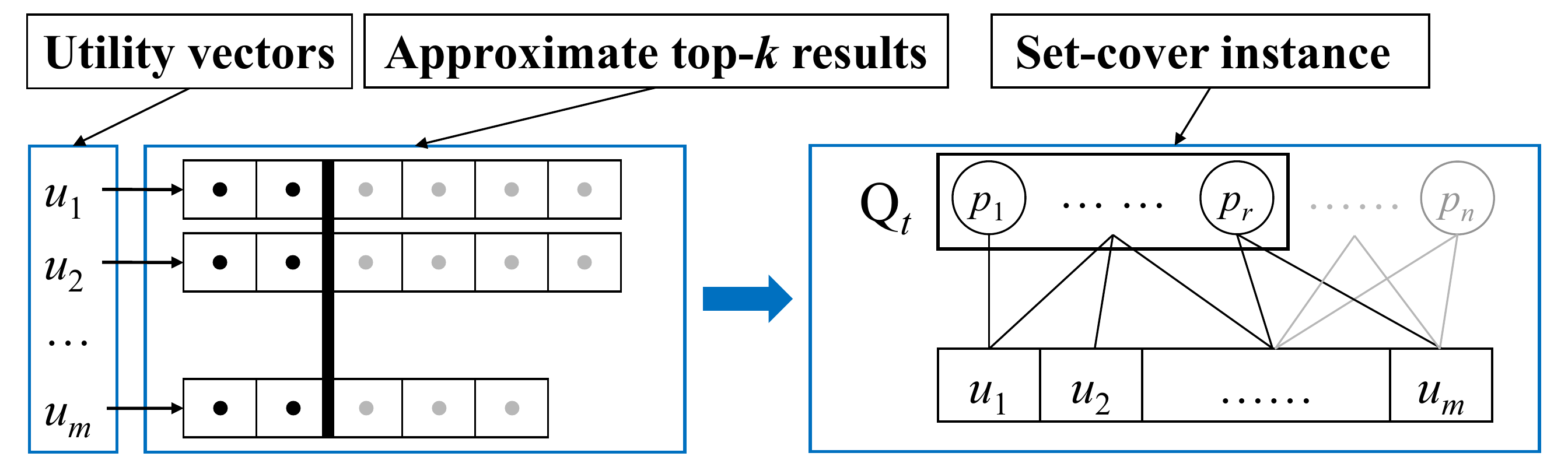}
  \caption{An illustration of FD-RMS}
  \label{fig:framework}
\end{figure}

In this section, we present our FD-RMS algorithm for $k$-RMS in a fully dynamic setting.
The general framework of FD-RMS is illustrated in Fig.~\ref{fig:framework}.
The basic idea is to transform \emph{fully-dynamic $k$-RMS}
to a \emph{dynamic set cover} problem.
Let us consider how to compute the result of $\mathtt{RMS}(k,r)$ on database $P_t$.
First of all, we draw a set of $m$ random utility vectors $\{u_1,\ldots,u_m\}$
from $ \mathbb{U} $ and maintain the $\varepsilon$-approximate top-$k$ result
of each $u_i$ ($i \in [1,m]$) on $P_t$, i.e., $ \Phi_{k,\varepsilon}(u_i,P_t) $.
Note that $\varepsilon$ should be given as an input parameter of FD-RMS
and we will discuss how to specify its value at the end of Section~\ref{sec:framework}.
Then, we construct a set system $ \Sigma = (\mathcal{U},\mathcal{S}) $
based on the approximate top-$k$ results,
where the universe $\mathcal{U}=\{u_1,\ldots,u_m\}$
and the collection $\mathcal{S}$ consists of $n_t$ sets ($n_t=|P_t|$)
each of which corresponds to one tuple in $P_t$.
Specifically, for each tuple $p \in P_t$, we define $S(p)$ as a set of
utility vectors for which $p$ is an $\varepsilon$-approximate top-$k$ result on $P_t$.
Or formally, $S(p) = \{u \in \mathcal{U} : p \in \Phi_{k,\varepsilon}(u,P_t)\}$
and $\mathcal{S}=\{ S(p) : p \in P_t \}$.
After that, we compute a result $Q_t$ for $\mathtt{RMS}(k,r)$ on $P_t$ 
using an (approximate) solution for set cover on $\Sigma$.
Let $\mathcal{C} \subseteq \mathcal{S}$ be a set-cover solution of $\Sigma$,
i.e., $\bigcup_{S(p) \in \mathcal{C}} S(p)=\mathcal{U}$.
We use the set $Q_t$ of tuples corresponding to $\mathcal{C}$,
i.e., $Q_t=\{p \in P_t : S(p) \in \mathcal{C}\}$,
as the result of $\mathtt{RMS}(k,r)$ on $P_t$.
Given the above framework, there are still two challenges of updating the result of $k$-RMS
in a fully dynamic setting. Firstly, because the size of $Q_t$ is restricted to $r$,
it is necessary to always keep an appropriate value of $m$ over time
so that $|\mathcal{C}| \leq r$.
Secondly, the updates in approximate top-$k$ results triggered by tuple insertions
and deletions in the database
lead to the changes in the set collection $\mathcal{S}$.
Therefore, it is essential to maintain the set-cover
solution $\mathcal{C}$ over time for the changes in $\mathcal{S}$.
In fact, both challenges can be treated as a \emph{dynamic set cover} problem that keeps
a set-cover solution w.r.t.~changes in both $\mathcal{U}$ and $\mathcal{S}$.
Therefore, we will first introduce the background on \emph{dynamic set cover}
in Section~\ref{subsec:set:cover}.
After that, we will elaborate on how FD-RMS processes $k$-RMS in a fully dynamic setting
using the \emph{dynamic set cover} algorithm in Section~\ref{subsec:alg}.

\subsection{Background: Dynamic Set Cover}
\label{subsec:set:cover}

Given a set system $\Sigma=(\mathcal{U},\mathcal{S})$, the \emph{set cover} problem asks for
the smallest subset $\mathcal{C}^{*}$ of $\mathcal{S}$ whose union equals to the universe $\mathcal{U}$.
It is one of Karp's 21 NP-complete problems~\cite{DBLP:conf/coco/Karp72},
and cannot be approximated to $(1-o(1))\cdot\ln{m}$ ($m=|\mathcal{U}|$) unless P=NP~\cite{DBLP:journals/jacm/Feige98}.
A common method to find an approximate set-cover solution is the greedy algorithm.
Starting from $\mathcal{C} = \varnothing$, it always adds the set that contains the largest number of uncovered
elements in $\mathcal{U}$ to $\mathcal{C}$ at each iteration until $\bigcup_{S \in \mathcal{C}}S=\mathcal{U}$.
Theoretically, the solution $\mathcal{C}$ achieves an approximation ratio of $(1+\ln{m})$,
i.e., $|\mathcal{C}| \leq (1+\ln{m}) \cdot |\mathcal{C}^{*}|$.
But obviously, the greedy algorithm cannot dynamically update the set-cover solution
when the set system $\Sigma$ is changed.

Recently, there are some theoretical advances on covering and relevant problems
(e.g., vertex cover, maximum matching, set cover, and maximal independent set) in dynamic
settings~\cite{DBLP:journals/siamcomp/BhattacharyaHI18,DBLP:conf/stoc/GuptaK0P17,DBLP:conf/stoc/AbboudA0PS19,DBLP:conf/stacs/HjulerIPS19}.
Although these theoretical results have opened up new ways to design dynamic set cover algorithms,
they cannot be directly applied to the update procedure of FD-RMS because of two limitations.
First, existing dynamic algorithms for set cover~\cite{DBLP:conf/stoc/AbboudA0PS19,DBLP:conf/stoc/GuptaK0P17}
can only handle the update in the universe $\mathcal{U}$ but assume that the set collection $\mathcal{S}$ is not changed.
But in our scenario, the changes in top-$k$ results lead to the update of $\mathcal{S}$.
Second, due to the extremely large constants introduced in their analyses,
the solutions returned may be far away from the optima in practice.

Therefore, we devise a more practical approach to dynamic set cover
that supports any update in both $\mathcal{U}$ and $\mathcal{S}$.
Our basic idea is to introduce the notion of \emph{stability} to a set-cover solution.
Then, we prove that any stable solution is $O(\log{m})$-approximate ($m=|\mathcal{U}|$) for set cover.
Based on this result, we are able to design an algorithm to maintain a set-cover solution
w.r.t.~any change in $\Sigma$ by guaranteeing its \emph{stability}.

We first formalize the concept of \emph{stability} of a set-cover solution.
Let $\mathcal{C} \subseteq \mathcal{S}$ be a set-cover solution on $\Sigma=(\mathcal{U},\mathcal{S})$.
We define an assignment $\phi$ from each element $u \in \mathcal{U}$ to a \emph{unique} set $S \in \mathcal{C}$
that contains $u$ (or formally, $\phi: \mathcal{U} \rightarrow \mathcal{C}$).
For each set $S \in \mathcal{C}$, its cover set $\mathtt{cov}(S)$ is defined as the set of elements assigned to $S$,
i.e., $\mathtt{cov}(S) = \{u \in \mathcal{U} \,:\, \phi(u)=S\}$.
By definition, the cover sets of different sets in $\mathcal{C}$ are \emph{mutually disjoint} from each other.
Then, we can organize the sets in $\mathcal{C}$ into hierarchies according to the numbers of elements covered by them.
Specifically, we put a set $S \in \mathcal{C}$ in a higher level if it covers more elements and vice versa.
We associate each level $\mathcal{L}_j$ ($j \in \mathbb{N}$) with a range of cover
number\footnote{Here, the base $2$ may be replaced by any constant greater than 1.}
$[2^{j},2^{j+1})$.
Each set $S \in \mathcal{C}$ is assigned to a level $\mathcal{L}_j$ if $2^{j} \leq |\mathtt{cov}(S)| < 2^{j+1}$.
We use $A_j$ to denote the set of elements assigned to any set in $\mathcal{L}_j$,
i.e., $A_j = \{u \in \mathcal{U} \,:\, \phi(u) \in \mathcal{L}_j\}$.
Moreover, the notations $\mathcal{L}$ with subscripts, i.e.,
$\mathcal{L}_{>j}$ or $\mathcal{L}_{\geq j}$ and $\mathcal{L}_{<j}$ or $\mathcal{L}_{\leq j}$,
represent the sets in all the levels above (excl.~or incl.) and below $\mathcal{L}_j$ (excl.~or incl.), respectively.
The same subscripts are also used for $A$.
Based on the above notions, we formally define the \emph{stability} of a set-cover solution in Definition~\ref{def:stable}
and give its approximation ratio in Theorem~\ref{thm:approx:ratio}.

\begin{definition}[Stable Set-Cover Solution]\label{def:stable}
  A solution $\mathcal{C}$ for \emph{set cover} on $\Sigma=(\mathcal{U},\mathcal{S})$
  is stable if:
  \begin{enumerate}
    \item For each set $S \in \mathcal{L}_j$, $2^{j} \leq |\mathtt{cov}(S)| < 2^{j+1}$;
    \item For each level $\mathcal{L}_j$, there is no $S \in \mathcal{S}$ s.t. $|S \cap A_j| \geq 2^{j+1}$.
  \end{enumerate}
\end{definition}
\begin{theorem}\label{thm:approx:ratio}
  If a set-cover solution $\mathcal{C}$ is stable,
  then it satisfies that $|\mathcal{C}| \leq O(\log m) \cdot |\mathcal{C}^{*}|$.
\end{theorem}
\begin{proof}
  Let $\mathtt{OPT}=|\mathcal{C}^{*}|$, $\rho^{*}=\frac{m}{\mathtt{OPT}}$,
  and $j^{*}$ be the level index such that $2^{j^{*}} \leq \rho^{*} < 2^{j^{*}+1}$.
  According to Condition (1) of Definition~\ref{def:stable},
  we have $|\mathtt{cov}(S)| \geq 2^{j^{*}}$ for any $S \in \mathcal{L}_{\geq j^{*}}$.
  Thus, it holds that
  $|\mathcal{L}_{\geq j^{*}}| \leq \frac{|A_{\geq j^{*}}|}{2^{j^{*}}} \leq \frac{m}{2^{j^{*}}}
  \leq \frac{\rho^{*}}{2^{j^{*}}} \cdot \mathtt{OPT} \leq 2 \cdot \mathtt{OPT}$.
  For some level $\mathcal{L}_j$ with $j < j^{*}$,
  according to Condition (2) of Definition~\ref{def:stable}, any $S \in \mathcal{S}$
  covers at most $2^{j+1}$ elements in $A_j$.
  Hence, $\mathcal{S}^{*}$ needs at least $\frac{|A_j|}{2^{j+1}}$ sets
  to cover $A_j$, i.e., $\mathtt{OPT} \geq \frac{|A_j|}{2^{j+1}}$.
  Since $|\mathtt{cov}(S)| \geq 2^{j}$ for each $S \in \mathcal{L}_j$, it holds that
  $|\mathcal{L}_j| \leq \frac{|A_j|}{2^{j}} \leq 2 \cdot \mathtt{OPT}$.
  As $1 \leq |\mathtt{cov}(S)| \leq m$, the range of level index is $[0,\log_{2}{m}]$.
  Thus, the number of levels below $\mathcal{L}_{j^{*}}$ is at most $\log_{2}{m}$.
  To sum up, we prove that
  \begin{equation*}\textstyle
    |\mathcal{C}| = |\mathcal{L}_{\geq j^{*}}| + |\mathcal{L}_{< j^{*}}|
    \leq (2 + 2 \log_{2}{m}) \cdot \mathtt{OPT}
  \end{equation*}
  and conclude the proof.
\end{proof}

\begin{algorithm}[t]
  \caption{\textsc{Dynamic Set Cover}}\label{alg:set:cover}
  \Input{Set system $\Sigma$, set operation $\sigma$}
  \Output{Stable set-cover solution $\mathcal{C}$}
  \tcc{compute an initial solution $\mathcal{C}$ on $\Sigma$}
  $ \mathcal{C} \gets$ \textsc{Greedy}$(\Sigma)$\;\label{ln:cover:init}
  \tcc{update $\mathcal{C}$ for $\sigma=(u,S,\pm)$ or $(u,\mathcal{U},\pm)$}
  \uIf{$\sigma=(u,S,-)$ and $u \in \mathtt{cov}(S)$\label{ln:sigma:s}}{
    $ \mathtt{cov}(S) \gets \mathtt{cov}(S) \setminus \{u\} $\;
    $ \mathtt{cov}(S^+) \gets \mathtt{cov}(S^+) \cup \{u\} $ for $S^+ \in \mathcal{S}$
    s.t.~$u \in S^+$\;
    \textsc{ReLevel}$(S)$ and \textsc{ReLevel}$(S^+)$\;\label{ln:move:1}
  }
  \uElseIf{$\sigma=(u,\mathcal{U},+)$}{
    $ \mathtt{cov}(S^+) \gets \mathtt{cov}(S^+) \cup \{u\} $ for $S^+ \in \mathcal{S}$
    s.t.~$u \in S^+$\;
    \textsc{ReLevel}$(S^+)$\;\label{ln:move:2}
  }
  \ElseIf{$\sigma=(u,\mathcal{U},-)$}{
    $ \mathtt{cov}(S^-) \gets \mathtt{cov}(S^-) \setminus \{u\} $ if $u \in \mathtt{cov}(S^-)$\;
    \textsc{ReLevel}$(S^-)$\;\label{ln:move:3}
  }
  \textsc{Stabilize}$(\mathcal{C})$\;\label{ln:sigma:t}
  \BlankLine
  \Fn{\textnormal{\textsc{Greedy}$(\Sigma)$}\label{ln:greedy:s}}{
    $I \gets \mathcal{U}$, $\mathcal{L}_j \gets \varnothing$ for every $j \geq 0$\;
    \While{$I \neq \varnothing$}{
      $S^* \gets \arg\max_{S \in \mathcal{S}} |I \cap S|$,
      $\mathtt{cov}(S^*) \gets I \cap S^*$\;\label{ln:cover:max}
      Add $S^*$ to $\mathcal{L}_j$ s.t.~$2^{j} \leq |\mathtt{cov}(S^*)| < 2^{j+1}$\;
      $I \gets I \setminus \mathtt{cov}(S^*)$\;
    }
    \Return{$\mathcal{C} \gets \bigcup_{j \geq 0} \mathcal{L}_j$}\;\label{ln:greedy:t}
  }
  \Fn{\textnormal{\textsc{ReLevel}$(S)$}\label{ln:move:s}}{
    \uIf{$\mathtt{cov}(S) = \varnothing$}{
      $\mathcal{C} \gets \mathcal{C} \setminus \{S\}$\;
    }
    \Else{
      Let $\mathcal{L}_{j}$ be the current level of $S$\;
      \If{$|\mathtt{cov}(S)| < 2^{j}$ or $|\mathtt{cov}(S)| \geq 2^{j+1}$}{
        Let $j'$ be the index s.t. $2^{j'} \leq |\mathtt{cov}(S)| < 2^{j'+1}$\;
        Move $S$ from $\mathcal{L}_{j}$ to $\mathcal{L}_{j'}$\;
      }  
    }
    \label{ln:move:t}
  }
  \Fn{\textnormal{\textsc{Stabilize}$(\mathcal{C})$}\label{ln:stable:s}}{
    \While{$\exists S \in \mathcal{S}$ and $\mathcal{L}_j$ s.t. $|S \cap A_{j}|
    \geq 2^{j+1}$}{
      $\mathtt{cov}(S) \gets \mathtt{cov}(S) \cup (S \cap A_{j})$, \textsc{ReLevel}($S$)\;
      \While{$\exists S' \in \mathcal{C} : \mathtt{cov}(S) \cap \mathtt{cov}(S') \neq
      \varnothing$}{
        $\mathtt{cov}(S') \gets \mathtt{cov}(S') \setminus \mathtt{cov}(S)$,
        \textsc{ReLevel}($S'$)\;
      }
      \label{ln:stable:t}
    }
  }
\end{algorithm}
We then describe our method for \emph{dynamic set cover} in Algorithm~\ref{alg:set:cover}.
First of all, we use \textsc{Greedy} to initialize a set-cover solution $\mathcal{C}$
on $\Sigma$ (Line~\ref{ln:cover:init}).
As shown in Lines~\ref{ln:greedy:s}--\ref{ln:greedy:t},
\textsc{Greedy} follows the classic greedy algorithm for set cover, and
the only difference is that all the sets in $\mathcal{C}$ are assigned to different levels
according to the sizes of their cover sets.
Then, the procedure of updating $\mathcal{C}$ for set operation $\sigma$ is shown
in Lines~\ref{ln:sigma:s}--\ref{ln:sigma:t}.
Our method supports four types of set operations to update $\Sigma$ as follows:
$\sigma=(u,S,\pm)$, i.e., to add/remove an element $u$ to/from a set $S \in \mathcal{S}$;
$\sigma=(u,\mathcal{U},\pm)$,
i.e., to add/remove an element $u$ to/from the universe $\mathcal{U}$.
We identify three cases in which the assignment of $u$ must be changed
for $\sigma$. When $\sigma=(u,S,-)$ and $\phi(u)=S$, it will reassign
$u$ to another set containing $u$; For $\sigma=(u,\mathcal{U},\pm)$,
it will add or delete the assignment of $u$ accordingly.
After that, for each set with some change in its cover set,
it calls \textsc{ReLevel} (e.g., Lines~\ref{ln:move:1},~\ref{ln:move:2}, and~\ref{ln:move:3})
to check whether the set should be moved to a new level based on the updated size
of its cover set.
The detailed procedure of \textsc{ReLevel} is given
in Lines~\ref{ln:move:s}--\ref{ln:move:t}.
Finally, \textsc{Stabilize} (Line~\ref{ln:sigma:t}) is always called
for every $\sigma$ to guarantee the stability of $\mathcal{C}$ since
$\mathcal{C}$ may become unstable due to the changes in $\Sigma$ and $\phi(u)$.
The procedure of stabilization is presented in Lines~\ref{ln:stable:s}--\ref{ln:stable:t}.
It finds all sets that violate
Condition~(2) of Definition~\ref{def:stable} and adjust $\mathcal{C}$
for these sets until no set should be adjusted anymore.

\noindent\textbf{Theoretical Analysis:}
Next, we will analyze Algorithm~\ref{alg:set:cover} theoretically.
We first show that a set-cover solution returned by \textsc{Greedy} is stable.
Then, we prove that \textsc{Stabilize} converges to a stable solution in finite steps.

\begin{lemma}\label{lm:stability}
  The solution $\mathcal{C}$ returned by \textnormal{\textsc{Greedy}} is stable.
\end{lemma}
\begin{proof}
  First of all, it is obvious that each set $S \in \mathcal{C}$ is assigned to the correct level
  according to the size of its cover set and Condition (1) of Definition~\ref{def:stable} is satisfied.
  Then, we sort the sets in $\mathcal{C}$ as $S^*_1,\ldots,S^*_{|\mathcal{C}|}$
  by the order in which they are added.
  Let $S^*_i$ be the set s.t.~$|\mathtt{cov}(S^*_i)| < 2^{j+1}$
  and $|\mathtt{cov}(S^*_{i'})| \geq 2^{j+1}$ for any $i'<i$,
  i.e., $S^*_i$ is the first set added to level $\mathcal{L}_j$.
  We have $|I \cap S^*_i| = |\mathtt{cov}(S^*_i)| < 2^{j+1}$
  where $I$ is the set of uncovered elements before $S^*_i$ is added to $\mathcal{C}$.
  If there were a set $S \in \mathcal{S}$ such that $|S \cap A_j| > 2^{j+1}$,
  we would acquire $|I \cap S| \geq |S \cap A_j| > 2^{j+1}$ and $|I \cap S|>|I \cap S^*_i|$,
  which contradicts with Line~\ref{ln:cover:max} of Algorithm~\ref{alg:rms:init}.
  Thus, $\mathcal{C}$ must satisfy Condition (2) of Definition~\ref{def:stable}.
  To sum up, $\mathcal{C}$ is a stable solution.
\end{proof}

\begin{lemma}\label{lm:termination}
  The procedure \textnormal{\textsc{Stabilize}} converges to a stable solution in $O(m\log{m})$ steps.
\end{lemma}
\begin{proof}
  For an iteration of the \texttt{while} loop (i.e., Lines~\ref{ln:stable:s}--\ref{ln:stable:t})
  that picks a set $S$ and a level $\mathcal{L}_j$,
  the new level $\mathcal{L}_{j'}$ of $S$ always satisfies $j' > j$.
  Accordingly, all the elements in $\mathtt{cov}(S)$ are moved
  from $A_{\leq j}$ to $A_{j'}$. At the same time, no element in $A_{\geq j'}$
  is moved to lower levels. Since each level contains at most $m$ elements ($|A_j| \leq m$),
  \textsc{Stabilize} moves at most $m$ elements across $O(\log m)$ levels.
  Therefore, it must terminate in $O(m \log m)$ steps.
  Furthermore, after termination, the set-cover solution $\mathcal{C}$ must satisfy
  both conditions in Definition~\ref{def:stable}.
  Thus, we conclude the proof.
\end{proof}

The above two lemmas can guarantee that the set-cover solution provided by
Algorithm~\ref{alg:set:cover} is always stable after any change in the set system.
In the next subsection, we will present how to use it for fully-dynamic $k$-RMS.

\subsection{Algorithmic Description}\label{subsec:alg}

Next, we will present how FD-RMS maintains the $k$-RMS result
by always keeping a \emph{stable} set-cover solution on a dynamic set system
built from the approximate top-$k$ results over tuple insertions and deletions.

\begin{algorithm}[t]
  \caption{\textsc{Initialization}}\label{alg:rms:init}
  \Input{Query $\mathtt{RMS}(k,r)$, initial database $P_0$,
  parameters $\varepsilon \in (0,1)$ and $M \in \mathbb{Z}^+$ ($M>r$)}
  \Output{Result $Q_0$ of $\mathtt{RMS}(k,r)$ on $P_0$}
  Draw $M$ vectors $\{u_i \in \mathbb{U} : i \in [1,M]\}$
  where the first $d$ are the standard basis of $\mathbb{R}^d_+$
  and the remaining are uniformly sampled from $\mathbb{U}$\;
  Compute $\Phi_{k,\varepsilon}(u_i,P_0)$ of every $u_i$ where $i \in [1,M]$\;
  \label{ln:rms:init:topk}
  $L \gets r$, $H \gets M$, $m \gets (L+H)/2$\;\label{ln:rms:search:s}
  \While{$\mathtt{true}$}{
    \ForEach{$p \in P_0$\label{ln:rms:set:s}}{
      $S(p) \gets \{u_i \,:\, i \in [1,m] \wedge p \in \Phi_{k,\varepsilon}(u_i,P_0) \}$\;
    }
    $\Sigma = (\mathcal{U},\mathcal{S})$ where $\mathcal{U}=\{u_i : i \in [1,m]\}$
    and $\mathcal{S}=\{S(p) : p \in P_0\})$\;\label{ln:rms:set:t}
    $\mathcal{C} \gets$ \textsc{Greedy}$(\Sigma)$\;
    \uIf{$|\mathcal{C}| < r$}{
      $L \gets m+1$, $m \gets (L+H)/2$\;
    }
    \uElseIf{$|\mathcal{C}| > r$}{
      $H \gets m-1$, $m \gets (L+H)/2$\;
    }
    \ElseIf{$|\mathcal{C}| = r$ or $m = M$}{
      \textbf{break}\;\label{ln:rms:search:t}
    }
  }
  \Return{$Q_0 \gets \{ p \in P_0 \,:\, S(p) \in \mathcal{C} \}$}\;
  \label{ln:rms:init:return}
\end{algorithm}

\noindent\textbf{Initialization:}
We first present how FD-RMS computes an initial result $Q_0$ for
$\mathtt{RMS}(k,r)$ on $P_0$ from scratch in Algorithm~\ref{alg:rms:init}.
There are two parameters in FD-RMS:
the approximation factor of top-$k$ results $\varepsilon$
and the upper bound of sample size $M$. The lower bound of sample size is set to $r$
because we can always find a set-cover solution of size equal to the size of the universe
(i.e., $m$ in FD-RMS).
First of all, it draws $M$ utility vectors $\{u_1,\ldots,u_M\}$, where the first $d$ vectors
are the standard basis of $\mathbb{R}^d_+$
and the remaining are uniformly sampled from $\mathbb{U}$,
and computes the $\varepsilon$-approximate top-$k$ result of each vector.
Subsequently, it finds an appropriate $m \in [r,M]$
so that the size of the set-cover solution on the set system $\Sigma$ built on
$\mathcal{U} = \{u_1,\ldots,u_m\}$ is exactly $r$.
The detailed procedure is as presented in Lines~\ref{ln:rms:search:s}--\ref{ln:rms:search:t}.
Specifically, it performs a binary search on range $[r,M]$ to determine the value of $m$.
For a given $m$, it first constructs a set system $\Sigma$ according to
Lines~\ref{ln:rms:set:s}--\ref{ln:rms:set:t}.
Next, it runs \textsc{Greedy} in Algorithm~\ref{alg:set:cover}
to compute a set-cover solution $\mathcal{C}$ on $\Sigma$.
After that, if $|\mathcal{C}| \neq r$ and $m < M$,
it will refresh the value of $m$ and rerun the above procedures;
Otherwise, $m$ is determined and the current set-cover solution $\mathcal{C}$
will be used to compute $Q_0$ for $\mathtt{RMS}(k,r)$.
Finally, it returns all the tuples whose corresponding sets are included in $\mathcal{C}$
as the result $Q_0$ for $\mathtt{RMS}(k,r)$ on $P_0$ (Line~\ref{ln:rms:init:return}).

\begin{algorithm}[t]
  \caption{\textsc{Update}}\label{alg:rms:update}
  \Input{Query $\mathtt{RMS}(k,r)$, database $P_{t-1}$, operation $\Delta_t$,
  set-cover solution $\mathcal{C}$}
  \Output{Result $Q_t$ for $\mathtt{RMS}(k,r)$ on $P_t$}
  Update $P_{t-1}$ to $P_t$ w.r.t.~$\Delta_t$\;\label{ln:update:database}
  \For{$i \gets 1,\ldots,M$}{
    Update $\Phi_{k,\varepsilon}(u_i,P_{t-1})$ to $\Phi_{k,\varepsilon}(u_i,P_{t})$ w.r.t.~$\Delta_t$\;\label{ln:update:topk}
  }
  Maintain $\Sigma$ based on $\Phi_{k,\varepsilon}(u_i,P_{t})$\;
  \label{ln:update:sets}
  \uIf{$\Delta_t=\langle p,+ \rangle$\label{ln:insert:s}}{
    \ForEach{$u \in S(p)$}{
      \If{$u \in \mathtt{cov}(S(p'))$ and $u \notin S(p')$}{
        Update $\mathcal{C}$ for $\sigma=(u,S(p'),-)$\;\label{ln:insert:t}
      }
    }
  }
  \ElseIf{$\Delta_t=\langle p,- \rangle$\label{ln:delete:s}}{
    Delete $S(p)$ from $\mathcal{C}$ if $S(p) \in \mathcal{C}$\;
    \ForEach{$u \in \mathtt{cov}(S(p))$}{
      Update $\mathcal{C}$ for $\sigma=(u,S(p),-)$\;\label{ln:delete:t}
    }
  }
  \If{$|\mathcal{C}| \neq r$}{
    $m, \mathcal{C} \gets$ \textsc{UpdateM}$(\Sigma)$\;
  }
  \Return{$Q_t \gets \{ p \in P_t \,:\, S(p) \in \mathcal{C} \}$}\;
\end{algorithm}

\begin{figure*}[t]
  \centering
  \begin{subfigure}{0.3\textwidth}
    \centering
    \includegraphics[width=\textwidth]{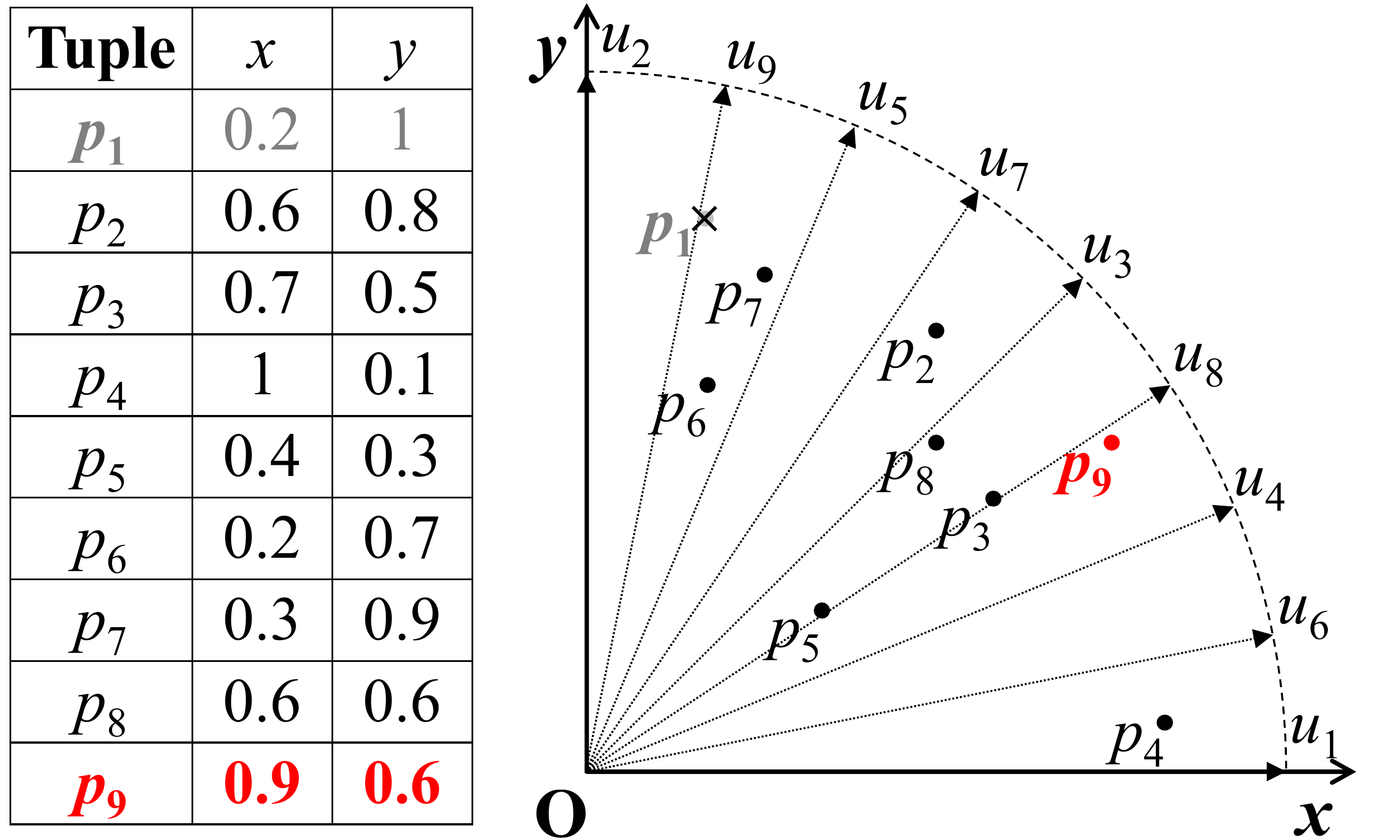}
    \caption{Dataset}\label{fig:example:database}
  \end{subfigure}
  \hfill
  \begin{subfigure}{0.225\textwidth}
    \centering
    \includegraphics[width=\textwidth]{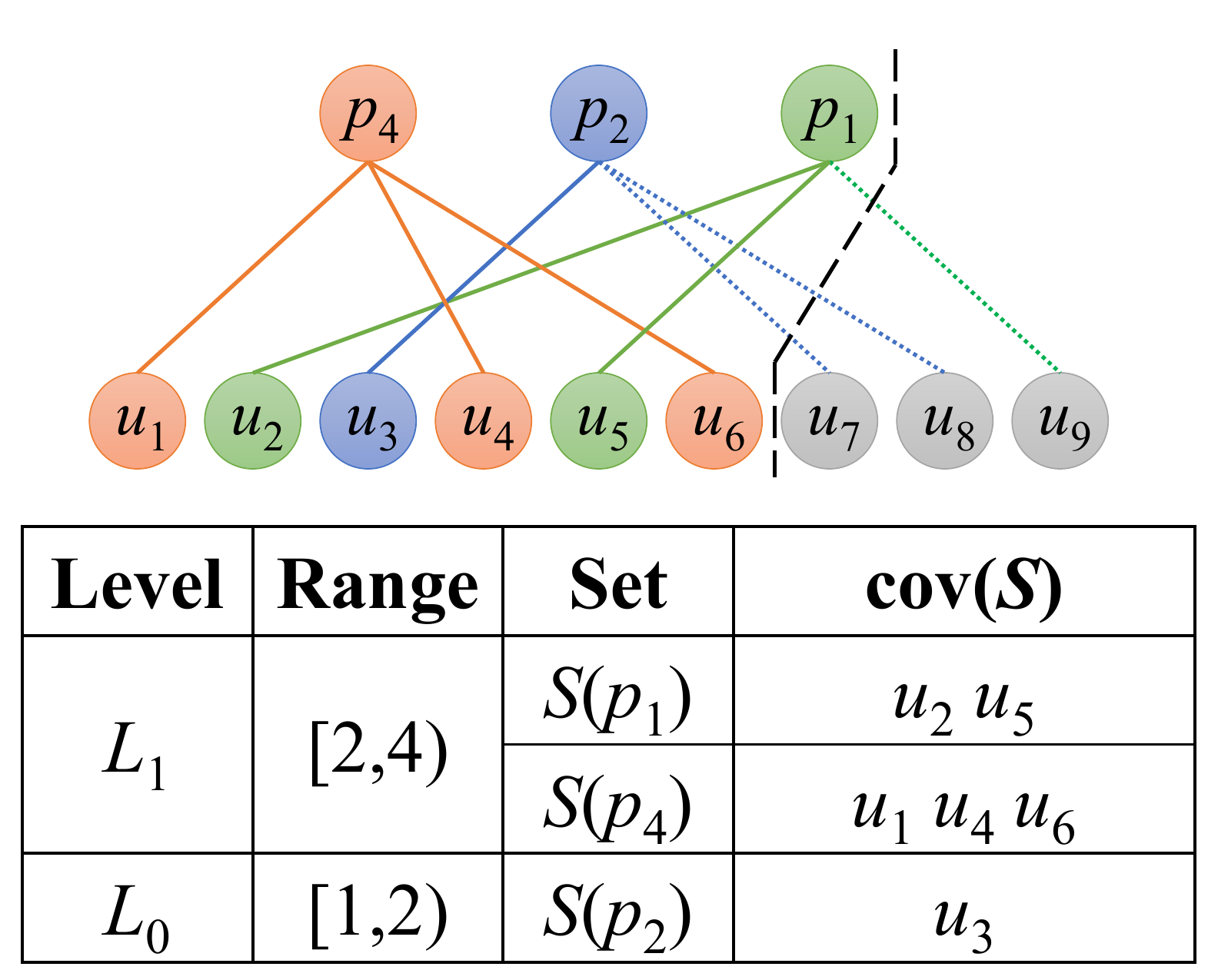}
    \caption{Initial construction}\label{fig:example:fdrms:P0}
  \end{subfigure}
  \hfill
  \begin{subfigure}{0.225\textwidth}
    \centering
    \includegraphics[width=\textwidth]{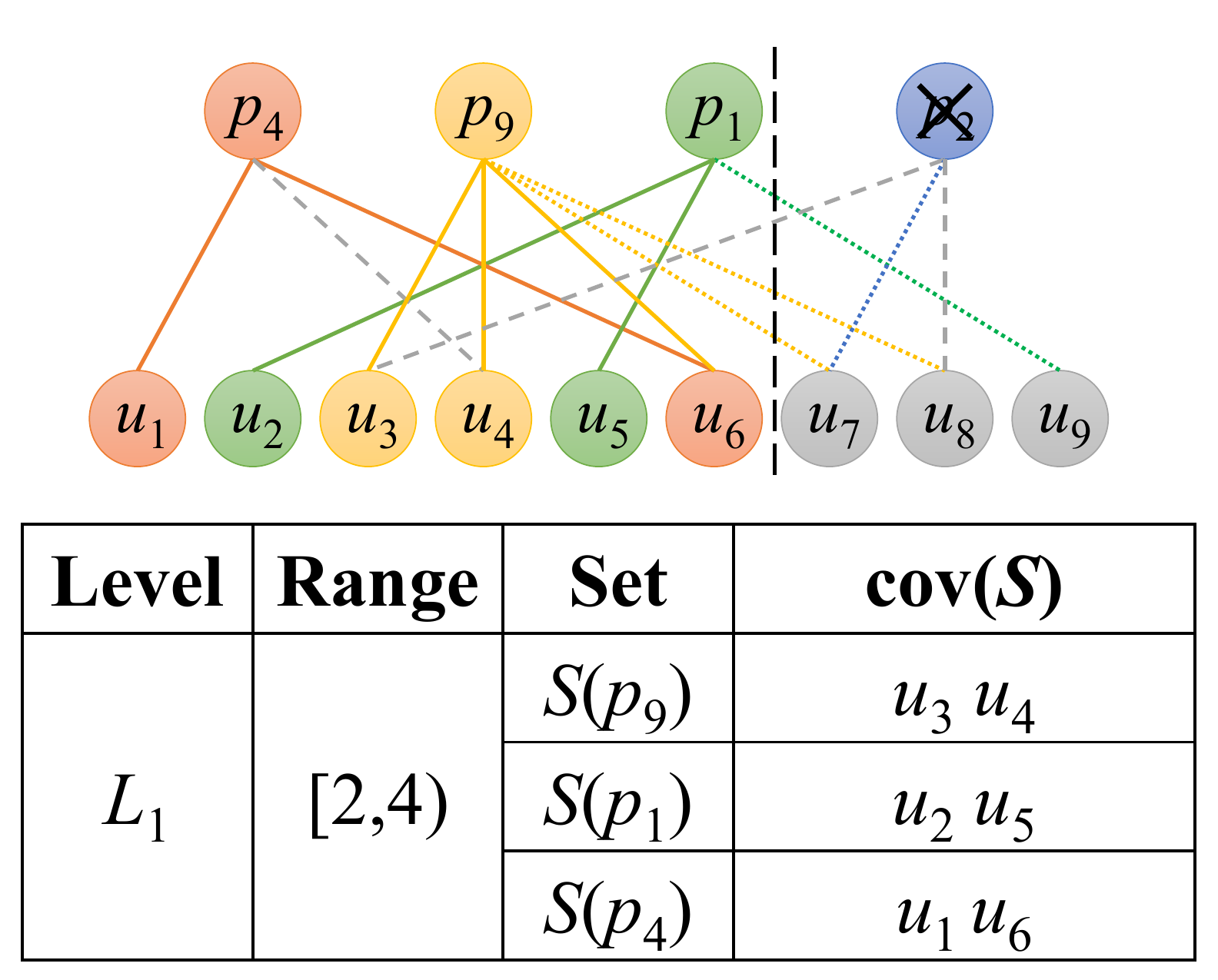}
    \caption{Add tuple $p_9$}\label{fig:example:fdrms:P1}
  \end{subfigure}
  \hfill
  \begin{subfigure}{0.225\textwidth}
    \centering
    \includegraphics[width=\textwidth]{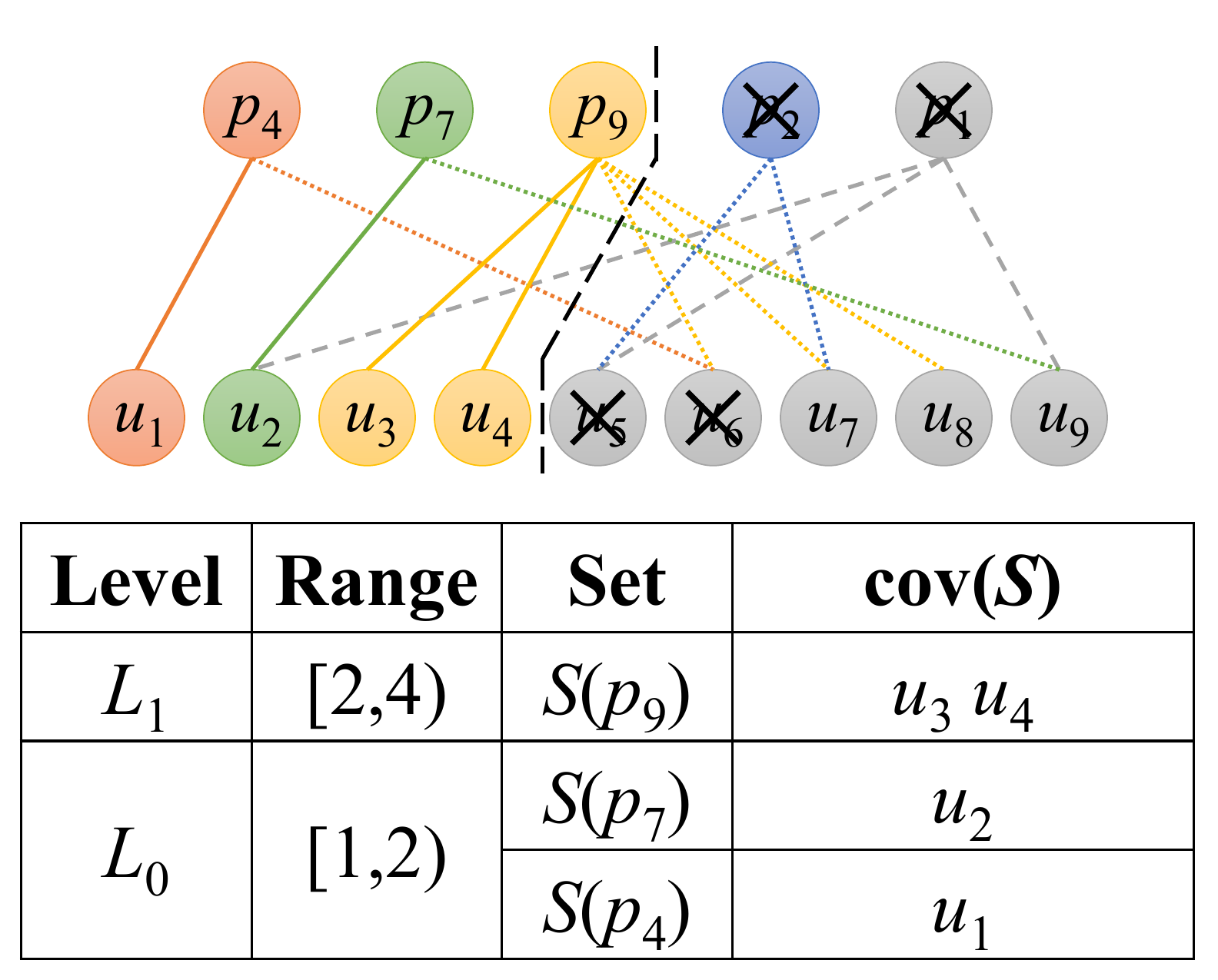}
    \caption{Delete tuple $p_1$}\label{fig:example:fdrms:P2}
  \end{subfigure}
  \caption{An example of using FD-RMS to process a $k$-RMS with $k=1$ and $r=3$}
  \label{fig:example:fdrms}
\end{figure*}

\noindent\textbf{Update:}
The procedure of updating the result of $\mathtt{RMS}(k,r)$
w.r.t.~$\Delta_t$ is shown in Algorithm~\ref{alg:rms:update}.
First, it updates the database from $P_{t-1}$ to $P_t$ and
the approximate top-$k$ result from $\Phi_{k,\varepsilon}(u_i,P_{t-1})$
to $\Phi_{k,\varepsilon}(u_i,P_t)$ for each $u_i$ w.r.t.~$\Delta_t$
(Lines~\ref{ln:update:database}--\ref{ln:update:topk}).
Then, it also maintains the set system $\Sigma$
according to the changes in approximate top-$k$ results
(Line~\ref{ln:update:sets}).
Next, it updates the set-cover solution $\mathcal{C}$
for the changes in $\Sigma$ as follows.
\begin{itemize}
\item \textbf{Insertion:} The procedure of updating $\mathcal{C}$
w.r.t.~an insertion $\Delta_t=\langle p,+ \rangle$
is presented in Lines~\ref{ln:insert:s}--\ref{ln:insert:t}.
The changes in top-$k$ results lead to two updates in $\Sigma$:
(1) the insertion of $S(p)$ to $\mathcal{S}$
and (2) a series of deletions each of which
represents a tuple $p'$ is deleted from $\Phi_{k,\varepsilon}(u,P_t)$
due to the insertion of $p$.
For each deletion, it needs to check whether $u$ is previously
assigned to $S(p')$.
If so, it will update $\mathcal{C}$ by reassigning $u$ to a new set according to
Algorithm~\ref{alg:set:cover} because $u$ has been deleted from $S(p')$.
\item \textbf{Deletion:} The procedure of updating $\mathcal{C}$
w.r.t.~a deletion $ \Delta_t=\langle p,- \rangle $
is shown in Lines~\ref{ln:delete:s}--\ref{ln:delete:t}.
In contrast to an insertion, the deletion of $p$
leads to the removal of $S(p)$ from $\mathcal{S}$
and a series of insertions.
Thus, it must delete $S(p)$ from $\mathcal{C}$.
Next, it will reassign each $u \in \mathtt{cov}(S(p))$ to a new set
according to Algorithm~\ref{alg:set:cover}.
\end{itemize}
Then, it checks whether the size of $\mathcal{C}$ is still $r$.
If not, it will update the sample size $m$ and the universe $\mathcal{U}$ so that
the set-cover solution $\mathcal{C}$ consists of $r$ sets.
The procedure of updating $m$ and $\mathcal{U}$ as well as maintaining $\mathcal{C}$
on the updated $\mathcal{U}$ is shown in Algorithm~\ref{alg:rms:update:m}.
When $|\mathcal{C}|<r$, it will add new utility vectors from $u_{m+1}$, and so on,
to the universe and maintain $\mathcal{C}$ until $|\mathcal{C}|=r$ or $m=M$.
On the contrary, if $|\mathcal{C}| > r$,
it will drop existing utility vectors from $u_{m}$, and so on, from the universe
and maintain $\mathcal{C}$ until $|\mathcal{C}|=r$.
Finally, the updated $m$ and $\mathcal{C}$ are returned.
After all above procedures, it also returns $Q_t$
corresponding to the set-cover solution $\mathcal{C}$ on the updated $\Sigma$
as the result of $\mathtt{RMS}(k,r)$ on $P_t$.

\begin{algorithm}[t]
  \caption{\textsc{UpdateM}$(\Sigma)$}\label{alg:rms:update:m}
  \Output{Updated sample size $m$ and solution $\mathcal{C}$ on $\Sigma$}
  \uIf{$|\mathcal{C}| < r$}{
    \While{$m < M$ and $|\mathcal{C}| < r$}{
      $m \gets m+1$, $\mathcal{U} \gets \mathcal{U} \cup \{u_m\}$\;
      \ForEach{$p \in \Phi_{k,\varepsilon}(u_m,P_t)$}{
        $S(p) \gets S(p) \cup \{u_m\}$\;
      }
      Update $\mathcal{C}$ for $\sigma=(u_m,\mathcal{U},+)$\;
    }
  }
  \ElseIf{$|\mathcal{C}| > r$}{
    \While{$|\mathcal{C}| > r$}{
      $\mathcal{U} \gets \mathcal{U} \setminus \{u_m\}$\;
      \ForEach{$p \in \Phi_{k,\varepsilon}(u_m,P_t)$}{
        $S(p) \gets S(p) \setminus \{u_m\}$\;
      }
      Update $\mathcal{C}$ for $\sigma=(u_m,\mathcal{U},-)$\;
      $m \gets m-1$\;
    }
  }
  \Return{$m, \mathcal{C}$}\;
\end{algorithm}

\begin{example}
  Fig.~\ref{fig:example:fdrms} illustrates an example of using FD-RMS to process a $k$-RMS
  with $k=1$ and $r=3$. Here, we set $\varepsilon=0.002$ and $M=9$.
  In Fig.~\ref{fig:example:fdrms:P0}, we show how to compute $Q_0$ for $\mathtt{RMS}(1,3)$
  on $P_0=\{p_1,\ldots,p_8\}$. It first uses $m=(3+9)/2=6$ and runs
  \textnormal{\textsc{Greedy}} to get
  a set-cover solution $\mathcal{C}=\{S(p_1), S(p_2), S(p_4)\}$.
  Since $|\mathcal{C}|=3$, it does not change $m$ anymore and returns
  $Q_0=\{p_1,p_2,p_4\}$ for $\mathtt{RMS}(1,3)$ on $P_0$. Then,
  the result of FD-RMS after the update procedures for~$\Delta_1 = \langle p_9,+ \rangle$ as
  Algorithm~\ref{alg:rms:update} is shown in Fig.~\ref{fig:example:fdrms:P1}.
  For $\mathtt{RMS}(1,3)$ on $P_1=\{p_1,\ldots,p_9\}$,
  the result $Q_1$ is updated to $\{p_1,p_4,p_9\}$.
  Finally, after the update procedures for $ \Delta_2 = \langle p_1,- \rangle $,
  as shown in Fig.~\ref{fig:example:fdrms:P2}, $m$ is updated to $4$ and the result
  $Q_2$ for $\mathtt{RMS}(1,3)$ on $P_2$ is $\{p_4,p_7,p_9\}$.
\end{example}

\noindent\textbf{Theoretical Bound:}
The theoretical bound of FD-RMS is analyzed as follows.
First of all, we need to verify the set-cover solution $\mathcal{C}$ maintained by
Algorithms~\ref{alg:rms:init}--\ref{alg:rms:update:m} is always stable.
According to Lemma~\ref{lm:stability}, it is guaranteed that the set-cover
solution $\mathcal{C}$ returned by Algorithm~\ref{alg:rms:init} is stable.
Then, we need to show it remains stable
after the update procedures of Algorithms~\ref{alg:rms:update} and~\ref{alg:rms:update:m}.
In fact, both algorithms use Algorithm~\ref{alg:set:cover} to maintain
the set-cover solution $\mathcal{C}$. Hence, the stability of $\mathcal{C}$
can be guaranteed by Lemma~\ref{lm:termination} since \textsc{Stabilize}
is always called after every update in Algorithm~\ref{alg:set:cover}.

Next, we indicate the relationship between the result of $k$-RMS
and the \emph{set-cover} solution
and provide the bound on the maximum-$k$ regret ratio of $Q_{t}$
returned by FD-RMS on $P_t$.
\begin{theorem}\label{thm:rms:ratio}
  The result $Q_t$ returned by FD-RMS is
  a $\big(k,O(\varepsilon^{*}_{k,r'}+\delta)\big)$-regret set of $P_t$ with high probability
  where $r'=O(\frac{r}{\log{m}})$ and $\delta=O(m^{-\frac{1}{d}})$.
\end{theorem}
\begin{proof}
  Given a parameter $\delta > 0$, a $\delta$-net~\cite{DBLP:conf/wea/AgarwalKSS17}
  of $\mathbb{U}$ is a finite set $U \subset \mathbb{U}$ where there
  exists a vector $\overline{u}$ with $\| u-\overline{u} \| \leq \delta$ for any $u \in \mathbb{U}$.
  Since a random set of $O(\frac{1}{\delta^{d-1}}\log\frac{1}{\delta})$ vectors in $\mathbb{U}$
  is a $\delta$-net with probability at least $\frac{1}{2}$~\cite{DBLP:conf/wea/AgarwalKSS17},
  one can generate a $\delta$-net of size $O(\frac{1}{\delta^{d-1}}\log\frac{1}{\delta})$
  with high probability by random sampling from $\mathbb{U}$ in $O(1)$ repeated trials.
  
  Let $B$ be the standard basis of $\mathbb{R}^d_+$
  and $U$ be a $\delta$-net of $\mathbb{U}$
  where $B=\{u_1,\ldots,u_d\} \subset U$.
  Since $p \in P_t$ is scaled to $p[i] \leq 1$ for $i \in [1,d]$,
  we have $\| p \| \leq \sqrt{d}$.
  According to the definition of $\delta$-net,
  there exists a vector $\overline{u} \in U$
  such that $\| \overline{u} - u \| \leq \delta$ for every $u \in \mathbb{U}$.
  Hence, for any tuple $p \in P_t$,
  \begin{equation}\label{eq:1}
  | \langle \overline{u},p \rangle - \langle u,p \rangle |
  = | \langle \overline{u} - u,p \rangle |
  \leq \| \overline{u} - u \| \cdot \| p \| \leq \delta \cdot \sqrt{d}
  \end{equation}
  Moreover, as $Q_{t}$ corresponds to a set-cover solution $\mathcal{C}$
  on $\Sigma$, there exists a tuple $q \in Q_{t}$ such that
  $\langle \overline{u},q \rangle \geq (1-\varepsilon) \cdot \omega_{k}(\overline{u},P_t)$
  for any $\overline{u} \in U$.
  We first consider a basis vector $u_i \in U$ for some $i \in [1,d]$.
  We have $\omega(u_i,Q_t) \geq (1-\varepsilon) \cdot \omega_{k}(u_i,P_t)$ and
  thus $ \omega(u_i,Q_t) \geq (1-\varepsilon) \cdot c $ where
  $c = \min_{i \in [1,d]} \omega_{k}(u_i,P_t)$.
  Since $\| u \| = 1$, there must exist some $i$ with $u[i] \geq \frac{1}{\sqrt{d}}$
  for any $u \in \mathbb{U}$.
  Therefore, it holds that $\omega(u,Q_{t}) \geq \omega(u_i,Q_{t}) \cdot \frac{1}{\sqrt{d}}
  \geq (1-\varepsilon)\cdot\frac{c}{\sqrt{d}}$
  for any $u \in \mathbb{U}$.

  Next, we discuss two cases for $u \in \mathbb{U}$ separately.
  \begin{itemize}
    \item \textbf{Case 1 ($\omega_k(u,P_t) \leq \frac{c}{\sqrt{d}}$):}
    In this case, there always exists $q \in Q_{t}$ such that
    $\langle u,q \rangle \geq (1-\varepsilon)\cdot\omega_k(u,P_t)$.
    \item \textbf{Case 2 ($\omega_k(u,P_t) > \frac{c}{\sqrt{d}}$):}
    Let $\overline{u} \in U$ be the utility vector such that
    $\| \overline{u} - u \| \leq \delta$.
    Let $\Phi_{k}(u,P_t)=\{p_1,\ldots,p_k\}$ be the top-$k$ results of $u$ on $P_t$.
    According to Equation~\ref{eq:1}, we have
    $\langle \overline{u},p_i \rangle \geq \langle u,p_i \rangle 
    - \delta\cdot\sqrt{d}$ for all $i \in [1,k]$ and thus
    $\langle \overline{u},p_i \rangle \geq \omega_{k}(u,P_t)
    - \delta\cdot\sqrt{d}$. Thus, there exists $k$ tuples in $P_t$
    with scores at least $\omega_{k}(u,P_t) - \delta\cdot\sqrt{d}$ for $\overline{u}$.
    We can acquire
    $\omega_{k}(\overline{u},P_t) \geq \omega_{k}(u,P_t) - \delta\cdot\sqrt{d}$.
    Therefore, there exists $q \in Q_{t}$ such that
    \begin{align*}
    \langle u,q \rangle & \geq \langle \overline{u},q \rangle - \delta\cdot\sqrt{d}
    \geq (1-\varepsilon) \cdot \omega_k(\overline{u},P_t) - \delta\cdot\sqrt{d} \\
    & \geq (1-\varepsilon) \cdot \big(\omega_{k}(u,P_t) - \delta\cdot\sqrt{d}\big)
    - \delta\cdot\sqrt{d} \\
    & \geq \big(1-\varepsilon-\frac{(1-\varepsilon)d\delta}{c}-\frac{d\delta}{c}\big)
    \cdot \omega_{k}(u,P_t) \\
    & \geq (1-\varepsilon-\frac{2d\delta}{c})\cdot\omega_{k}(u,P_t)
    \end{align*}
  \end{itemize}
  Considering both cases, we have $\omega(u,Q_t) \geq
  (1-\varepsilon-\frac{2d\delta}{c})\cdot\omega_{k}(u,P_t)$ for any $u \in \mathbb{U}$
  and thus $\mathtt{mrr}_k(Q_{t})$ over $P_t$ is at most $\varepsilon+\frac{2d\delta}{c}$.
  In all of our experiments, the value of $c$ is always between $0.5$ and $1$,
  and thus we regard $c$ as a constant in this proof.
  Therefore, $Q_t$ is a $\big(k,O(\varepsilon+\delta)\big)$-regret set of $P_t$
  with high probability for any $c,d=O(1)$.
  Moreover, since FD-RMS uses $m$ utility vectors including $B$ to compute $Q_t$
  and $m=O(\frac{1}{\delta^{d-1}}\log\frac{1}{\delta})$,
  we can acquire $\delta=O(m^{-\frac{1}{d}})$.
  
  Finally, because any $(k,\varepsilon)$-regret set of $P_t$ corresponds to a set-cover
  solution on $\Sigma$
  (otherwise, the regret ratio is larger than $\varepsilon$ for some utility vector)
  and the size of the optimal set-cover solution on $\Sigma$ is $O(\frac{r}{\log{m}})$
  according to Theorem~\ref{thm:approx:ratio},
  the maximum $k$-regret ratio of any size-$r'$ subset of $P_t$ is at least $\varepsilon$
  where $r'=O(\frac{r}{\log{m}})$, i.e., $\varepsilon^{*}_{k,r'} \geq \varepsilon$.
  Therefore, we conclude that $Q_t$ is a $\big(k,O(\varepsilon^{*}_{k,r'}+\delta)\big)$-regret set of $P_t$
  with high probability.
\end{proof}

Finally, the upper bound of the maximum $k$-regret ratio of $Q_{t}$ returned by FD-RMS
on $P_t$ is analyzed in the following corollary derived from the result of Theorem~\ref{thm:rms:ratio}.
\begin{corollary}\label{col:rms:bound}
  It satisfies that $\mathtt{mrr}_k(Q_{t}) = O(r^{-\frac{1}{d}})$ with high probability
  if we assume $\varepsilon = O(m^{-\frac{1}{d}})$.
\end{corollary}
\begin{proof}
  As indicated in the proof of Theorem~\ref{thm:rms:ratio},
  $\mathcal{U}=\{u_1,u_2,$ $\ldots,u_{m}\}$ is a $\delta$-net of $\mathbb{U}$
  where $\delta=O(m^{-\frac{1}{d}})$ with high probability.
  Moreover, we have $\mathtt{mrr}_k(Q_{t}) = O(\varepsilon + \delta)$
  and thus $\mathtt{mrr}_k(Q_{t}) = O(m^{-\frac{1}{d}})$
  if $\varepsilon = O(m^{-\frac{1}{d}})$.
  In addition, at any time, $\mathcal{U}$ must have at least $r$ utility vectors,
  i.e., $m \geq r$.
  Thus, we have $\mathtt{mrr}_k(Q_{t}) = O(r^{-\frac{1}{d}})$
  since $m^{-\frac{1}{d}} \leq r^{-\frac{1}{d}}$ for any $d>1$
  and conclude the proof.
\end{proof}
Since $\varepsilon$ is tunable in FD-RMS, by trying different values of $\varepsilon$,
we can always find an appropriate one such that $\varepsilon = O(m^{-\frac{1}{d}})$.
Hence, from Corollary~\ref{col:rms:bound}, we show that
the upper bound of FD-RMS is slightly higher than that of
\textsc{Cube}~\cite{DBLP:journals/pvldb/NanongkaiSLLX10} and
\textsc{Sphere}~\cite{DBLP:conf/sigmod/XieW0LL18} (i.e., $O(r^{-\frac{1}{d-1}}$))
under a mild assumption.

\noindent\textbf{Complexity Analysis:}
First, we use tree-based methods to maintain
the approximate top-$k$ results for FD-RMS (see Section~\ref{subsec:impl} for details).
Here, the time complexity of each top-$k$ query is $O(n_0)$ where $n_0=|P_0|$
because the size of $\varepsilon$-approximate top-$k$ tuples can be $O(n_0)$.
Hence, it takes $O(M \cdot n_0)$ time to compute the top-$k$ results.
Then, \textsc{Greedy} runs $O(r)$ iterations to get a set-cover solution.
At each iteration, it evaluates $O(n_0)$ sets to find $S^*$ in Line~\ref{ln:cover:max}
of Algorithm~\ref{alg:set:cover}.
Thus, the time complexity of \textsc{Greedy} is $O(r \cdot n_0)$.
FD-RMS calls \textsc{Greedy} $O(\log{M})$ times to determine the value of $m$.
Therefore, the time complexity of computing $Q_{0}$ on $P_0$
is $O\big((M + r\log{M}) \cdot n_0\big)$.
In Algorithm~\ref{alg:rms:update},
the time complexity of updating the top-$k$ results and set system $\Sigma$
is $O\big(\mathtt{u}(\Delta_t)\cdot n_t \big)$ where $\mathtt{u}(\Delta_t)$
is the number of utility vectors whose top-$k$ results are changed by $\Delta_t$.
Then, the maximum number of reassignments in cover sets is
$|S(p)|$ for $\Delta_t$, which is bounded by $O(\mathtt{u}(\Delta_t))$.
In addition, the time complexity of \textsc{Stabilize} is $O(m \log{m})$ according to Lemma~\ref{lm:termination}.
Moreover, the maximum difference between the old and new values of $m$ is bounded by $O(m)$.
Hence, the total time complexity of updating $Q_t$ w.r.t.~$\Delta_t$
is $O\big(\mathtt{u}(\Delta_t) \cdot n_t + m^2 \log{m}\big)$.

\subsection{Implementation Issues}\label{subsec:impl}
\noindent\textbf{Index Structures:}
As indicated in Line~\ref{ln:rms:init:topk} of Algorithm~\ref{alg:rms:init}
and Line~\ref{ln:update:topk} of Algorithm~\ref{alg:rms:update},
FD-RMS should compute the $\varepsilon$-approximate top-$k$ result of each $u_i$ ($i \in [1,M]$)
on $P_0$ and update it w.r.t.~$\Delta_t$.
Here, we elaborate on our implementation for top-$k$ maintenance.
In order to process a large number of (approximate) top-$k$ queries with frequent updates in the database,
we implement a \emph{dual-tree}~\cite{DBLP:conf/kdd/RamG12,DBLP:conf/sigmod/YuAY12,DBLP:conf/sdm/CurtinGR13}
that comprises a tuple index $\mathtt{TI}$ and a utility index $\mathtt{UI}$.

The goal of $\mathtt{TI}$ is to efficiently retrieve the $\varepsilon$-approximate
top-$k$ result $\Phi_{k,\varepsilon}(u,P_t)$ of any utility vector $u$ on the up-to-date $P_t$.
Hence, any space-partitioning index, e.g.,
\emph{k-d tree}~\cite{DBLP:journals/cacm/Bentley75}
and \emph{Quadtree}~\cite{DBLP:journals/acta/FinkelB74},
can serve as $\mathtt{TI}$ for top-$k$ query processing.
In practice, we use \emph{k-d tree} as $\mathtt{TI}$.
We adopt the scheme of~\cite{DBLP:conf/recsys/BachrachFGKKNP14}
to transform a top-$k$ query in $\mathbb{R}^d$
into a $k$NN query in $\mathbb{R}^{d+1}$.
Then, we implement the standard top-down methods to construct $\mathtt{TI}$ on $P_0$
and update it w.r.t.~$\Delta_t$.
The branch-and-bound algorithm is used for top-$k$ queries on $\mathtt{TI}$.

The goal of $\mathtt{UI}$ is to cluster the sampled utility vectors
so as to efficiently find each vector whose $\varepsilon$-approximate
top-$k$ result is updated by~$\Delta_t$.
Since the top-$k$ results of linear functions are merely determined by directions,
the basic idea of $\mathtt{UI}$ is to cluster the utilities with high \emph{cosine similarities} together.
Therefore, we adopt an angular-based binary space partitioning tree
called \emph{cone tree}~\cite{DBLP:conf/kdd/RamG12} as $\mathtt{UI}$.
We generally follow Algorithms 8--9 in~\cite{DBLP:conf/kdd/RamG12} to build $\mathtt{UI}$
for $\{u_1,\ldots,u_M\}$.
We implement a top-down approach based on Section 3.2 of~\cite{DBLP:conf/sigmod/YuAY12}
to update the top-$k$ results affected by~$\Delta_t$.

\noindent\textbf{Parameter Tuning:}
Now, we discuss how to specify the values of $\varepsilon$, i.e., the approximation factor of
top-$k$ queries, and $M$, i.e., the upper bound of $m$, in FD-RMS.
In general, the value of $\varepsilon$ has direct effect on $m$
as well as the efficiency and quality of results of FD-RMS.
In particular, if $\varepsilon$ is larger, the $\varepsilon$-approximate
top-$k$ result of each utility vector will include more tuples and the set system
built on top-$k$ results will be more dense. As a result, to guarantee
the result size to be exactly $r$, FD-RMS will use more utility vectors (i.e., a larger $m$)
for a larger $\varepsilon$.
Therefore, a smaller $\varepsilon$ leads to higher efficiency and lower solution quality
due to smaller $m$ and larger $\delta$, and vice versa.
In our implementation, we use a trial-and-error method to find
appropriate values of $\varepsilon$ and $M$:
For each query $\mathtt{RMS}(k,r)$ on a dataset, we test different values of
$\varepsilon$ chosen from $[0.0001, \ldots, 0.1024]$ and,
for each value of $\varepsilon$, $M$ is set to the smallest one
chosen from $[2^{10},\ldots,2^{20}]$ that always guarantees $m < M$.
If the result size is still smaller than $r$ when $m=M=2^{20}$,
we will not use larger $M$ anymore due to efficiency issue.
The values of $\varepsilon$ and $M$ that strike
the best balance between efficiency and quality of results will be used.
In Fig.~\ref{fig:eps}, we present how the value of $\varepsilon$
affects the performance of FD-RMS empirically.

%% file: section/05-experiments.tex
\section{Experiments}\label{sec:exp}

In this section, we evaluate the performance of FD-RMS on real-world and synthetic datasets.
We first introduce the experimental setup in Section~\ref{subsec:setup}.
Then, we present the experimental results in Section~\ref{subsec:results}.

\subsection{Experimental Setup}\label{subsec:setup}

\noindent\textbf{Algorithms:} The algorithms compared are listed as follows.
\begin{itemize}
  \item \textsc{Greedy}: the greedy algorithm for $1$-RMS in~\cite{DBLP:journals/pvldb/NanongkaiSLLX10}.
  \item \textsc{Greedy}$^*$: the randomized greedy algorithm for $k$-RMS when $k>1$
  proposed in~\cite{DBLP:journals/pvldb/ChesterTVW14}.
  \item \textsc{GeoGreedy}: a variation of \textsc{Greedy} for $1$-RMS in~\cite{DBLP:conf/icde/PengW14}.
  \item DMM-RRMS: a discretized matrix min-max based algorithm for $1$-RMS
  in~\cite{DBLP:conf/sigmod/AsudehN0D17}.
  \item $\varepsilon$-\textsc{Kernel}: computing an $\varepsilon$-kernel coreset
  as the $k$-RMS result~\cite{DBLP:conf/wea/AgarwalKSS17,DBLP:conf/icdt/CaoLWWWWZ17} directly.
  \item HS: a hitting-set based algorithm for $k$-RMS in~\cite{DBLP:conf/wea/AgarwalKSS17}.
  \item \textsc{Sphere}: an algorithm that combines $\varepsilon$-\textsc{Kernel}
  with \textsc{Greedy} for $1$-RMS in~\cite{DBLP:conf/sigmod/XieW0LL18}.
  \item URM: a $k$-\textsc{medoid} clustering based algorithm for $1$-RMS in~\cite{Shetiya2019}.
  Following~\cite{Shetiya2019}, we use DMM-RRMS to compute an initial solution for URM.
  \item FD-RMS: our fully-dynamic $k$-RMS algorithm proposed in Section~\ref{sec:framework}.
\end{itemize}
The algorithms that only work in two dimensions are not compared.
All the above algorithms except FD-RMS and URM\footnote{
The original URM in~\cite{Shetiya2019} is also a static algorithm.
We extend URM to support dynamic updates as described in
Algorithm~\ref{alg:urm} of Appendix~\ref{appendix:dynamic:urm}.}
cannot directly work in a fully dynamic setting.
In our experiments, they rerun from scratch to compute the up-to-date $k$-RMS result
once the skyline is updated by any insertion or deletion.
In addition, the algorithms that are not applicable when $k>1$
are not compared for the experiments with varying $k$.
Since $\varepsilon$-\textsc{Kernel} and HS are proposed for min-size
$k$-RMS that returns the smallest subset whose maximum $k$-regret ratio is at most $\varepsilon$,
we adapt them to our problem by performing a binary search on $\varepsilon$ in the range
$(0,1)$ to find the smallest value of $\varepsilon$
that guarantees the result size is at most $r$.

Our implementation of FD-RMS and URM is in Java 8 and published on
GitHub\footnote{https://github.com/yhwang1990/dynamic-rms}.
We used the C++ implementations of baseline algorithms published by authors
and followed the default parameter settings as described in the original papers.
All the experiments were conducted on a server running
Ubuntu 18.04.1 with a 2.3GHz processor and 256GB memory.

\noindent\textbf{Datasets:}
The datasets we use are listed as follows.
\begin{itemize}
  \item \textbf{BB}\footnote{\url{www.basketball-reference.com}}
  is a basketball dataset that contains $21,961$ tuples, each of which represents one player/season combination
  with $5$ attributes such as \emph{points} and \emph{rebounds}.
  \item \textbf{AQ}\footnote{\url{archive.ics.uci.edu/ml/datasets/Beijing+Multi-Site+Air-Quality+Data}}
  includes hourly air-pollution and weather data from 12 monitoring sites in Beijing.
  It has $382,168$ tuples and each tuple has $9$ attributes including the concentrations of $6$ air pollutants like
  $\text{PM}_{\text{2.5}}$, as well as $3$ meteorological parameters
  like \emph{temperature}.
  \item \textbf{CT}\footnote{\url{archive.ics.uci.edu/ml/datasets/covertype}}
  contains the cartographic data of forest covers in the Roosevelt National Forest of northern Colorado.
  It has $581,012$ tuples and we choose $8$ numerical attributes,
  e.g., \emph{elevation} and \emph{slope}, for evaluation.
  \item \textbf{Movie}\footnote{\url{grouplens.org/datasets/movielens}}
  is the tag genome dataset published by MovieLens.
  We extract the relevance scores of $13,176$ movies and $12$ tags for evaluation.
  Each tuple represents the relevance scores of $12$ tags to a movie.
  \item \textbf{Indep} is generated as described in~\cite{DBLP:conf/icde/BorzsonyiKS01}.
  It is a set of uniform points on the unit hypercube where different attributes are
  independent of each other.
  \item \textbf{AntiCor} is also generated as described in~\cite{DBLP:conf/icde/BorzsonyiKS01}.
  It is a set of random points with anti-correlated attributes.
\end{itemize}
\begin{table}[t]
  \centering
  \footnotesize
  \caption{Statistics of datasets}\label{tbl:stats}
  \begin{tabular}{|c|c|c|c|c|}
  \hline
  \textbf{Dataset}    & $n$          & $d$       & \emph{\#skylines} & \emph{updates} (\%) \\ \hline
  \textbf{BB}         & $21,961$     & $5$       & $200$             & $1.07$                \\ \hline
  \textbf{AQ}         & $382,168$    & $9$       & $21,065$          & $5.60$                \\ \hline
  \textbf{CT}         & $581,012$    & $8$       & $77,217$          & $13.4$                \\ \hline
  \textbf{Movie}      & $13,176$     & $12$      & $3,293$           & $26.5$                \\ \hline
  \textbf{Indep}      & $100$K--$1$M & $2$--$10$ & \multicolumn{2}{c|}{see Fig.~\ref{fig:stats}} \\ \hline
  \textbf{AntiCor}    & $100$K--$1$M & $2$--$10$ & \multicolumn{2}{c|}{see Fig.~\ref{fig:stats}} \\ \hline
  \end{tabular}
\end{table}
\begin{figure}[t]
  \centering
  \begin{subfigure}{0.3\textwidth}
    \centering
    \includegraphics[width=\textwidth]{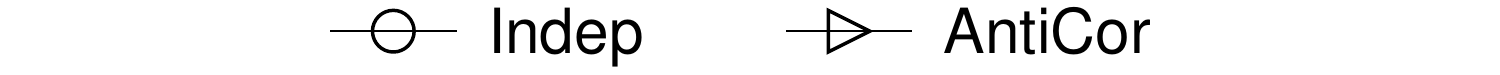}
  \end{subfigure}
  \\
  \begin{subfigure}{0.24\textwidth}
    \centering
    \includegraphics[width=\textwidth]{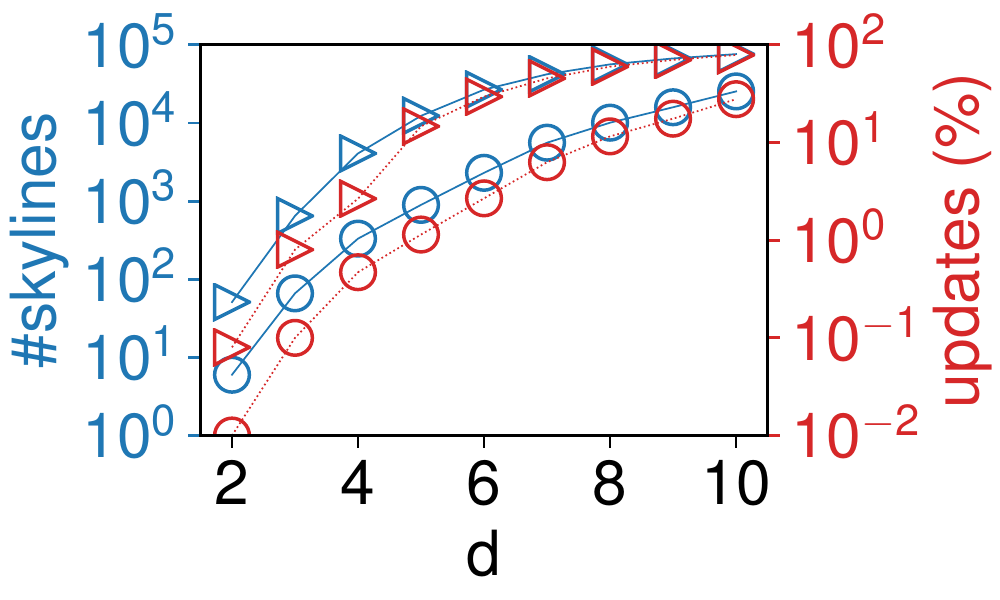}
  \end{subfigure}
  \begin{subfigure}{0.24\textwidth}
    \centering
    \includegraphics[width=\textwidth]{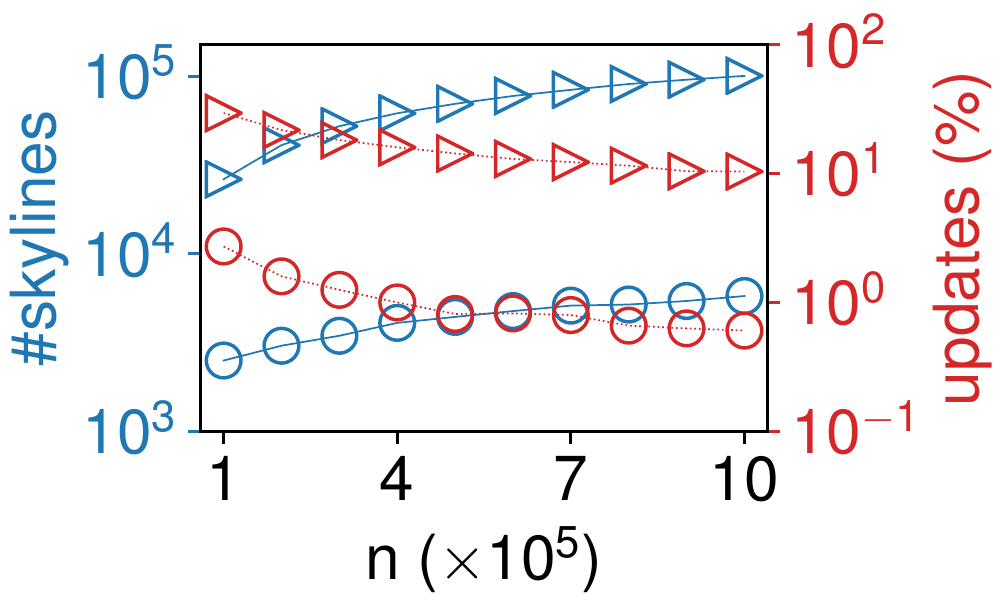}
  \end{subfigure}
  \caption{Sizes and update rates of the skylines of synthetic datasets}
  \label{fig:stats}
\end{figure}
The statistics of datasets are reported in Table~\ref{tbl:stats}.
Here, $n$ is the number of tuples; $d$ is the dimensionality; \emph{\#skylines}
is the number of tuples on the skyline; and \emph{updates} (\%) is the percentage
of tuple operations that trigger any update on the skyline.
Note that we generated several \textbf{Indep} and \textbf{AntiCor} datasets
by varying $n$ from $100$K to $1$M and $d$ from $2$ to $10$ for scalability tests.
By default, we used the ones with $n=100$K and $d=6$.
The sizes and update rates of the skylines of synthetic datasets
are shown in Fig.~\ref{fig:stats}.

\noindent\textbf{Workloads:} The workload of each experiment was generated as follows:
First, we randomly picked $50\%$ of tuples as the initial dataset $P_0$;
Second, we inserted the remaining $50\%$ of tuples one by one into the dataset
to test the performances for insertions;
Third, we randomly deleted $50\%$ of tuples one by one from the dataset to test the performances for deletions.
It is guaranteed that the orders of operations kept the same for all algorithms.
The $k$-RMS results were recorded $10$ times when $10\%, 20\%, \ldots, 100\%$ of the operations were performed.

\noindent\textbf{Performance Measures:}
The efficiency of each algorithm was measured by \emph{average update time}, i.e.,
the average wall-clock time of an algorithm to update the result of $k$-RMS for each operation.
For the static algorithms, we only took the time for $k$-RMS computation into account
and ignored the time for skyline maintenance for fair comparison.
The quality of results was measured by the \emph{maximum $k$-regret ratio} ($\mathtt{mrr}_k$)
for a given size constraint $r$, and, of course, the smaller $\mathtt{mrr}_k$ the better.
To compute $\mathtt{mrr}_k(Q)$ of a result $Q$, we generated a test
set of $500$K random utility vectors and used the maximum regret value found as our estimate.
Since the $k$-RMS results were recorded $10$ times for each query,
we reported the average of the maximum $k$-regret ratios of $10$ results for evaluation.

\subsection{Experimental Results}\label{subsec:results}

\begin{figure*}[t]
  \centering
  \begin{subfigure}{.18\textwidth}
    \centering
    \includegraphics[width=\textwidth]{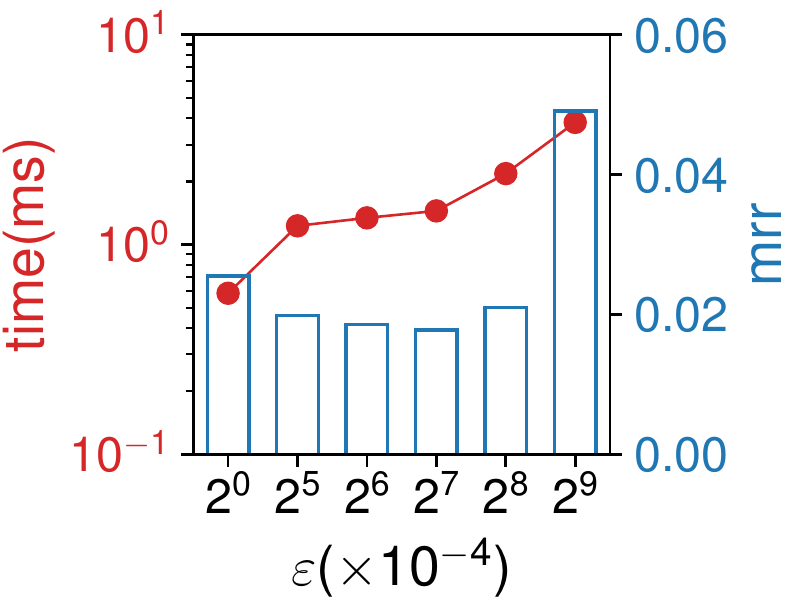}
    \caption{BB}
  \end{subfigure}
  \begin{subfigure}{.18\textwidth}
    \centering
    \includegraphics[width=\textwidth]{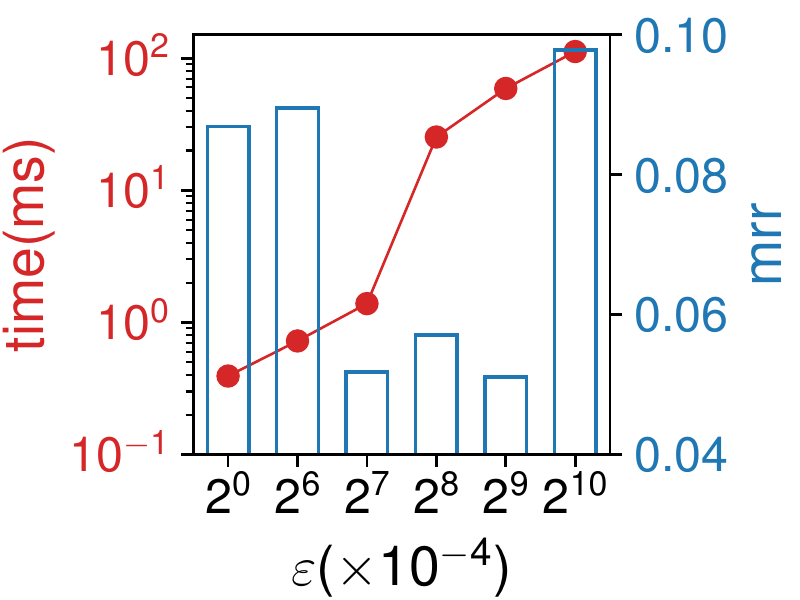}
    \caption{AQ}
  \end{subfigure}
  \begin{subfigure}{.18\textwidth}
    \centering
    \includegraphics[width=\textwidth]{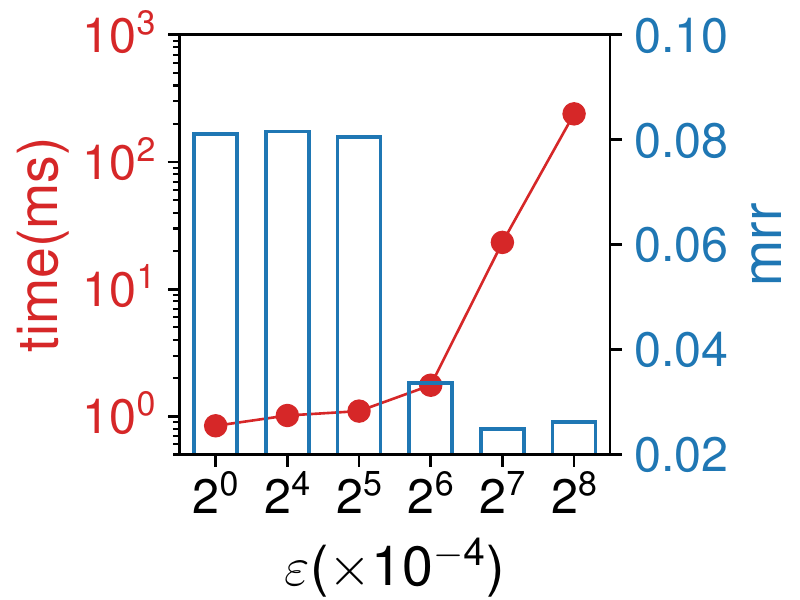}
    \caption{CT}
  \end{subfigure}
  \\
  \begin{subfigure}{.18\textwidth}
    \centering
    \includegraphics[width=\textwidth]{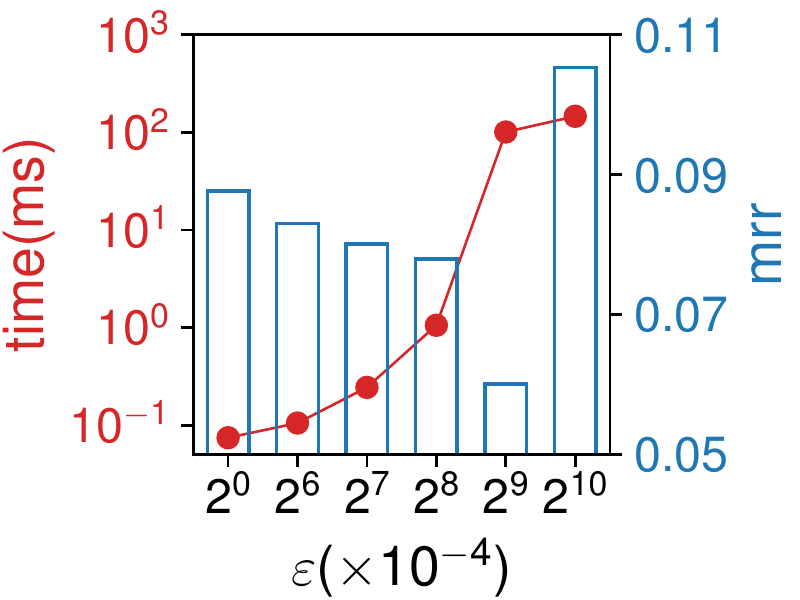}
    \caption{Movie}
  \end{subfigure}
  \begin{subfigure}{.18\textwidth}
    \centering
    \includegraphics[width=\textwidth]{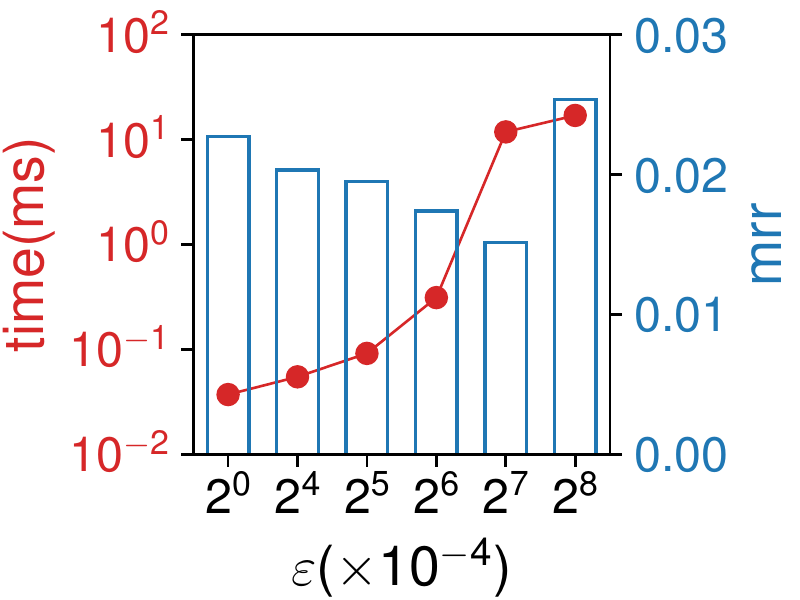}
    \caption{Indep}
  \end{subfigure}
  \begin{subfigure}{.18\textwidth}
    \centering
    \includegraphics[width=\textwidth]{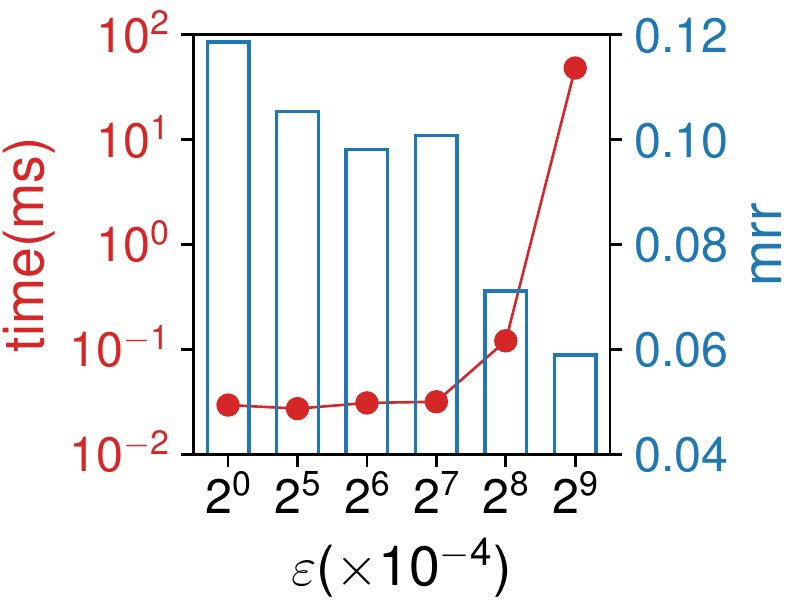}
    \caption{AntiCor}
  \end{subfigure}
  \caption{Performance of FD-RMS with varying $\varepsilon$
  ($k=1$; $r=20$ for BB and $r=50$ for other datasets).
  Note that the red line represents the update time and the blue bars denote
  the maximum regret ratios.}
  \label{fig:eps}
\end{figure*}
\begin{figure*}[t]
  \centering
  \begin{subfigure}{.96\textwidth}
    \centering
    \includegraphics[width=\textwidth]{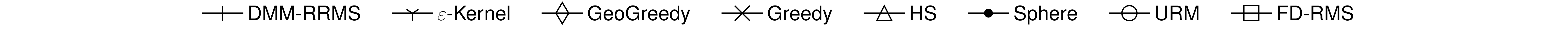}
  \end{subfigure}
  \begin{subfigure}{.36\textwidth}
    \centering
    \includegraphics[width=.48\textwidth]{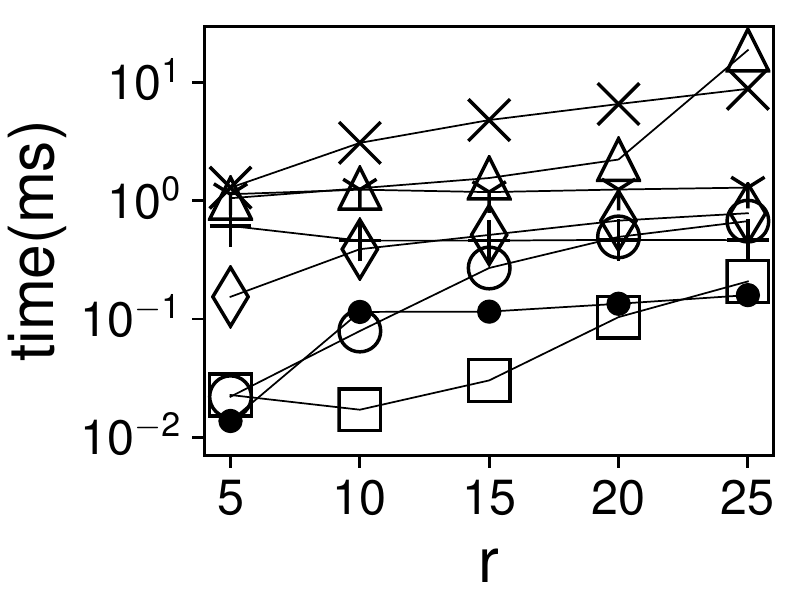}
    \includegraphics[width=.48\textwidth]{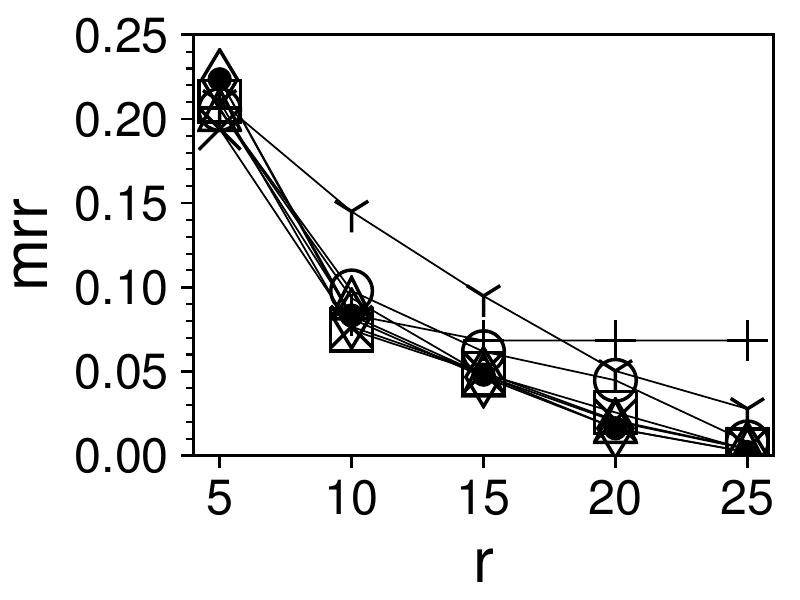}
    \caption{BB}
  \end{subfigure}
  \begin{subfigure}{.36\textwidth}
    \centering
    \includegraphics[width=.48\textwidth]{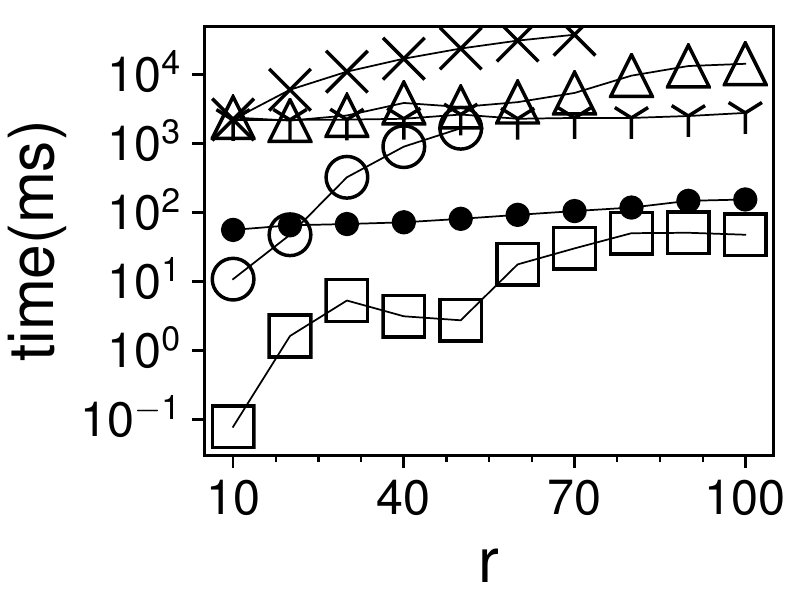}
    \includegraphics[width=.48\textwidth]{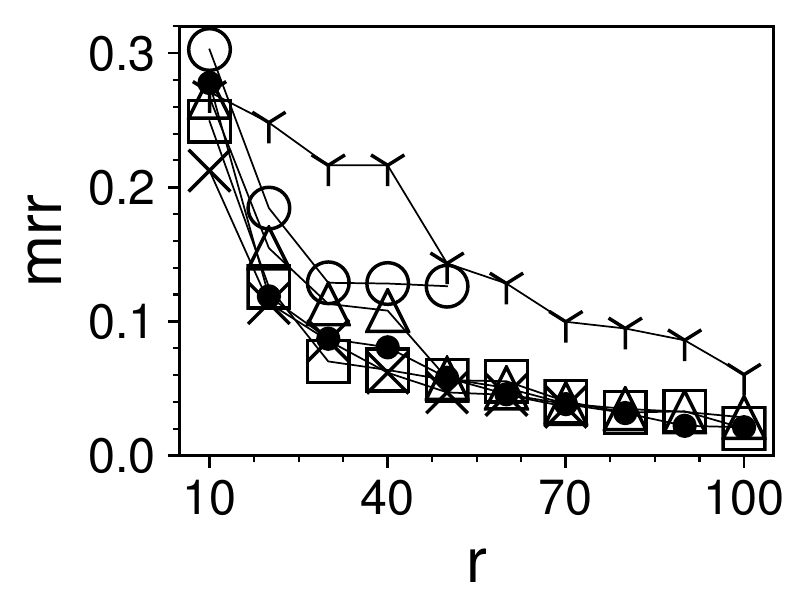}
    \caption{AQ}
  \end{subfigure}
  \\
  \begin{subfigure}{.36\textwidth}
    \centering
    \includegraphics[width=.48\textwidth]{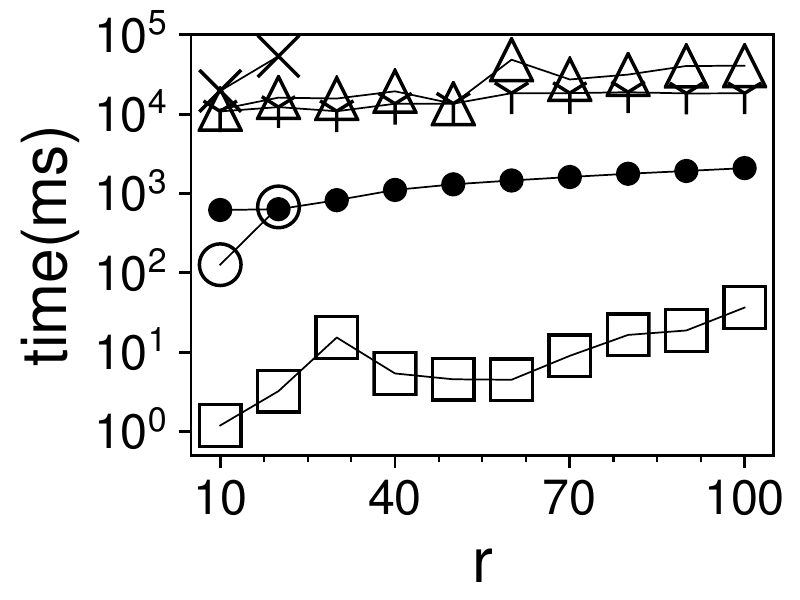}
    \includegraphics[width=.48\textwidth]{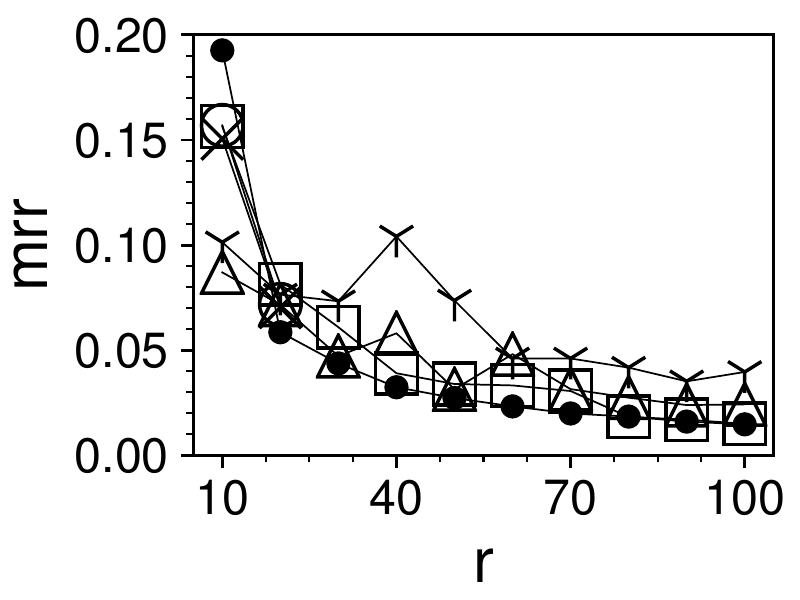}
    \caption{CT}
  \end{subfigure}
  \begin{subfigure}{.36\textwidth}
    \centering
    \includegraphics[width=.48\textwidth]{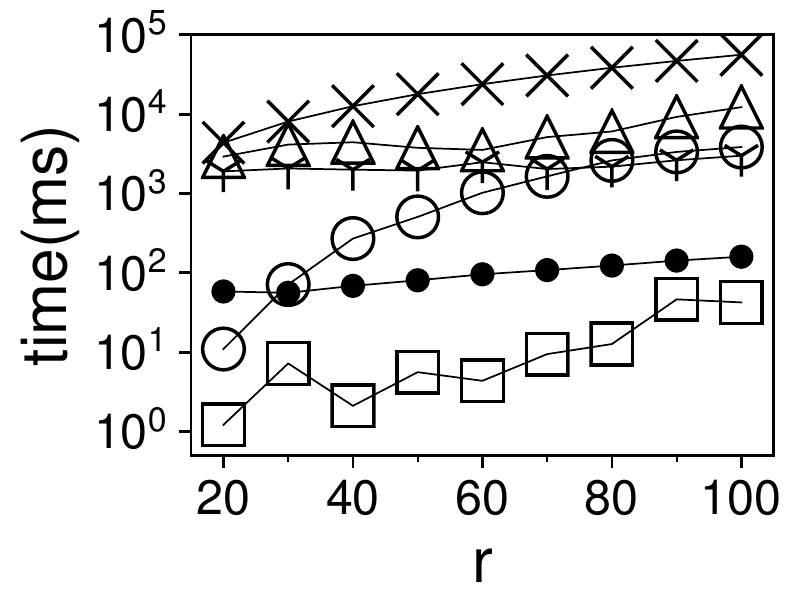}
    \includegraphics[width=.48\textwidth]{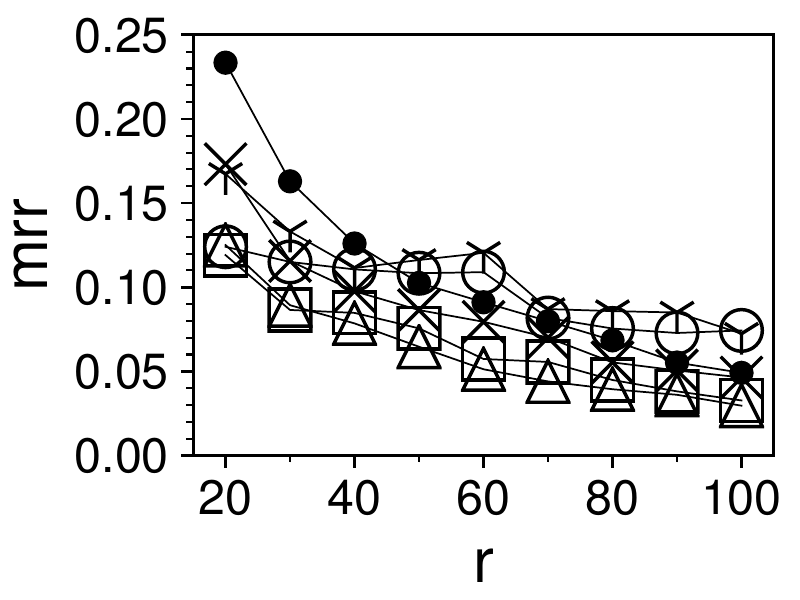}
    \caption{Movie}
  \end{subfigure}
  \\
  \begin{subfigure}{.36\textwidth}
    \centering
    \includegraphics[width=.48\textwidth]{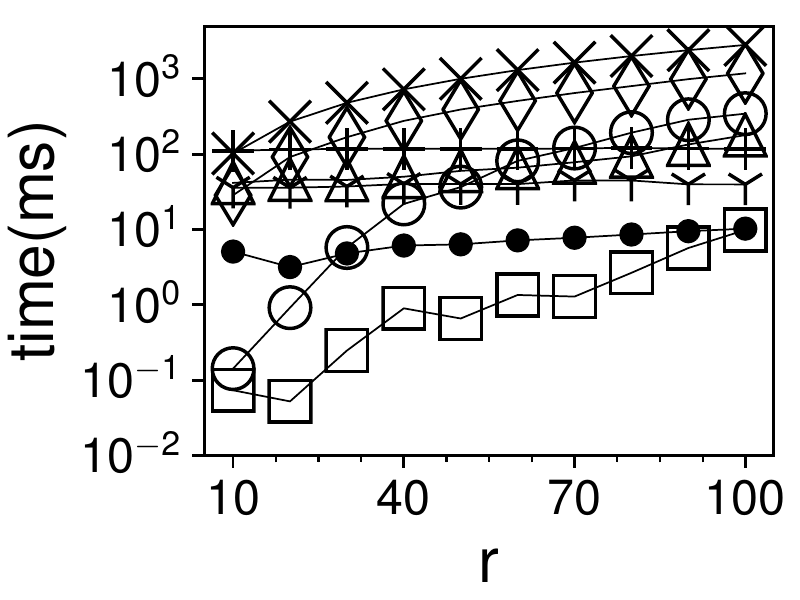}
    \includegraphics[width=.48\textwidth]{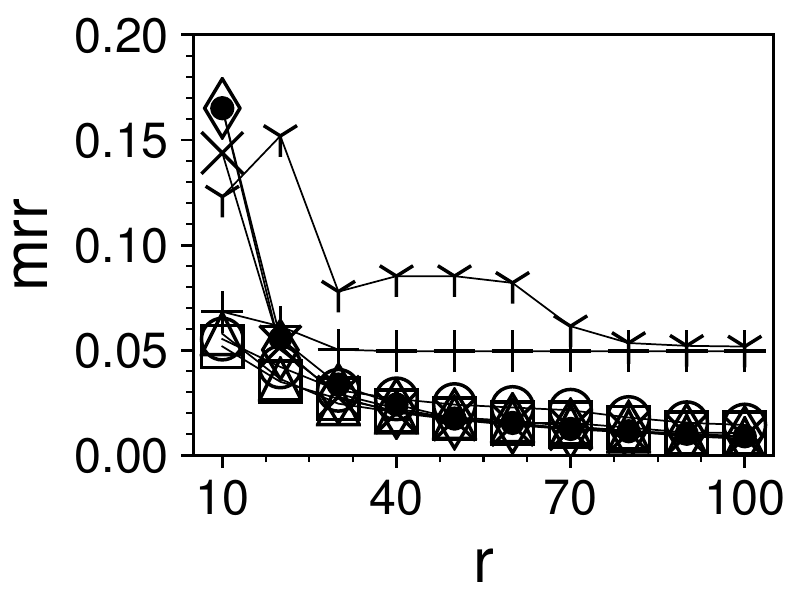}
    \caption{Indep}
  \end{subfigure}
  \begin{subfigure}{.36\textwidth}
    \centering
    \includegraphics[width=.48\textwidth]{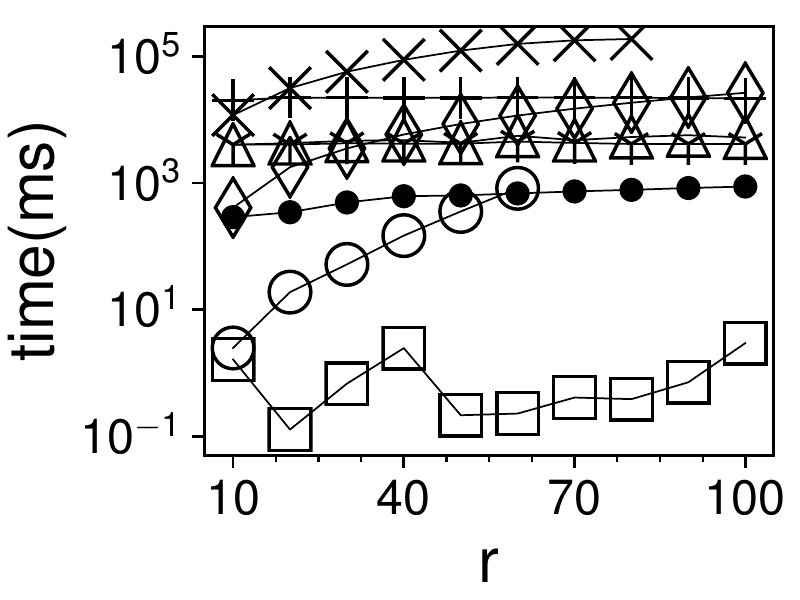}
    \includegraphics[width=.48\textwidth]{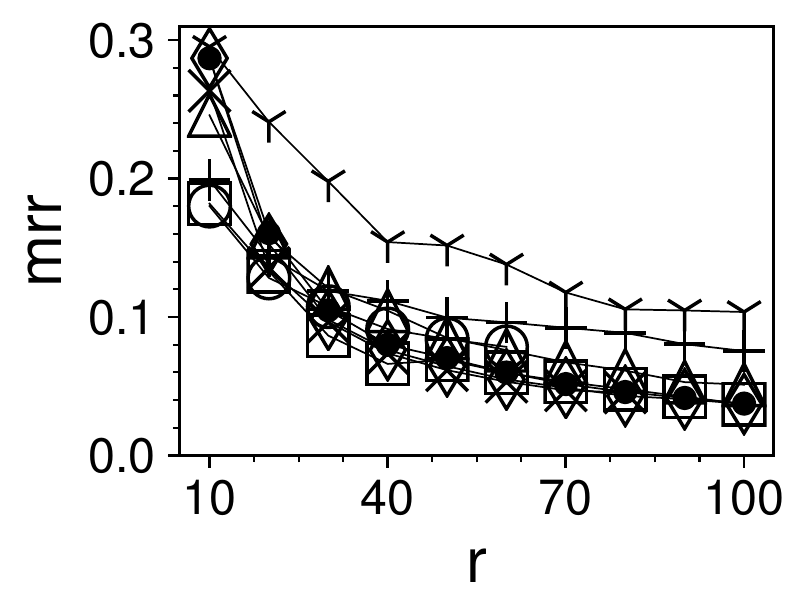}
    \caption{AntiCor}
  \end{subfigure}
  \caption{Update time and maximum regret ratios with varying the result size $r$ ($k=1$)}
  \label{fig:r}
\end{figure*}

\noindent\textbf{Effect of parameter $\varepsilon$ on FD-RMS:}
In Fig.~\ref{fig:eps}, we present the effect of the parameter $\varepsilon$
on the performance of FD-RMS.
We report the update time and maximum regret ratios
of FD-RMS for $k=1$ and $r=50$ on each dataset (except $r=20$ on \textbf{BB})
with varying $\varepsilon$.
We use the method described in Section~\ref{subsec:impl} to set the value of $M$
for each value of $\varepsilon$.
First of all, the update time of FD-RMS increases significantly
with $\varepsilon$. This is because both the time to process an $\varepsilon$-approximate
top-$k$ query and the number of top-$k$ queries (i.e.,~$M$) grow with $\varepsilon$,
which requires a larger overhead to maintain both top-$k$ results and set-cover solutions.
Meanwhile, the quality of results first becomes better when
$\varepsilon$ is larger but then could degrade if $\varepsilon$ is too large.
The improvement in quality with increasing $\varepsilon$ is attributed to
larger $m$ and thus smaller $\delta$.
However, once $\varepsilon$ is greater than the maximum regret ratio $\varepsilon^*_{k,r}$
of the optimal result (whose upper bound can be inferred from practical results),
e.g.,~$\varepsilon=0.0512$ on \textbf{BB},
the result of FD-RMS will contain less than $r$ tuples and
its maximum regret ratio will be close to $\varepsilon$ no matter how large $m$ is.
To sum up, by setting $\varepsilon$ to the one that is slightly lower
than $\varepsilon^*_{k,r}$ among $[0.0001,\ldots,0.1024]$, FD-RMS performs better
in terms of both efficiency and solution quality, and the values of $\varepsilon$
in FD-RMS are decided in this way for the remaining experiments.

\noindent\textbf{Effect of result size $r$:}
In Fig.~\ref{fig:r}, we present the performance of different algorithms for $1$-RMS
(a.k.a.~$r$-regret query) with varying $r$.
In particular, $r$ is ranged from $10$ to $100$ on each dataset
(except \textbf{BB} where $r$ is ranged from $5$ to $25$).
In general, the update time of each algorithm grows while the maximum regret ratios
drop with increasing $r$. But, for FD-RMS, it could take less update time
when $r$ is larger in some cases. The efficiency of FD-RMS is positively correlated
with $m$ but negatively correlated with $\varepsilon$. On a specific dataset, FD-RMS
typically chooses a smaller $\varepsilon$ when $r$ is large, and vice versa.
When $\varepsilon$ is smaller, $m$ may decrease even though $r$ is larger.
Therefore, the update time of FD-RMS decreases with $r$ in some cases
because a smaller $\varepsilon$ is used.
Among all algorithms tested, \textsc{Greedy} is the slowest and fails to
provide results within one day on \textbf{AQ}, \textbf{CT}, and \textbf{AntiCor} when $r>80$.
\textsc{GeoGreedy} runs much faster than \textsc{Greedy}
while having equivalent quality on low-dimensional data.
However, it cannot scale up to high dimensions (i.e., $d>7$) because the cost of
finding \emph{happy points} grows significantly with $d$.
DMM-RRMS suffers from two drawbacks:
(1) it also cannot scale up to $d>7$ due to huge memory consumption;
(2) its solution quality is not competitive when $r \geq 50$
because of the sparsity of space discretization.
The solution quality of $\varepsilon$-\textsc{Kernel} is generally inferior to
any other algorithm because the size of an $\varepsilon$-kernel coreset is much larger
than the size of the minimum $(1,\varepsilon)$-regret set.
Although HS provides results of good quality in most cases,
it runs several orders of magnitude slower than FD-RMS.
\textsc{Sphere} demonstrates better performance than other static algorithms.
Nevertheless, its efficiency is still much lower than FD-RMS,
especially on datasets with large skyline sizes, e.g., \textbf{CT} and \textbf{AntiCor},
where FD-RMS runs up to three orders of magnitude faster.
URM shows good performance in both efficiency and solution quality
for small $r$ (e.g., $r \leq 20$) and skyline sizes (e.g., on \textbf{BB} and \textbf{Indep}).
However, it scales poorly to large $r$ and skyline sizes because
the convergence of k-\textsc{medoid} becomes very slow and
the number of linear programs for regret computation grows rapidly
when $r$ and and skyline sizes increase.
In addition, URM does not provide high-quality results in many cases
since k-\textsc{medoid} cannot escape from local optima.
To sum up, FD-RMS outperforms all other algorithms for fully-dynamic $1$-RMS
in terms of efficiency.
Meanwhile, the maximum regret ratios of the results of FD-RMS are very close
(the differences are less than $0.01$ in almost all cases) to the best of static algorithms.

\begin{figure*}[t]
  \centering
  \begin{subfigure}{.75\textwidth}
    \centering
    \includegraphics[width=\textwidth]{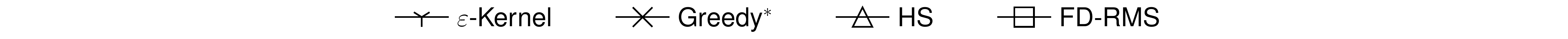}
  \end{subfigure}
  \begin{subfigure}{.36\textwidth}
    \centering
    \includegraphics[width=.48\textwidth]{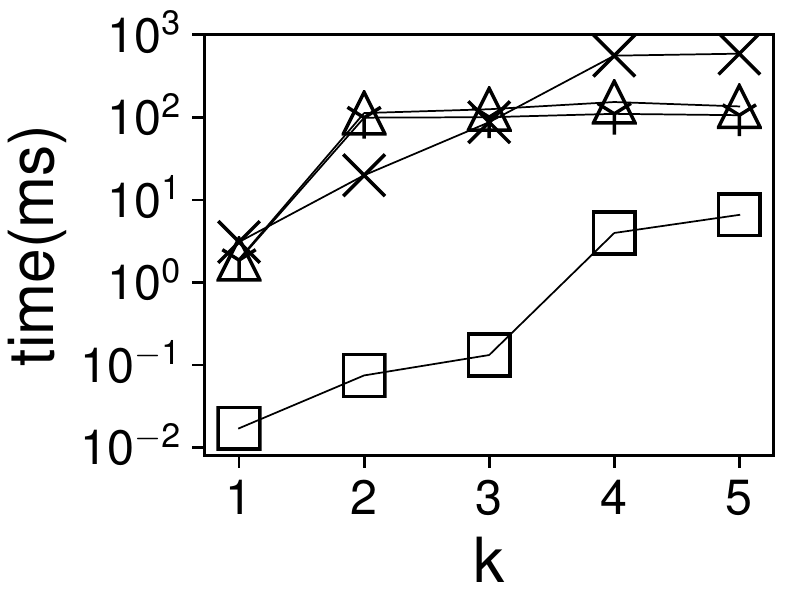}
    \includegraphics[width=.48\textwidth]{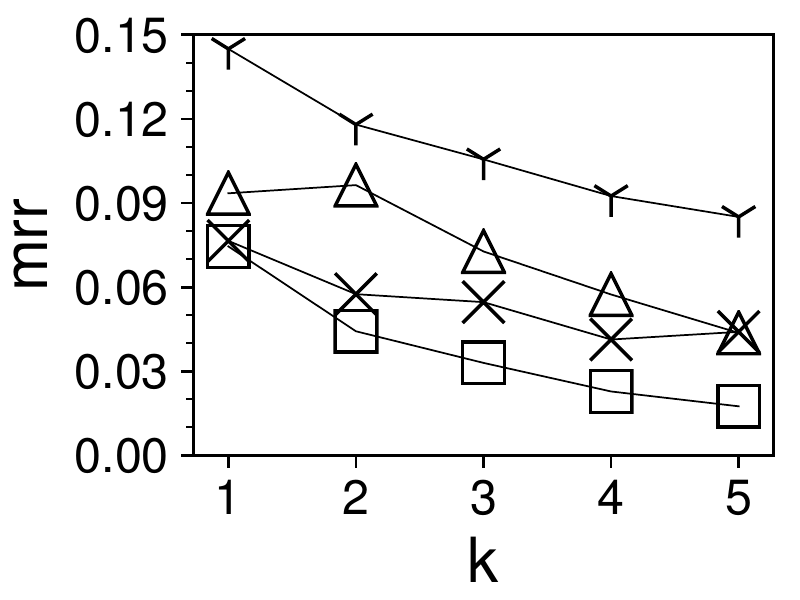}
    \caption{BB}
  \end{subfigure}
  \begin{subfigure}{.36\textwidth}
    \centering
    \includegraphics[width=.48\textwidth]{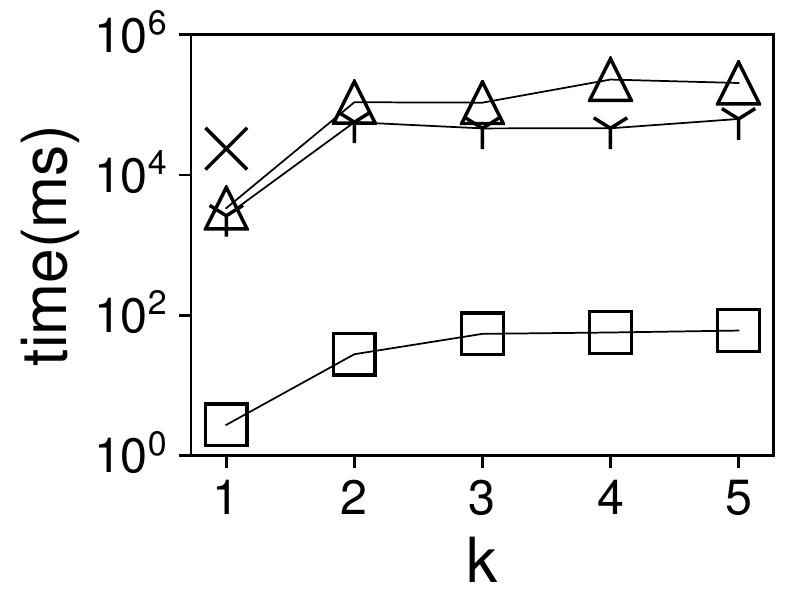}
    \includegraphics[width=.48\textwidth]{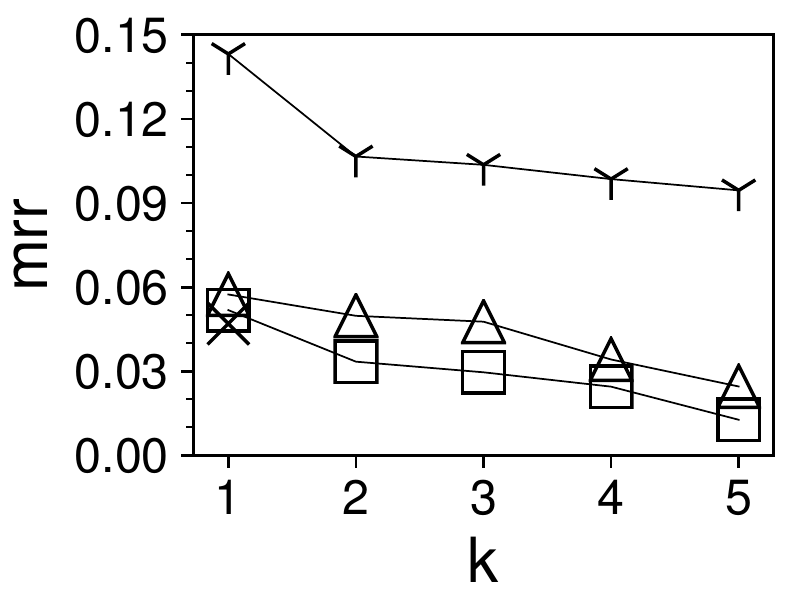}
    \caption{AQ}
  \end{subfigure}
  \\
  \begin{subfigure}{.36\textwidth}
    \centering
    \includegraphics[width=.48\textwidth]{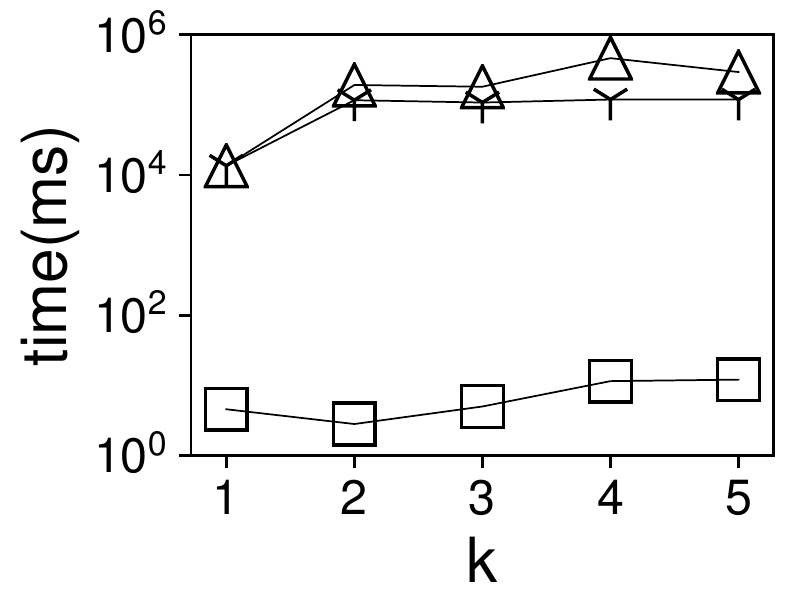}
    \includegraphics[width=.48\textwidth]{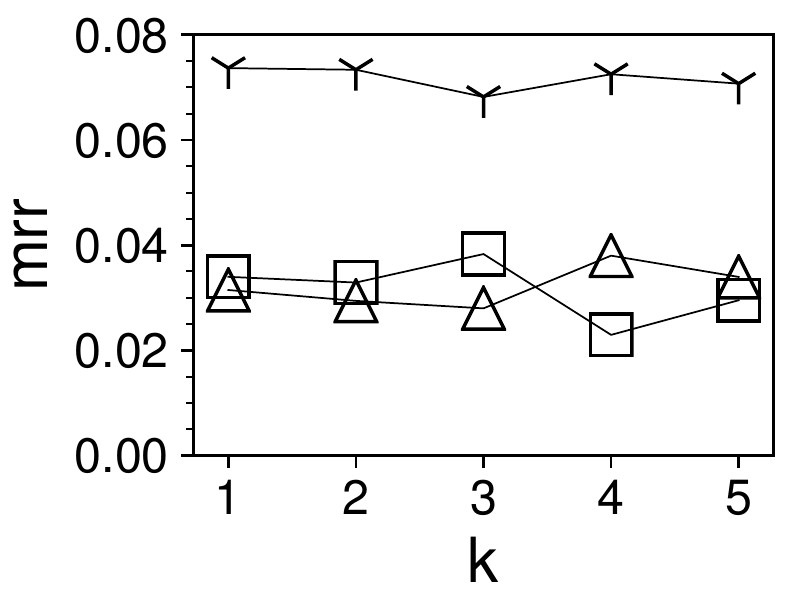}
    \caption{CT}
  \end{subfigure}
  \begin{subfigure}{.36\textwidth}
    \centering
    \includegraphics[width=.48\textwidth]{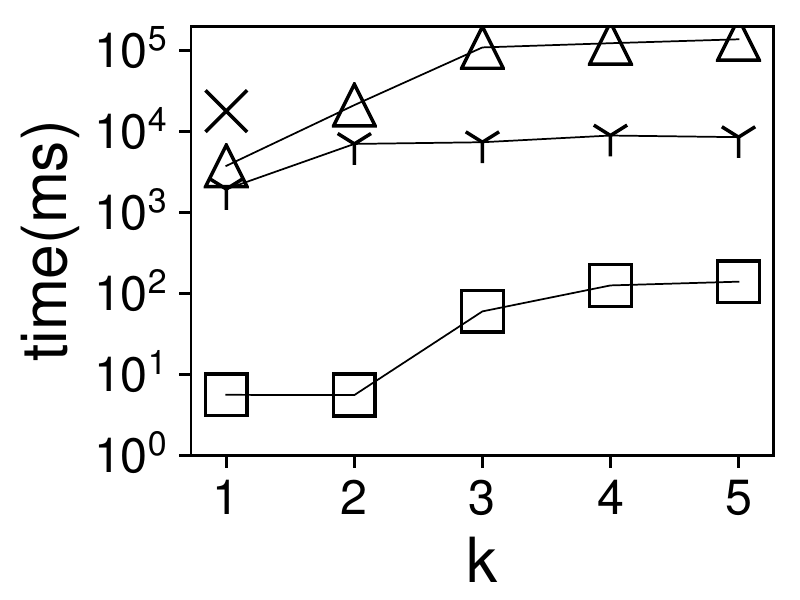}
    \includegraphics[width=.48\textwidth]{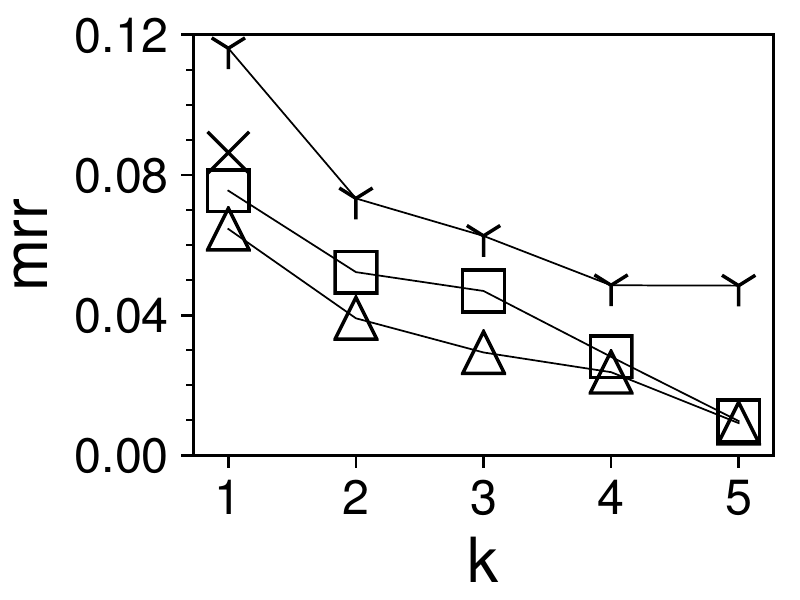}
    \caption{Movie}
  \end{subfigure}
  \\
  \begin{subfigure}{.36\textwidth}
    \centering
    \includegraphics[width=.48\textwidth]{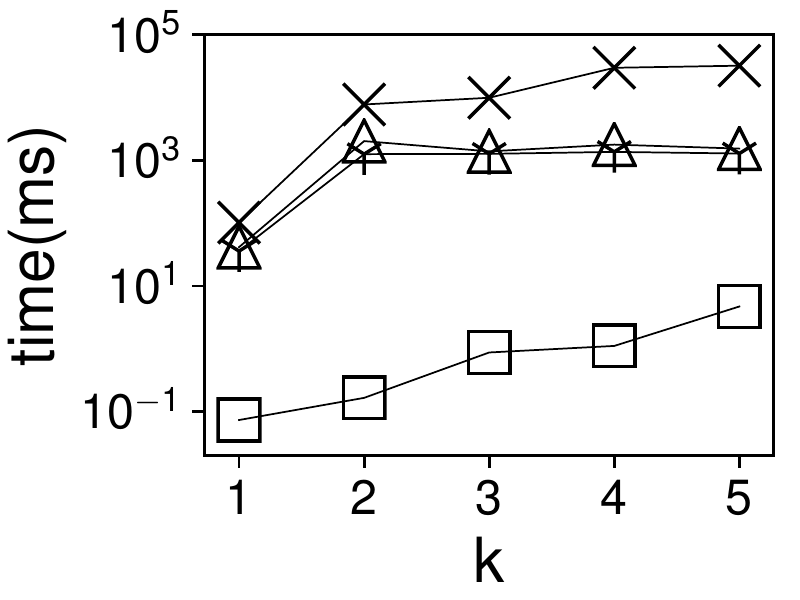}
    \includegraphics[width=.48\textwidth]{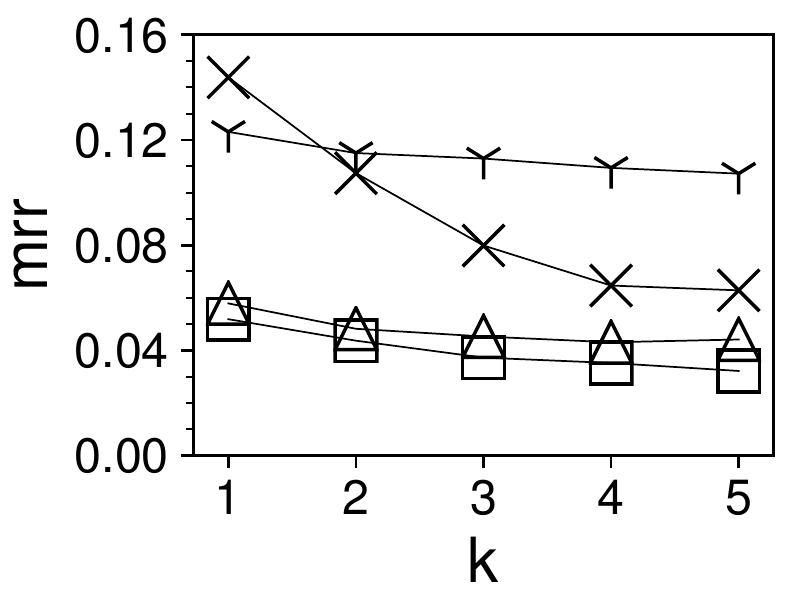}
    \caption{Indep}
  \end{subfigure}
  \begin{subfigure}{.36\textwidth}
    \centering
    \includegraphics[width=.48\textwidth]{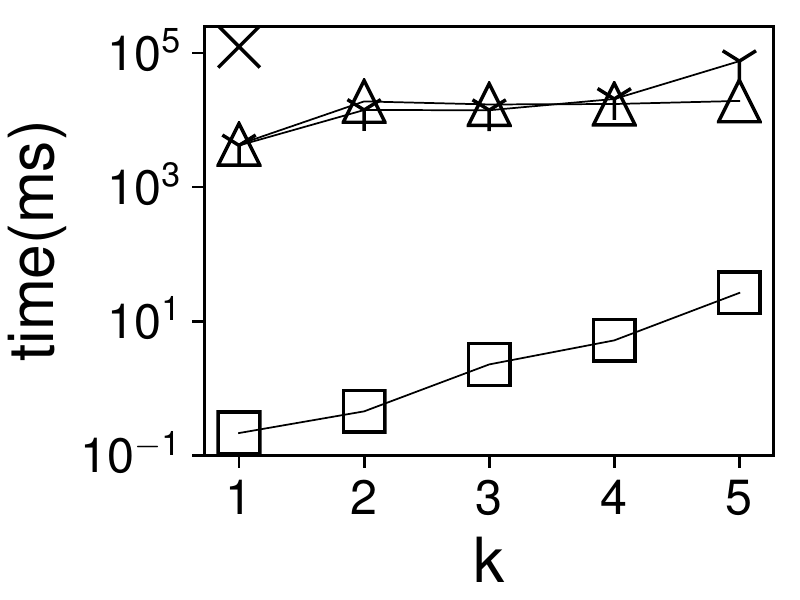}
    \includegraphics[width=.48\textwidth]{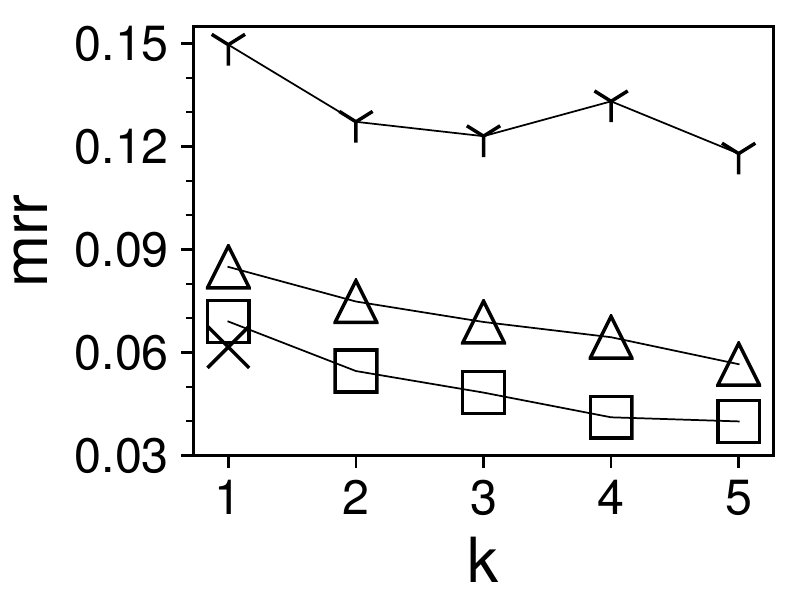}
    \caption{AntiCor}
  \end{subfigure}
  \caption{Update time and maximum regret ratios with varying $k$
  ($r=10$ for BB and Indep; $r=50$ for other datasets)}
  \label{fig:k}
\end{figure*}

\begin{figure*}[t]
  \centering
  \begin{subfigure}{.96\textwidth}
    \centering
    \includegraphics[width=\textwidth]{legend-k1.pdf}
  \end{subfigure}
  \begin{subfigure}{.36\textwidth}
    \centering
    \includegraphics[width=.48\textwidth]{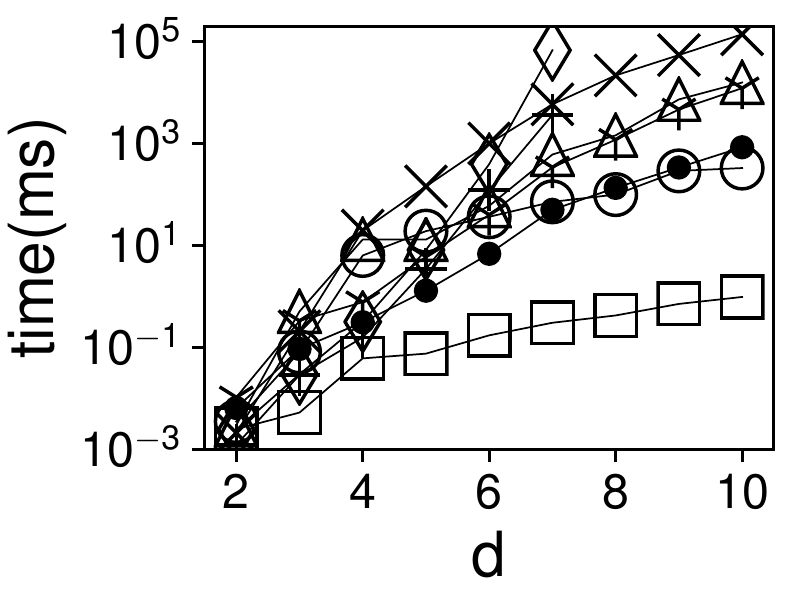}
    \includegraphics[width=.48\textwidth]{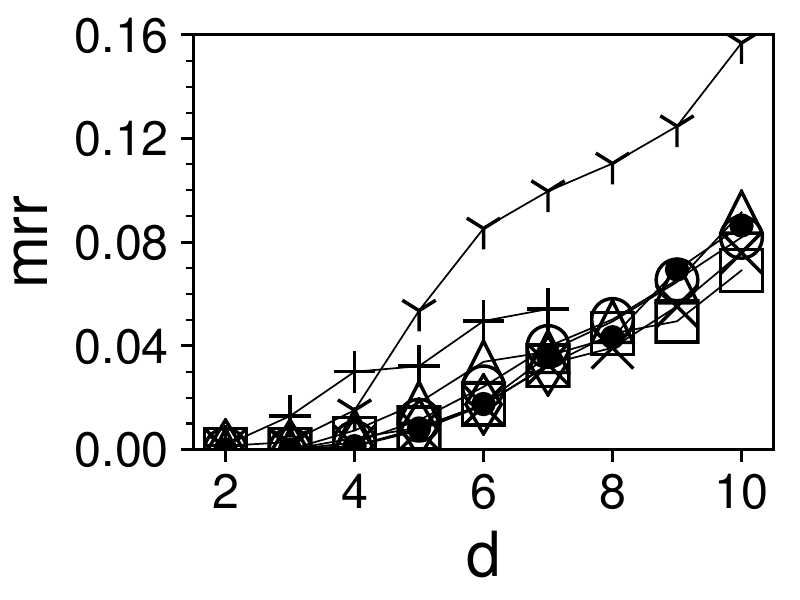}
    \caption{Indep, varying $d$}\label{fig:nd:a}
  \end{subfigure}
  \begin{subfigure}{.36\textwidth}
    \centering
    \includegraphics[width=.48\textwidth]{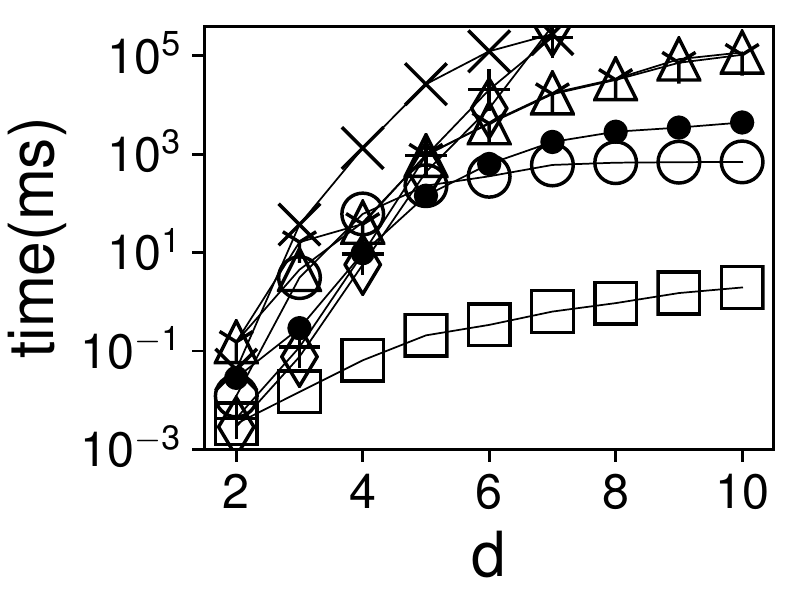}
    \includegraphics[width=.48\textwidth]{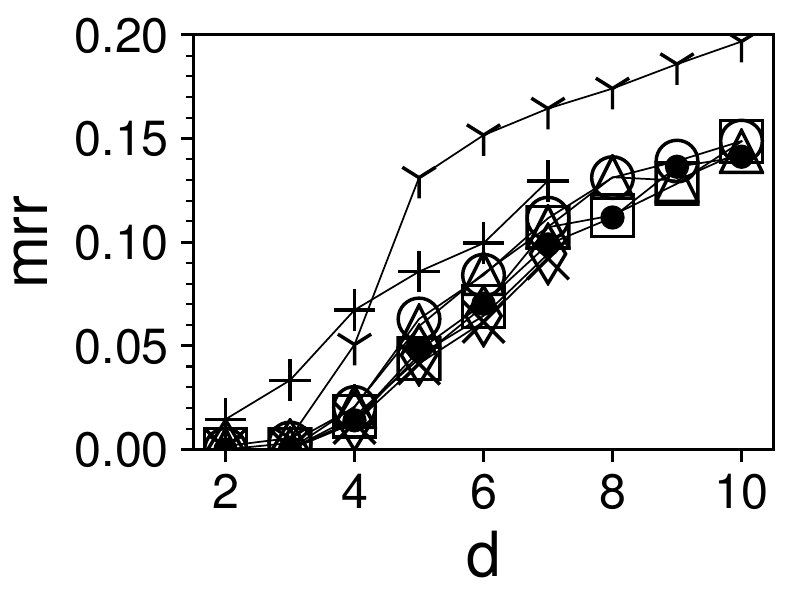}
    \caption{AntiCor, varying $d$}\label{fig:nd:b}
  \end{subfigure}
  \\
  \begin{subfigure}{.36\textwidth}
    \centering
    \includegraphics[width=.48\textwidth]{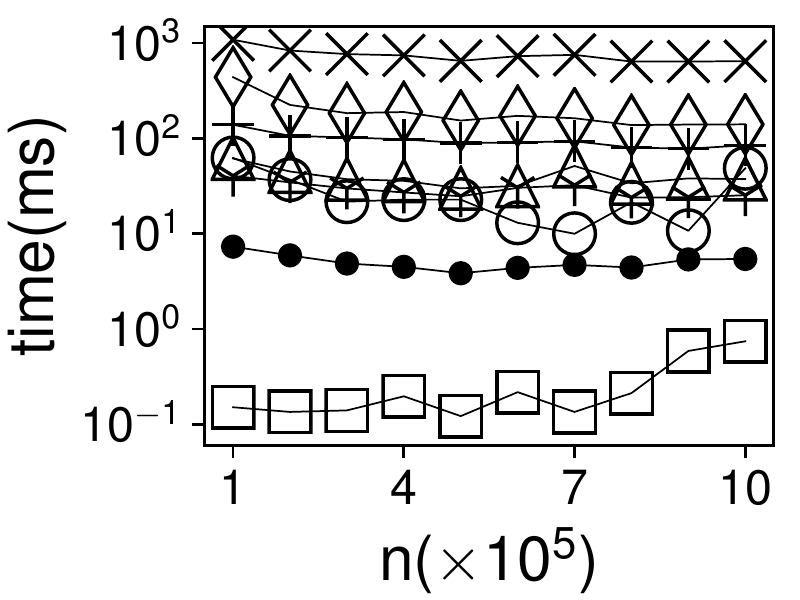}
    \includegraphics[width=.48\textwidth]{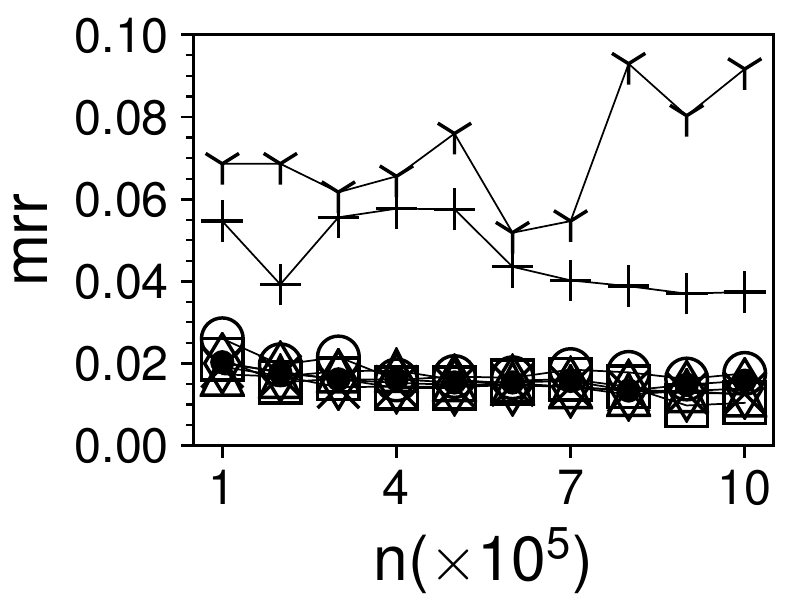}
    \caption{Indep, varying $n$}\label{fig:nd:c}
  \end{subfigure}
  \begin{subfigure}{.36\textwidth}
    \centering
    \includegraphics[width=.48\textwidth]{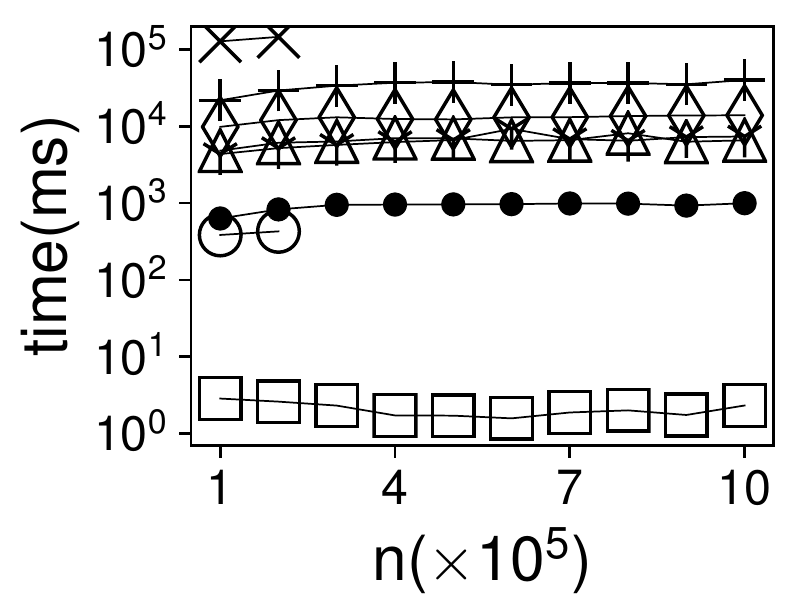}
    \includegraphics[width=.48\textwidth]{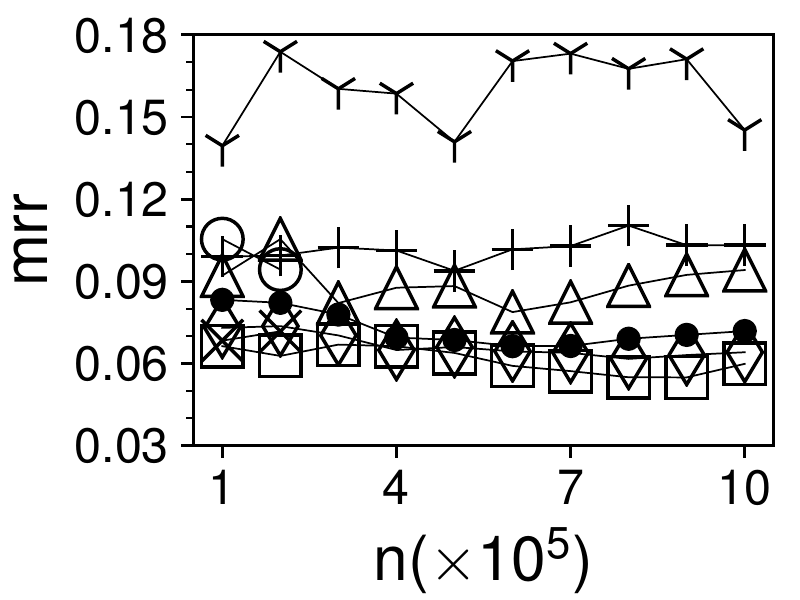}
    \caption{AntiCor, varying $n$}\label{fig:nd:d}
  \end{subfigure}
  \caption{Scalability with varying the dimensionality $d$ and dataset size $n$
  ($k=1, \; r=50$)}
  \label{fig:nd}
\end{figure*}

\noindent\textbf{Effect of $k$:}
The results for $k$-RMS with varying $k$ from $1$ to $5$ are illustrated
in Fig.~\ref{fig:k}.
We only compare FD-RMS with \textsc{Greedy}$^*$,
$\varepsilon$-\textsc{Kernel}, and HS
because other algorithms are not applicable to the case when $k>1$.
We set $r=10$ for \textbf{BB} and \textbf{Indep}
and $r=50$ for the other datasets.
The results of \textsc{Greedy}$^*$ for $k>1$ are only available
on \textbf{BB} and \textbf{Indep}. For the other datasets,
\textsc{Greedy}$^*$ fails to return any result within one day when $k>1$.
We can see all algorithms run much slower when $k$ increases.
For FD-RMS, lower efficiencies are caused by higher cost of maintaining top-$k$ results.
HS and $\varepsilon$-\textsc{Kernel} must consider all tuples in the datasets
instead of only skylines to validate that the maximum $k$-regret ratio is at most
$\varepsilon$ when $k>1$.
For \textsc{Greedy}$^*$, the number of linear programs to compute $k$-regret ratios
increases drastically with $k$.
Meanwhile, the maximum $k$-regret ratios drop with $k$,
which is obvious according to its definition.
FD-RMS achieves speedups of up to four orders of magnitude
than the baselines on all datasets.
At the same time, the solution quality of FD-RMS is also better
on all datasets except \textbf{Movie} and \textbf{CT},
where the results of HS are of slightly
higher quality in some cases.

\noindent\textbf{Scalability:}
Finally, we evaluate the scalability of different algorithms w.r.t.~the dimensionality $d$
and dataset size $n$.
To test the impact of $d$, we fix $n=100$K, $k=1$, $r=50$, and vary $d$ from $2$ to $10$.
The performance with varying $d$ is shown in Fig.~\ref{fig:nd:a}--\ref{fig:nd:b}.
Both the update time and maximum regret ratios of all algorithms increase dramatically with $d$.
Although almost all algorithms show good performance when $d=2,3$,
most of them quickly become very inefficient in high dimensions.
Nevertheless, FD-RMS has a significantly better scalability w.r.t.~$d$:
It achieves speedups of at least $100$ times over any other algorithm
while providing results of equivalent quality when $d \geq 7$.

To test the impact of $n$, we fix $d=6$, $k=1$, $r=50$, and vary $n$ from $100$K to $1$M.
The update time with varying $n$ is shown in Fig.~\ref{fig:nd:c}--\ref{fig:nd:d}.
For static algorithms, we observe different trends in efficiency on two datasets:
The update time slightly drops on \textbf{Indep} but keeps steady on \textbf{AntiCor}.
The efficiencies are determined by two factors, i.e.,
\emph{the number of tuples on the skyline} and
\emph{the frequency of skyline updates}.
As shown in Fig.~\ref{fig:stats}, when $n$ is larger, the number of tuples on the skyline
increases but the frequency of skyline updates decreases.
On \textbf{Indep}, the benefits of lower update frequencies outweigh the cost of more
skyline tuples; on \textbf{AntiCor}, two factors cancel each other.
FD-RMS runs slower when $n$ increases due to
higher cost of maintaining top-$k$ results on \textbf{Indep}.
But, on \textbf{AntiCor}, the update time keeps steady with $n$
because of smaller values of $\varepsilon$
and $m$, which cancel the higher cost of maintaining top-$k$ results.
We can observe that the maximum regret ratios of most algorithms
are not not significantly affected by $n$.
The solution quality of FD-RMS is close to the best of static algorithms.
Generally, FD-RMS always outperforms all baselines for different values of $n$.

%% file: section/02-related-work.tex
\section{Related Work}
\label{sec:related:work}

There have been extensive studies on the $k$-regret minimizing set ($k$-RMS) problem
(see~\cite{DBLP:journals/vldb/XieWL20} for a survey).
Nanongkai et al.~\cite{DBLP:journals/pvldb/NanongkaiSLLX10}
first introduced the notions of \emph{maximum regret ratio} and \emph{$r$-regret query}
(i.e., maximum $1$-regret ratio and $1$-RMS in this paper).
They proposed the \textsc{Cube} algorithm to provide an upper-bound guarantee
for the maximum regret ratio of the optimal solution of $1$-RMS.
They also proposed the \textsc{Greedy} heuristic for $1$-RMS, which always picked a tuple that
maximally reduced the maximum regret ratio at each iteration.
Peng and Wong~\cite{DBLP:conf/icde/PengW14} proposed the \textsc{GeoGreedy} algorithm to
improve the efficiency of \textsc{Greedy} by utilizing the geometric properties of $1$-RMS.
Asudeh et al.~\cite{DBLP:conf/sigmod/AsudehN0D17} proposed two discretized matrix min-max (DMM)
based algorithms for $1$-RMS. Xie et al.~\cite{DBLP:conf/sigmod/XieW0LL18}
designed the \textsc{Sphere} algorithm
for $1$-RMS based on the notion of $\varepsilon$-kernel~\cite{DBLP:journals/jacm/AgarwalHV04}.
Shetiya et al.~\cite{Shetiya2019} proposed a unified algorithm called URM
based on $k$-\textsc{Medoid} clustering for $l^p$-norm RMS
problems, of which $1$-RMS was a special case when $p=\infty$.
The aforementioned algorithms cannot be used for $k$-RMS when $k > 1$.
Chester et al.~\cite{DBLP:journals/pvldb/ChesterTVW14} first extended
the notion of $1$-RMS to $k$-RMS.
They also proposed a randomized \textsc{Greedy}$^*$ algorithm
that extended the \textsc{Greedy} heuristic to support $k$-RMS when $k > 1$.
The min-size version of $k$-RMS that returned the minimum subset whose maximum $k$-regret ratio
was at most $\varepsilon$ for a given $\varepsilon \in (0,1)$ was studied
in~\cite{DBLP:conf/wea/AgarwalKSS17,DBLP:conf/alenex/KumarS18}.
They proposed two algorithms for min-size $k$-RMS based on the notion of
$\varepsilon$-kernel~\cite{DBLP:journals/jacm/AgarwalHV04} and hitting-set,
respectively. However, all above algorithms are designed for the static setting and
very inefficient to process database updates.
To the best of our knowledge, FD-RMS is the first $k$-RMS algorithm that is optimized for
the fully dynamic setting and efficiently maintains the result for dynamic updates.

Different variations of regret minimizing set problems were also studied recently.
The $1$-RMS problem with nonlinear utility functions were studied
in~\cite{DBLP:journals/pvldb/FaulknerBL15,DBLP:journals/tods/QiZSY18,DBLP:conf/aaai/SomaY17a}.
Specifically, they generalized the class of utility functions to
\emph{convex functions}~\cite{DBLP:journals/pvldb/FaulknerBL15},
\emph{multiplicative functions}~\cite{DBLP:journals/tods/QiZSY18},
and \emph{submodular functions}~\cite{DBLP:conf/aaai/SomaY17a}, respectively.
Asudeh et al.~\cite{DBLP:conf/sigmod/AsudehN00J19} proposed the
rank-regret representative (RRR) problem.
The difference between RRR and RMS is that
the regret in RRR is defined by ranking while the regret in RMS is defined by score.
Several studies~\cite{DBLP:conf/icde/ZeighamiW19,DBLP:conf/aaai/StorandtF19,Shetiya2019}
investigated the \emph{average regret minimization} (ARM) problem.
Instead of minimizing the maximum regret ratio,
ARM returns a subset of $r$ tuples such that the average regret of all possible users is minimized.
The problem of \emph{interactive regret minimization} that aimed to enhance the regret minimization
problem with user interactions was studied
in~\cite{DBLP:conf/sigmod/NanongkaiLSM12,DBLP:conf/sigmod/XieWL19}.
Xie et al.~\cite{Xie2020} proposed a variation of min-size RMS
called $\alpha$-happiness query.
Since these variations have different formulations from the original $k$-RMS problem,
the algorithms proposed for them cannot be directly applied to the $k$-RMS problem.
Moreover, these algorithms are still proposed for the static setting
without considering database updates.

%% file: section/06-conclusion.tex
\section{Conclusion}\label{sec:conclusion}
In this paper, we studied the problem of maintaining $k$-regret minimizing sets
($k$-RMS) on dynamic datasets with arbitrary insertions and deletions of tuples.
We proposed the first fully-dynamic $k$-RMS algorithm called FD-RMS.
FD-RMS was based on transforming fully-dynamic $k$-RMS to a dynamic set cover problem,
and it could dynamically maintain the result of $k$-RMS with a theoretical guarantee.
Extensive experiments on real-world and synthetic datasets confirmed the efficiency,
effectiveness, and scalability of FD-RMS compared with existing static approaches to $k$-RMS.
For future work, it would be interesting to investigate whether our techniques
can be extended to $k$-RMS and related problems on higher dimensions (i.e.,~$d>10$)
or with nonlinear utility functions
(e.g.,~\cite{DBLP:journals/pvldb/FaulknerBL15,DBLP:journals/tods/QiZSY18,DBLP:conf/aaai/SomaY17a})
in dynamic settings.

%% file: section/appendix.tex
\appendices

\begin{table*}[t]
  \centering
  \caption{Frequently used notations}\label{tbl:notation}
  \begin{tabular}{|c|m{445pt}|}
  \hline
  \textbf{Symbol} & \textbf{Description} \\
  \hline
  $P_t$ & the database at time $t$ ($t \geq 0$) \\
  \hline
  $p$ & a tuple in database $P_t$ \\
  \hline
  $d$ & the dimensionality of $P_t$ \\
  \hline
  $n_t$ & the number of tuples in $P_t$ \\
  \hline
  $\Delta$ & a sequence $\langle \Delta_1, \Delta_2, \ldots \rangle$ of operations for database update \\
  \hline
  $\Delta_t$ & an operation $\langle p,+ \rangle$ or $\langle p,- \rangle$ at time $t$
               to update the database from $P_{t-1}$ to $P_t$ by adding/deleting tuple $p$ \\
  \hline
  $\mathbb{U}$ & the space of all nonnegative utility vectors \\
  \hline
  $u$ & a utility vector in $\mathbb{U}$ \\
  \hline
  $\varphi(u,P_t)$ & the top-ranked tuple in $P_t$ w.r.t.~$u$ \\
  \hline
  $\omega(u,P_t)$ & the score of $\varphi(u,P_t)$ w.r.t.~$u$ \\
  \hline
  $\varphi_j(u,P_t)$ & the $j$\textsuperscript{th}-ranked tuple in $P_t$ w.r.t.~$u$ \\
  \hline
  $\omega_j(u,P_t)$ & the score of $\varphi_j(u,P_t)$ w.r.t.~$u$ \\
  \hline
  $\Phi_k(u,P_t)$ & the set of top-$k$ results of $P_t$ w.r.t.~$u$, i.e.,
  $ \Phi_k(u,P_t) = \{ \varphi_j(u,P_t) : 1 \leq j \leq k \} $ \\
  \hline
  $\Phi_{k,\varepsilon}(u,P_t)$ & the set of $\varepsilon$-approximate top-$k$ results of $P_t$
  w.r.t.~$u$, i.e., $ \Phi_{k,\varepsilon}(u,P_t)
  = \{ p \in P_t : \langle u,p \rangle \geq (1-\varepsilon)\cdot\omega_{k}(u,P) \} $ \\
  \hline
  $\mathtt{rr}_k(u,Q)$ & the $k$-regret ratio of a subset $Q \subseteq P_t$ w.r.t.~$u$ \\
  \hline
  $\mathtt{mrr}_k(Q)$ & the maximum $k$-regret ratio of a subset $Q \subseteq P_t$ \\
  \hline
  $\mathtt{RMS}(k,r)$ & a $k$-RMS problem with size constraint
  $r \in \mathbb{Z}^+$ and $r \geq d$ \\
  \hline
  $Q^*_t$ & the optimal result of $\mathtt{RMS}(k,r)$ on $P_t$ \\
  \hline
  $\varepsilon^*_{k,r}$ & the (optimal) maximum $k$-regret ratio of $Q^*_t$ over $P_t$ \\
  \hline
  $Q_t$ & the result of $\mathtt{RMS}(k,r)$ on $P_t$ returned by FD-RMS \\
  \hline
  $\Sigma=(\mathcal{U},\mathcal{S})$ & a set system with the universe $\mathcal{U}$
  and the collection $\mathcal{S}$ of sets. When $\Sigma$ is built on the approximate
  top-$k$ results of $m$ utility vectors on $P_t$, we have $m=|\mathcal{U}|$
  and $n_t=|\mathcal{S}|$. \\
  \hline
  $\sigma$ & an operation $\langle u,S,\pm \rangle$ or $\langle u,\mathcal{U},\pm \rangle$
  to update the set system $\Sigma$ by adding/removing $u$ to/from $S \in \mathcal{S}$ or
  adding/removing $u$ to/from $\mathcal{U}$ \\
  \hline
  $S(p)$ & a set in $\mathcal{S}$ that contains all utility vectors in $\mathcal{U}$
  where $p$ is an approximate top-$k$ result on $P_t$ \\
  \hline
  $\mathcal{C}^*$ & the optimal solution for set cover on $\Sigma$ \\
  \hline
  $\mathcal{C}$ & an approximate solution for set cover on $\Sigma$ \\
  \hline
  $\mathtt{cov}(S)$ & the cover set of $S$ in $\mathcal{C}$ \\
  \hline
  $\varepsilon$ & an input parameter to specify the approximation factor
  of top-$k$ queries in FD-RMS \\
  \hline
  $M$ & an input parameter to provide the upper bound of the number $m$ of utility
  vectors used in FD-RMS \\
  \hline
  $\mathtt{u}(\Delta_t)$ & the number of utility vectors whose approximate top-$k$
  results are changed by $\Delta_t$ \\
  \hline
  \end{tabular}
\end{table*}

\section{Frequently Used Notations}
A list of frequently used notations in this paper is summarized in Table~\ref{tbl:notation}.

\section{Additional Experimental Evaluation}

\subsection{Extension of Greedy for Dynamic Updates}
In Section~\ref{sec:exp}, all algorithms (except FD-RMS and URM) we compare are static algorithms
that recompute the results each time when the skyline is updated.
In fact, it is possible to extend a static algorithm, e.g.,
\textsc{Greedy}~\cite{DBLP:journals/pvldb/NanongkaiSLLX10}, to
support dynamic updates as follows: Initially, it computes a $k$-regret minimizing set
of size $r'=r+\Delta_r$ (where $\Delta_r$ is an integer parameter) using \textsc{Greedy}.
Among this set, a set of $r$ tuples is reported as the result.
The result is recomputed whenever (1) $\Delta_r$ tuples from the $k$-regret minimizing set
of size $r'$ are deleted (in which case it recomputes a $k$-regret minimizing set of size $r'$);
(2) there are a total of $\Delta_s$ number of insertions into the skyline
where $\Delta_s$ is another parameter.
We refer to this extension of \textsc{Greedy} as the Dynamic Greedy (DG) algorithm.

In this subsection, we compare the Dynamic Greedy (DG) algorithm with \textsc{Greedy}
and FD-RMS. Specifically, we test four sets of parameters $\Delta_s,\Delta_r$, namely
DG-1 with $\Delta_s=10,\Delta_r=0.25r$, DG-2 with $\Delta_s=25,\Delta_r=0.5r$,
DG-3 with $\Delta_s=50,\Delta_r=0.75r$, and DG-4 with $\Delta_s=100,\Delta_r=r$.
The original \textsc{Greedy} is a special case of DG with $\Delta_s=1,\Delta_r=0$.
The experimental results are shown in Fig.~\ref{fig:dynamic:greedy}.
First of all, when $\Delta_s$ and $\Delta_r$ increase, DG runs faster because
the recomputation is executed at lower frequencies. But unfortunately,
DG still runs much slower than FD-RMS even when $\Delta_s=100,\Delta_r=r$
because of the huge gap in the efficiency of \textsc{Greedy} and FD-RMS.
Meanwhile, the maximum regret ratios of the results of DG increase
with $\Delta_s$ and $\Delta_r$, and become significantly higher than
those of FD-RMS when $\Delta_s$ and $\Delta_r$ are large. This is because
the results of DG could often be obsolete w.r.t.~up-to-date
datasets. In general, DG is significantly inferior to FD-RMS in terms of
both efficiency and quality of results for different parameters.
Nevertheless, DG might be used to accelerate \textsc{Greedy} in dynamic settings
if the requirement for quality is not so strict.

\subsection{Effect of Update Patterns}
In Section~\ref{sec:exp}, we consider the updates to datasets consist of all insertions
first followed by all deletions to test the performance of different algorithms
for tuple insertions and deletions separately. In this subsection,
we evaluate the performance of FD-RMS in two update
patterns: (1) First all insertions followed by all deletions as Section~\ref{sec:exp} (I+D);
(2) A random mixture of insertions and deletions (Mixed).
We note that the efficiency and solution quality of a static
algorithm are not influenced by update patterns
because the update rates of skylines are similar for tuple insertions
and deletions, and it reruns from scratch for each skyline update.
The update time of FD-RMS in two update patterns is shown
in Fig.~\ref{fig:time:mixed}. We can see that the update time is nearly indifferent
to update patterns. In addition, we also confirm that the results returned by FD-RMS
remain the same in almost all cases regardless of the orders of tuple operations.

In addition, we show the update time of FD-RMS with varying $r$ for tuple
insertions and deletions in Fig.~\ref{fig:id}. The update time for deletions
is slightly longer than that for insertions on all datasets. This
is mainly attributed to the difference in the update of approximate top-$k$ results.
For insertions, the dual-tree can incrementally update the current top-$k$
result using the utility index only; but for deletions, the dual-tree must retrieve
the new top-$k$ result from scratch on the tuple index once a top-$k$ tuple is deleted.

\subsection{Extension of URM for Dynamic Updates}\label{appendix:dynamic:urm}

\begin{algorithm}
  \caption{\textsc{Dynamic URM}}
  \label{alg:urm}
  \Input{Initial database $P_0$, set of operations $\Delta$}
  \Output{Result $Q_t$ for $1$-RMS on $P_t$ at time $t$}
  \tcc{compute an initial result $Q_0$ on $P_0$}
  run DMM-RRMS to get an initial result $S_0$\;
  run $k$-\textsc{medoid} from $S_0$ until convergence as $Q_0$\;
  \tcc{update $Q_{t-1}$ to $Q_t$ for $\Delta_t = \langle p,\pm \rangle$}
  Update $P_{t-1}$ to $P_t$ w.r.t.~$\Delta_t$\;
  \uIf{$\Delta_t= \langle p,+ \rangle$}{
    \uIf{there exists $p' \in Q_{t-1}$ such that the max regret ratio of $p$
         in the region $R(p')$ of $p'$ is lower than $p'$}
    {$S' \gets Q_{t-1} \setminus \{p'\} \cup \{p\}$\;
     run $k$-\textsc{medoid} from $S'$ until convergence as $Q_t$\;}
    \Else{ $Q_t \gets Q_{t-1}$\; }
  }
  \Else{\tcc{for a deletion $\Delta_t= \langle p,- \rangle$}
    \uIf{$p \in Q_{t-1}$}
    {$S' \gets Q_{t-1} \setminus \{p\} \cup \{p'\}$ where $p' \notin Q_{t-1}$ is top-ranked
    for at least one utility vector in the region $R(p)$ of $p$\;
     run $k$-\textsc{medoid} from $S'$ until convergence as $Q_t$\;}
    \Else{ $Q_t \gets Q_{t-1}$\; }
  }
\end{algorithm}

Unlike other static algorithms, the URM algorithm~\cite{Shetiya2019}
adopts an iterative framework based on the $k$-\textsc{medoid} clustering
for RMS computation. It is intuitive to extend the iterative framework to support dynamic updates.
The high-level idea of updating the result of URM for tuple operations
without fully recomputation is as follows: If an operation does not
affect the convergence of the current result, then skip it directly;
Otherwise, first adjust the result for this operation and then
run the iterative framework starting from the new result
until it converges again. The detailed procedures of (dynamic) URM
are presented in Algorithm~\ref{alg:urm}.

\begin{figure*}[t]
  \captionsetup[subfigure]{aboveskip=2pt}
  \centering
  \begin{subfigure}{.32\textwidth}
    \centering
    \includegraphics[width=\textwidth]{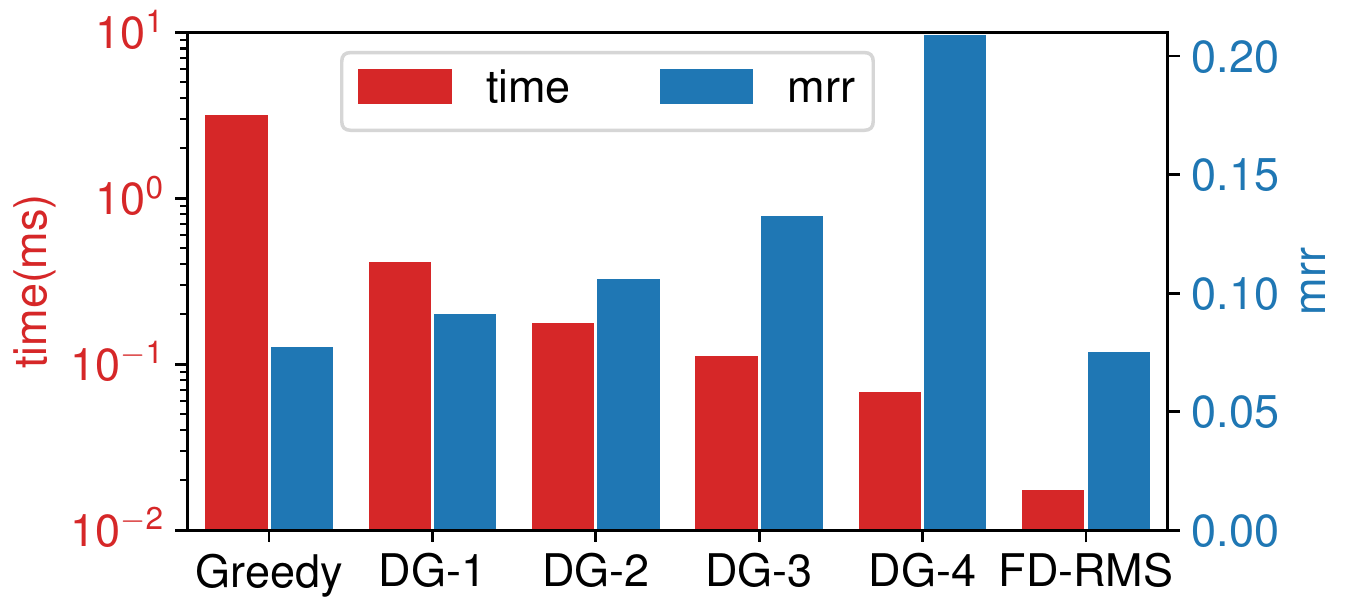}
    \caption{BB ($r=10$)}
  \end{subfigure}
  \begin{subfigure}{.32\textwidth}
    \centering
    \includegraphics[width=\textwidth]{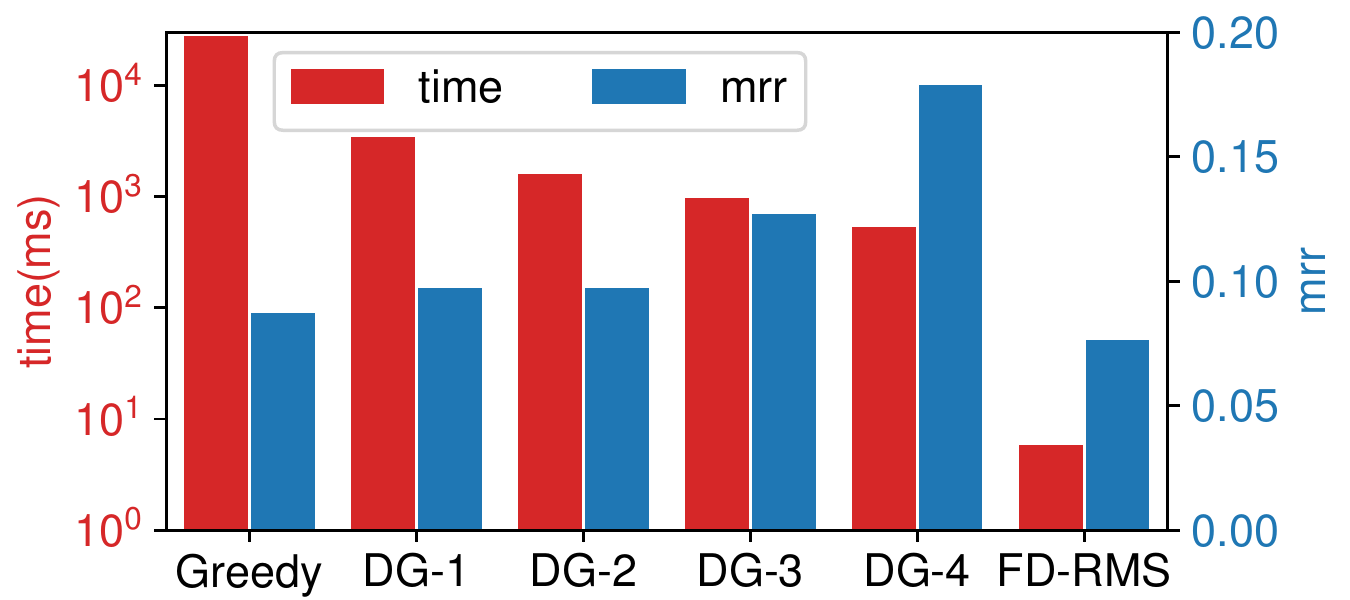}
    \caption{Movie ($r=50$)}
  \end{subfigure}
  \begin{subfigure}{.32\textwidth}
    \centering
    \includegraphics[width=\textwidth]{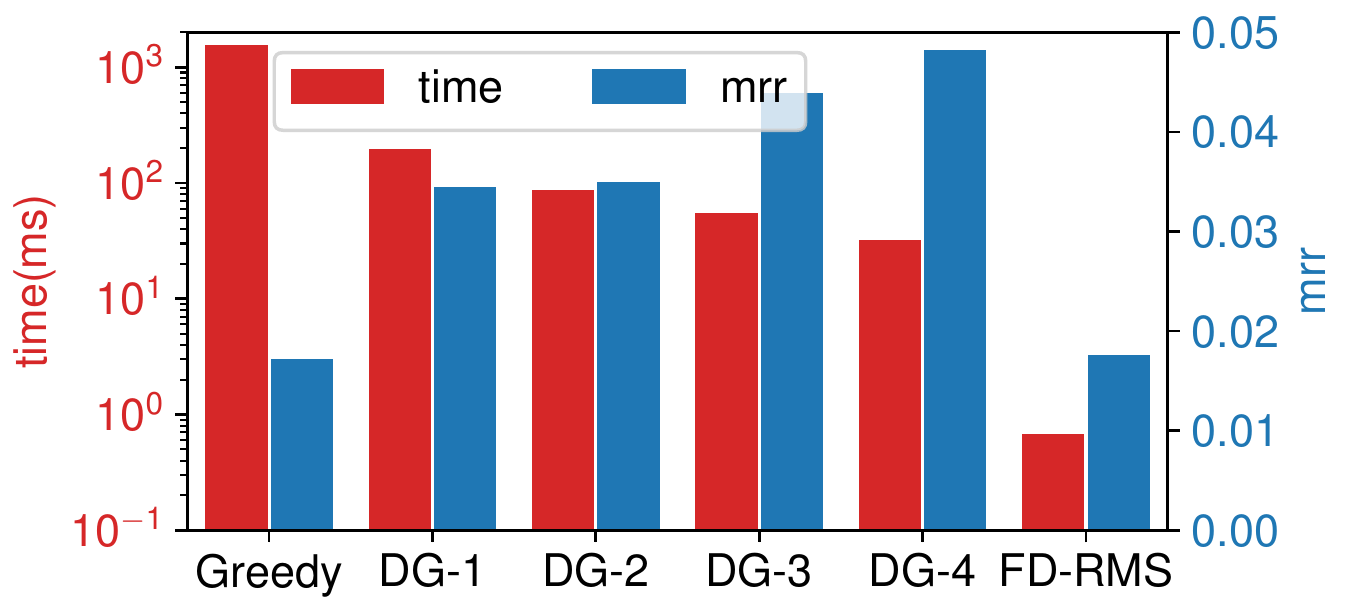}
    \caption{Indep ($d=6,r=50$)}
  \end{subfigure}
  \caption{Performance of the Dynamic Greedy (DG) algorithm in comparison with \textsc{Greedy}
  and FD-RMS ($k=1$). Parameter settings for DG are listed as follows:
  DG-1 -- $\Delta_s=10,\Delta_r=0.25r$; DG-2 -- $\Delta_s=25,\Delta_r=0.5r$;
  DG-3 -- $\Delta_s=50,\Delta_r=0.75r$; DG-4 -- $\Delta_s=100,\Delta_r=r$.}
  \label{fig:dynamic:greedy}
\end{figure*}

\begin{figure*}[t]
  \centering
  \begin{subfigure}{.18\textwidth}
    \centering
    \includegraphics[width=\textwidth]{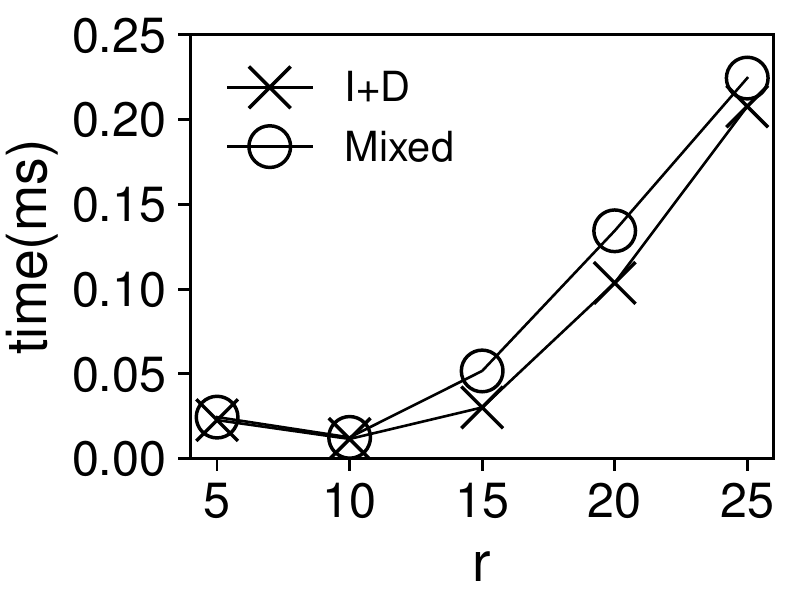}
    \caption{BB}
  \end{subfigure}
  \begin{subfigure}{.18\textwidth}
    \centering
    \includegraphics[width=\textwidth]{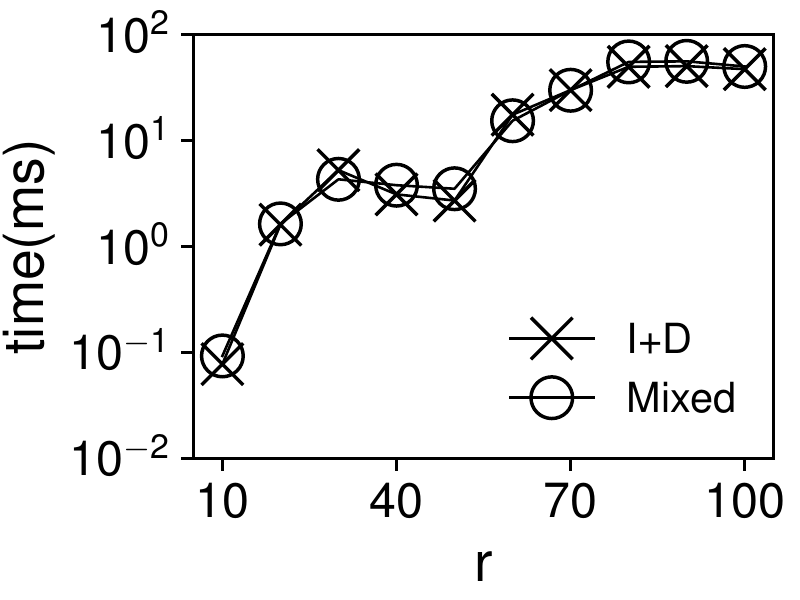}
    \caption{AQ}
  \end{subfigure}
  \begin{subfigure}{.18\textwidth}
    \centering
    \includegraphics[width=\textwidth]{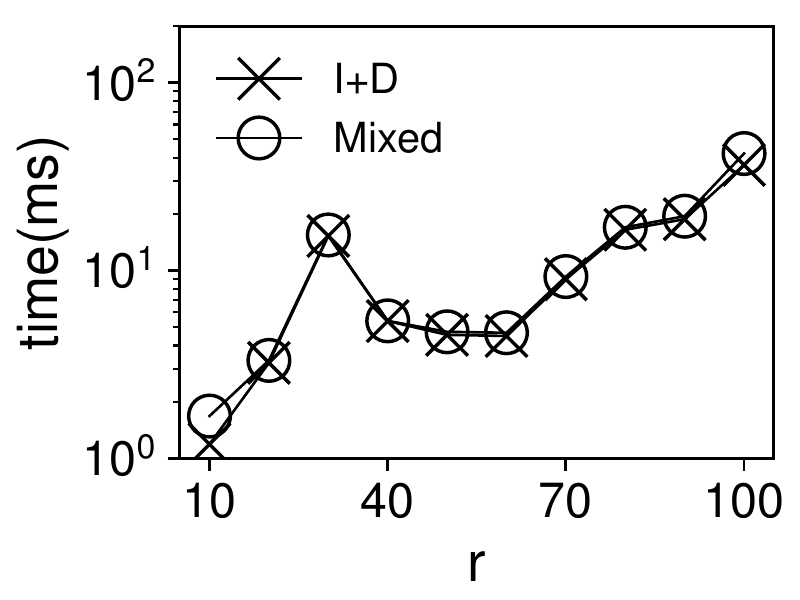}
    \caption{CT}
  \end{subfigure}
  \\
  \begin{subfigure}{.18\textwidth}
    \centering
    \includegraphics[width=\textwidth]{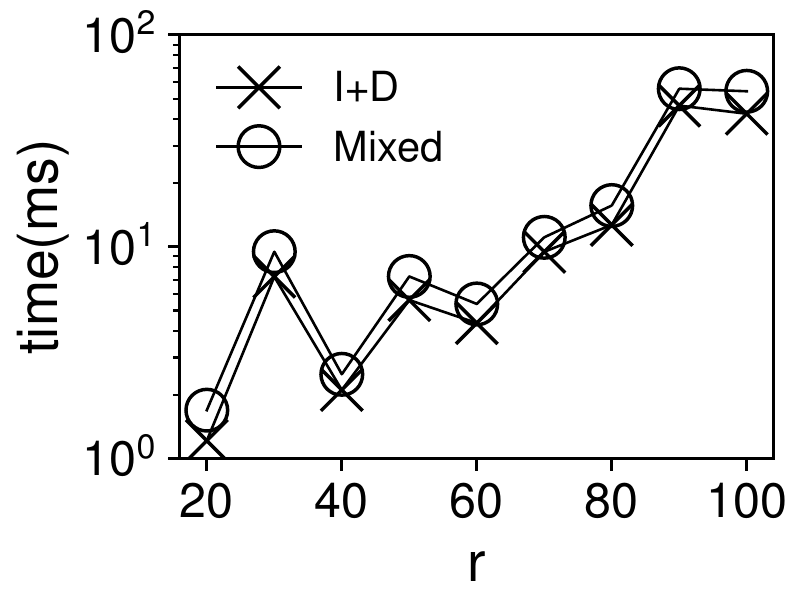}
    \caption{Movie}
  \end{subfigure}
  \begin{subfigure}{.18\textwidth}
    \centering
    \includegraphics[width=\textwidth]{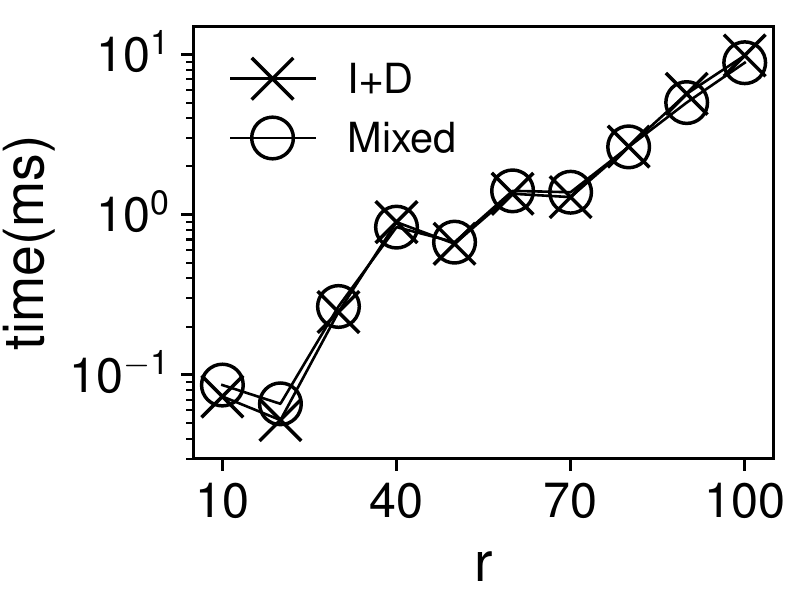}
    \caption{Indep ($d=6$)}
  \end{subfigure}
  \begin{subfigure}{.18\textwidth}
    \centering
    \includegraphics[width=\textwidth]{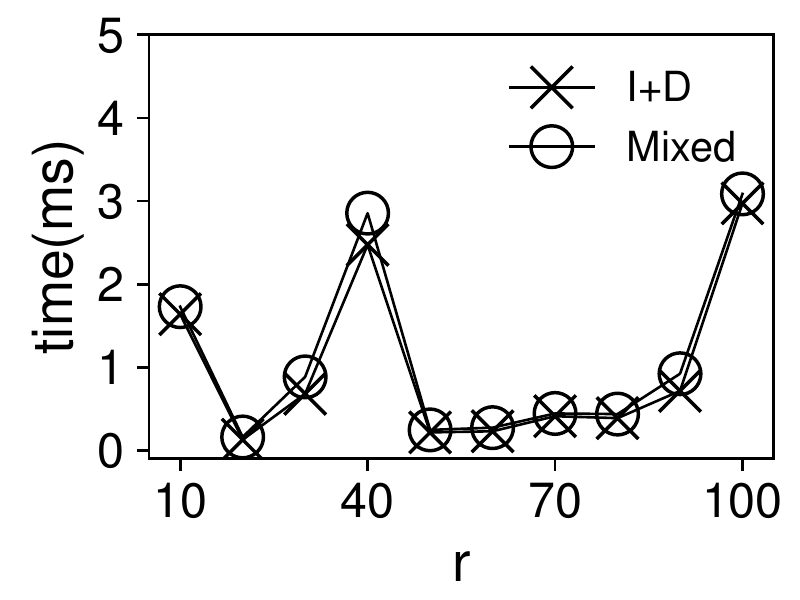}
    \caption{AntiCor ($d=6$)}
  \end{subfigure}
  \caption{Update time of FD-RMS in different update patterns ($k=1$)}
  \label{fig:time:mixed}
\end{figure*}

\begin{figure*}[t]
  \centering
  \begin{subfigure}{.18\textwidth}
    \centering
    \includegraphics[width=\textwidth]{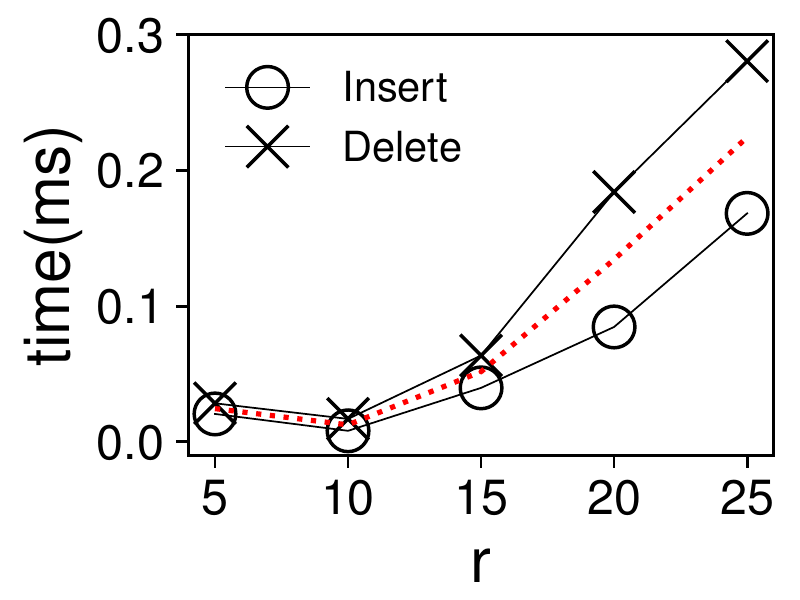}
    \caption{BB}
  \end{subfigure}
  \begin{subfigure}{.18\textwidth}
    \centering
    \includegraphics[width=\textwidth]{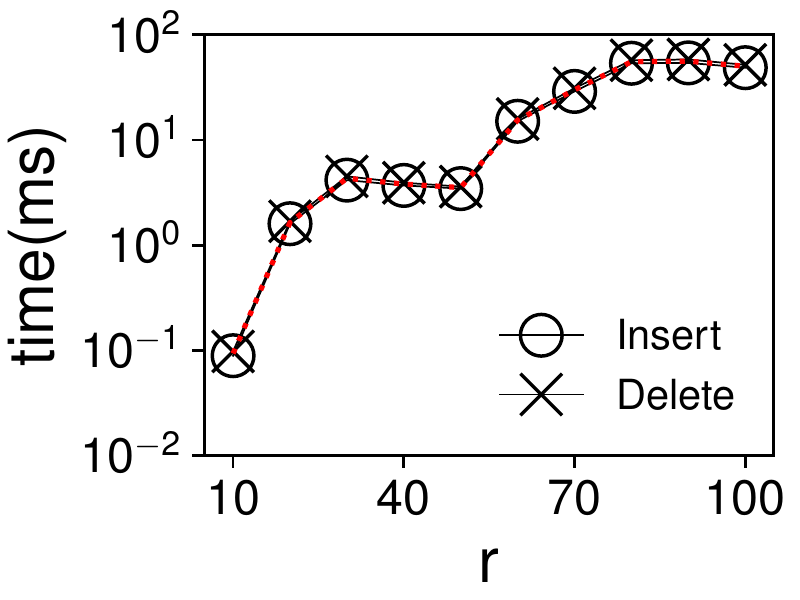}
    \caption{AQ}
  \end{subfigure}
  \begin{subfigure}{.18\textwidth}
    \centering
    \includegraphics[width=\textwidth]{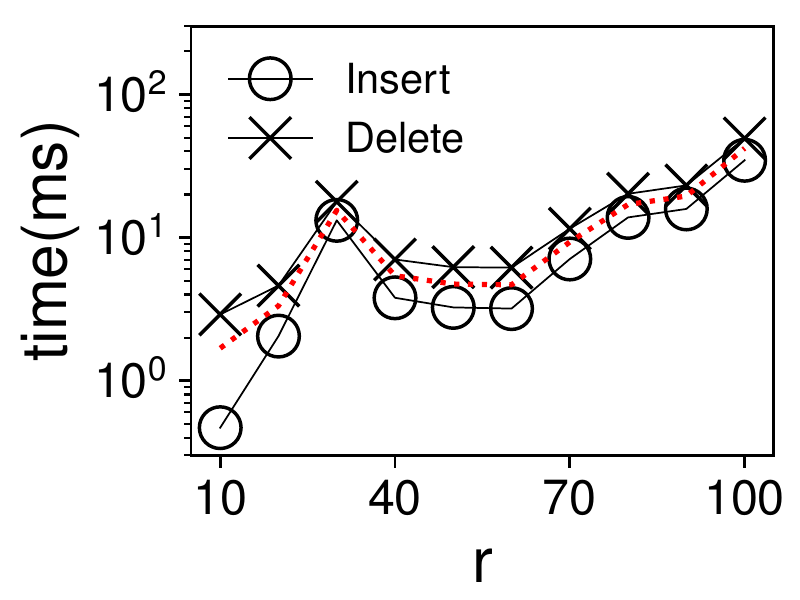}
    \caption{CT}
  \end{subfigure}
  \\
  \begin{subfigure}{.18\textwidth}
    \centering
    \includegraphics[width=\textwidth]{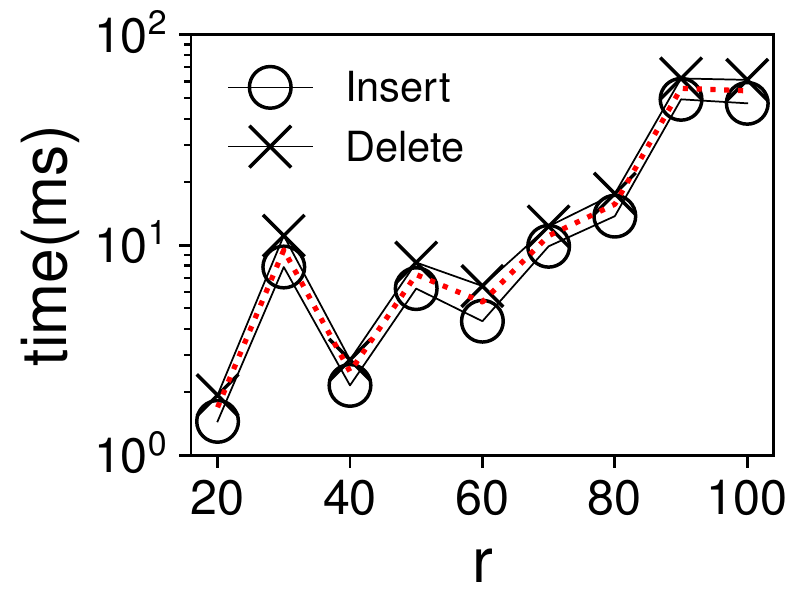}
    \caption{Movie}
  \end{subfigure}
  \begin{subfigure}{.18\textwidth}
    \centering
    \includegraphics[width=\textwidth]{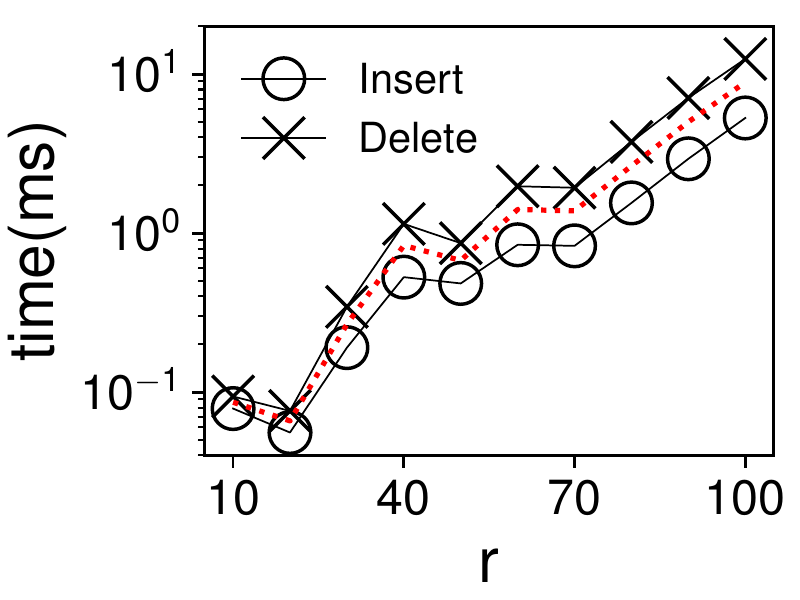}
    \caption{Indep}
  \end{subfigure}
  \begin{subfigure}{.18\textwidth}
    \centering
    \includegraphics[width=\textwidth]{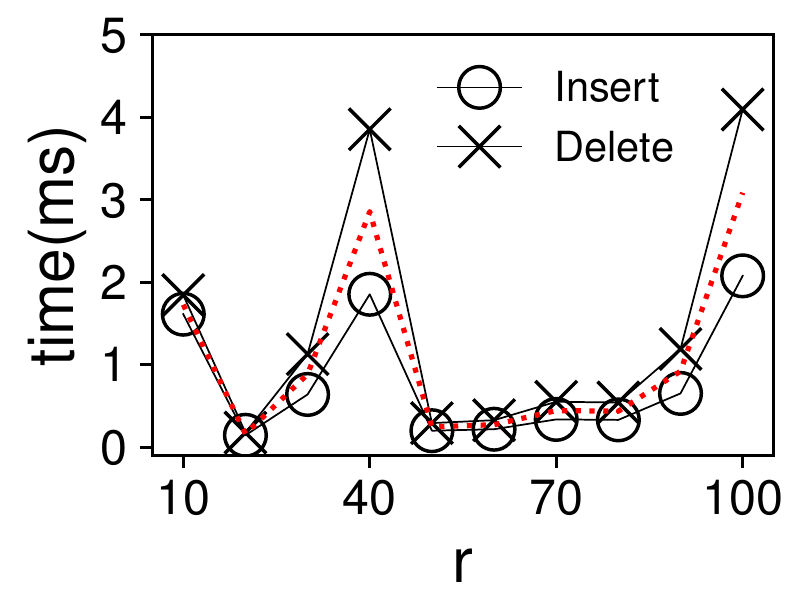}
    \caption{AntiCor}
  \end{subfigure}
  \caption{Update time of FD-RMS for tuple insertion and deletion ($k=1$).
           The red line denotes the average update time for all tuple operations.}
  \label{fig:id}
\end{figure*}